\documentclass[11pt]{article}

\usepackage[letterpaper, margin = 1.1in]{geometry}
\usepackage{amsmath,amsthm,amssymb,mathrsfs,bm,graphicx}
\usepackage[unicode]{hyperref}
\usepackage[comma,sort&compress,square]{natbib}
\usepackage{setspace}
\usepackage{subcaption}
\usepackage{wrapfig}
\usepackage{seqsplit}

\hypersetup{
    colorlinks=true,
    linkcolor=blue,
    citecolor=blue,
    filecolor=blue,
    urlcolor=blue
}

\numberwithin{equation}{section}
\allowdisplaybreaks
\newtheorem{theorem}{Theorem}[section]
\newtheorem{lemma}{Lemma}[section]
\newtheorem{proposition}{Proposition}[section]
\newtheorem{corollary}{Corollary}[section]

\theoremstyle{definition}

\newtheorem{remark}{Remark}[section]
\newtheorem{example}{Example}[section]

\newcommand{\E}{\mathbb{E}}
\renewcommand{\P}{\mathbb{P}}
\newcommand{\R}{\mathbb{R}}
\newcommand{\N}{\mathbb{N}}

\newcommand{\cA}{\mathcal{A}}
\newcommand{\cB}{\mathcal{B}}
\newcommand{\cD}{\mathcal{D}}
\newcommand{\cF}{\mathcal{F}}

\newcommand{\cH}{\mathcal{H}}
\newcommand{\cN}{\mathcal{N}}
\newcommand{\cP}{\mathcal{P}}
\newcommand{\cR}{\mathcal{R}}
\newcommand{\cS}{\mathcal{S}}
\newcommand{\cX}{\mathcal{X}}

\newcommand{\sZ}{\mathscr{Z}}

\newcommand{\bW}{\bm{W}}
\newcommand{\bZ}{\bm{Z}}
\newcommand{\bw}{\bm{w}}
\newcommand{\one}{\bm{1}}

\newcommand{\ra}{\rightarrow}
\newcommand{\vep}{\varepsilon}
\newcommand{\pto}{\overset{P}{\longrightarrow}}
\newcommand{\dto}{\overset{D}{\longrightarrow}}

\usepackage{tikz}

\newcommand{\myoverline}[2][3.5em]{%
  \overset{%
    \tikz[baseline=-0.1ex]{
      \draw[line width=0.75pt]
        (0,0) .. controls (#1/2,0.5ex) .. (#1,0);
    }%
  }{#2}%
}

\newcommand{\Var}{\mathrm{Var}}
\newcommand{\ECE}{\mathrm{ECE}}
\newcommand{\lECE}{\ell_2\textnormal{-}\ECE}
\newcommand{\binECE}{\mathrm{binECE}}
\newcommand{\lbinECE}{\ell_2\textnormal{-}\binECE}
\newcommand{\binhatECE}{\myoverline{\lbinECE}}
\newcommand{\rankECE}{\mathrm{rankECE}}
\newcommand{\hatrankECE}{\myoverline{\rankECE}}

\newcommand{\rmd}{\mathrm{d}}

\title{A Ranking Approach for Measuring Calibration}

\author{
Anirban Chatterjee\thanks{\texttt{achatter@bu.edu}}\\
Department of Mathematics and Statistics\\ Boston University
\and
Rina Foygel Barber\thanks{\texttt{rina@uchicago.edu}}\\
Department of Statistics\\ University of Chicago
}

\date{}

\begin{document}

\maketitle

\begin{abstract}
When providing forecasted probabilities with a predictive model, the ideal model offers perfect calibration: the true probability of the outcome (i.e., the probability that $Y=1$) exactly matches the forecasted probability $f(X)$. In practice, models inevitably exhibit calibration error, and it is therefore important to be able to measure this miscalibration to assess a model's reliability. The Expected Calibration Error (ECE) is the most widely used measure of miscalibration, but is known to be impossible to estimate the ECE with guaranteed accuracy in an assumption-free setting. In this work, we propose an alternative measure, the rankECE, that is based on comparing points with neighboring values of the predicted probability $f(X)$. Our theoretical guarantees and empirical results establish that rankECE provides a better proxy for ECE as compared to binned approximations to ECE, which are the most commonly-used approximations in practice.
\end{abstract}

\section{Introduction}\label{sec:intro}

In recent years, machine learning and artificial intelligence have achieved remarkable predictive performance across a wide range of domains, leading to their widespread adoption in safety-critical applications such as medical diagnosis \citep{esteva2017dermatologist,esteva2019guide,gulshan2016development}, financial risk assessment \citep{baesens2003benchmarking,fritz2022financial}, and autonomous driving \citep{caesar2020nuscenes,sun2020scalability,chen2024end}. In many of these applications, models produce probabilistic forecasts for classification tasks, which are subsequently used to guide downstream decisions \citep{turay2022toward,cai2020review,mahbobi2023credit}. The reliability of these decisions, however, depends critically on whether the predicted probabilities are well calibrated \citep{murphy1977reliability,vanCalster2019CalibrationTA,shen2025algorithms}. Informally, a model is calibrated if its predicted probabilities agree with the corresponding empirical frequencies \citep{foster1998asymptotic,dawid1982well}. For example, consider a medical diagnostic system that predicts a patient has a $70\%$ probability of having a particular disease based on their symptoms and medical history. A well-calibrated model ensures that, among all patients assigned a predicted probability of $70\%$, approximately $70\%$ actually have the disease. Consequently, calibrated probabilities provide meaningful measures of uncertainty that can be directly interpreted and incorporated into decision-making. In contrast, poorly calibrated predictions can misrepresent uncertainty, potentially leading to suboptimal or even harmful decisions in high-stakes settings.

Formally, for a binary prediction task with input features $X \in \mathcal{X}$ and a model $f:\mathcal{X} \rightarrow [0,1]$ predicting the probability of label $Y = 1$, the model $f$ is said to be calibrated if
\begin{align}\label{eq:def_calibration}
\mathbb{E}\left[Y \mid f(X)\right] = f(X)\text{ almost surely.}
\end{align}

Since its introduction in early work on meteorological forecasting \citep{glenn1950verification,degroot1983comparison}, calibration has been recognized as a fundamental property of probabilistic predictions \citep{hilden1978measurement,guo2017calibration,gneiting2014probabilistic,murphy1984probability,minderer2021revisiting}. However, fitted models are often poorly calibrated, due to issues such as model misspecification (such as for classical models that may not fit the more complex trends in the data) or due to overfitting (such as in neural networks, which often exhibit miscalibration \citep{guo2017calibration,ovadia2019,havasi2021training}). This motivates the need for principled measures of calibration. The Expected Calibration Error (ECE) is perhaps the most widely used such measure. As we discuss below, however, reliably estimating ECE from finite samples is challenging, and existing approximations, primarily binning-based estimators, may fail to capture failures of the calibration condition in \eqref{eq:def_calibration}. In this work, we propose a rank-based approximation to ECE that characterizes calibration, provides a closer approximation to ECE than existing binning-based methods, and can be reliably estimated.

\subsection{Background and Related Works}\label{sec:background}
To evaluate calibration performance, the most common approach is to measure the average discrepancy between the two sides of the calibration identity in \eqref{eq:def_calibration}. In particular, one of the most widely used measure is the $\ell_2$-Expected Calibration Error $(\lECE)$ \citep{murphy1973new,degroot1983comparison,hendrycks2019using,kumar2019verified,lee2023t,foster1998asymptotic, sun2024confidence}, defined as
\begin{align}\label{eq:L2_ECE}
\lECE(f) = \E\bigg[\left|\E[Y|f(X)] - f(X)\right|^2\bigg].
\end{align}
More generally, one can define $\ell_p$-ECE variants by replacing the $2$ in the exponent with $p$ for any $p \geq 1$. In practice, the $\ell_1$ and $\ell_\infty$ variants are also widely used \citep{guo2017calibration,naeini2015obtaining,nixon2019measuring,roelofs2022mitigating}; for simplicity, however, we focus on the $\ell_2$ variant from \eqref{eq:L2_ECE} throughout. 

While this measure is both widely used and intuitive, it presents several practical challenges. In general, the conditional expectation appearing in~\eqref{eq:L2_ECE} cannot be reliably estimated without placing assumptions on the distribution (e.g., smoothness assumptions), and consequently reliably estimating $\lECE(f)$ in finite samples to assess whether a model $f$ is approximately calibrated remains a significant practical challenge, despite being central desiderata of calibration measures \citep{lee2023t,ciosek2026measuring,rossellini2025can,brocker2022uniform,gupta2020distribution,angelopoulos2024theoretical}.  In other words, any nonasymptotic guarantee for accurately estimating $\lECE(f)$ must inherently rely on assuming certain properties of $f$; it is impossible to construct a uniformly consistent estimator of $\lECE$ (see Remark~\ref{rmk:hardness_result} for a precise formulation of this hardness result).

\subsubsection{Binned approximations to ECE} 
Due to the difficulty of estimating ECE, in practice a common approach is to consider a binned approximation of $\lECE$ \citep{guo2017calibration,naeini2015obtaining}, where the calibration of a predictor $f$ is estimated within each bin $\cB_j$ of a partition $[0,1] = \cB_1 \cup \dots \cup \cB_j$:
\begin{align}\label{eq:L2_binned_ECE}
\ell_2\text{-binECE}(f) = \sum_{j=1}^{K} \P\left(f(X)\in \cB_j\right) \cdot  \E\left[Y-f(X)\mid f(X)\in \cB_j\right]^2.
\end{align}
This quantity can equivalently be interpreted as the $\lECE$ of the piecewise-constant approximation of $f$ formed by averaging within each bin: that is, $\ell_2\text{-binECE}(f) = \lECE(f_{\mathrm{bin}})$, where
\begin{equation}\label{eqn:fbin}
    f_{\mathrm{bin}}(x) = \E[f(X) \mid f(X)\in \cB_j]\textnormal{ for any $x$ with $f(x)\in \cB_j$}.
\end{equation}Typically, in finite samples, $\lbinECE(f)$ is estimated by comparing empirical mean of predicted probabilities with observed label frequencies within each bin on a held-out test set (see \eqref{eq:binhatECE}). However, this binning-based approximation may fail to capture certain forms of miscalibration (see Example 3.2 in \cite{kumar2019verified}). Moreover, the choice of the number of bins is critical, and such methods are known to be highly sensitive to implementation choices and to exhibit inherent bias \citep{nixon2019measuring,roelofs2022mitigating,futami2024information,tao2024benchmark}, often necessitating additional debiasing procedures for reliable inference \citep{sun2024confidence,lee2023t,kumar2019verified}. 

Beyond binning-based approaches, plug-in estimators using nonparametric estimates of the conditional expectation function, such as kernel smoothing \citep{zhang2020mix,blasiok2023smooth} and local polynomial regression \citep{austin2019integrated,cleveland2017local}, have also been explored. However, as mentioned above, it is impossible to achieve nonasymptotic assumption-free consistent estimation of $\lECE$ with any estimator.

\subsubsection{Alternatives to ECE}

ECE (and its variants) is perhaps the most widely used measure of predictive calibration. However, due to the challenges of estimating ECE from finite samples, several alternative notions have been proposed to characterize the calibration condition in \eqref{eq:def_calibration}. Kernel-based measures, including Maximum Mean Calibration Error \citep{kumar2018trainable} and Kernel Calibration Error \citep{widmann2019calibration,widmann2021calibration}, quantify calibration via RKHS embeddings of the calibration residual. \citet{gopalan2022low} introduced weighted calibration error (wCE), a general framework parameterized by a class of test functions that subsumes $\ell_1$-ECE and kernel-based measures. \citet{rossellini2025can} proposed cutoff calibration by restricting test functions to interval indicators. \citet{blasiok2023unifying} introduced distance from calibration (dCE), which measures the distance to the nearest perfectly calibrated predictor and is equivalent to wCE with 1-Lipschitz test functions; this notion is also closely related to weak calibration studied in \citet{kakade2008deterministic}.

\subsection{Our Contributions}

Motivated by the challenges of estimating $\lECE$ and the limitations of binning-based approaches, we introduce $\rankECE$ as a tuning-free, rank-based measure of calibration. Our theoretical results show that $\rankECE$ provides a more reliable approximation to $\lECE$, as compared to the binning-based approach. Moreover, $\rankECE$ offers the favorable property that it is zero if and only if $f$ is perfectly calibrated; consequently, we develop finite-sample and asymptotically valid tests that are consistent for detecting departures from perfect calibration. 
Our experiments demonstrate the effectiveness of the proposed measure. Across a wide range of settings, $\rankECE(f)$ consistently provides a closer approximation to $\lECE(f)$ than the widely used $\lbinECE(f)$. We further validate this finding on real-world sentiment prediction tasks using the Amazon and Yelp Polarity corpora. Finally, we show that calibration tests based on $\rankECE(f)$ perform competitively with existing SKCE-based tests \citep{widmann2019calibration}.

The rest of the paper is organized as follows. Section~\ref{sec:def_rank_ece} introduces the proposed rank-based measure of calibration. Section~\ref{sec:rank_prop} studies its theoretical properties as approximation to $\lECE$ and its use for calibration testing. Section~\ref{sec:compare} compares it with existing binning-based approaches. Section~\ref{sec:experiments} presents numerical experiments on both synthetic and real data. All proofs are presented in the appendix, together with some additional results and extensions.

\section{A Rank Measure for Calibration}\label{sec:def_rank_ece}
In this section, we formally introduce our proposed rank-based measure of calibration. We begin by establishing the necessary notation and assumptions. Let $\cX$ denote the feature space (assumed to be a standard Borel space), with $X\in\cX$ representing the covariates/features, and let $Y\in\{0,1\}$ denote the binary label. Consider a pre-trained predictor $f:\cX\ra[0,1]$ for the conditional probability $\P(Y=1\mid X)$. 
For notational convenience, define $Z = f(X)$. Throughout this work, for convenience we will implicitly assume that $Z$ is nonatomic (i.e., $\P(Z=t)=0$ for all $t\in[0,1]$), to avoid the possibility of ties. In particular, this assumption is needed in our definition of $\rankECE$ as well as in the theoretical results. In Appendix \ref{appendix:atomic} we will extend the definition and generalize our main results to remove this assumption.

Suppose that we observe samples $(Y_1,Z_1),\ldots,(Y_n,Z_n)$, where $Z_i=f(X_i)$ for all $1\leq i\leq n$. We define the sample rank-based approximation to $\lECE(f)$ as 
\begin{align}\label{eq:rank_ece_sample}
\hatrankECE_n(f)=\frac{1}{n}\sum_{i=1}^{n-1}\left(Y_{\pi(i)} - Z_{\pi(i)}\right)\left(Y_{\pi(i+1)} - Z_{\pi(i+1)}\right),
\end{align}
where $\pi:[n]\ra[n]$ is the unique permutation such that $Z_{\pi(1)}\leq Z_{\pi(2)}\leq \cdots \leq Z_{\pi(n)}$. The motivation for this definition follows directly from the following identity:
\begin{align*}
\lECE(f)=\E\left[\left|\E[Y\mid Z]-Z\right|^2\right]= \E\left[(Y-Z)(Y'-Z)\right],
\end{align*}
where $Y$ and $Y'$ are independent draws from the conditional distribution of $Y\mid Z$. To construct a sample analogue, we treat consecutive order statistics as having approximately equal predicted probabilities. That is, since $Z_{\pi(i)}$ and $Z_{\pi(i+1)}$ (which are neighbors in the sorted list of observed $Z$ values) are approximately equal, we interpret their associated labels $Y_{\pi(i)}$ and $Y_{\pi(i+1)}$ as approximately i.i.d.\ draws from the conditional distribution $Y\mid Z$, yielding the quantity in \eqref{eq:rank_ece_sample}.

The quantity defined in \eqref{eq:rank_ece_sample} is a sample-based estimate. Taking expectations on both sides of \eqref{eq:rank_ece_sample} we can now define the population analogue as
\begin{align}\label{eq:rank_ECE}
\rankECE_n(f)=\E\left[\hatrankECE_n(f)\right]=\E\left[\frac{1}{n}\sum_{i=1}^{n-1}\left(Y_{\pi(i)} - Z_{\pi(i)}\right)\left(Y_{\pi(i+1)} - Z_{\pi(i+1)}\right)
\right].
\end{align}
For notational brevity, from this point on we will generally omit the dependence on $n$ in $\rankECE(f)$ and $\hatrankECE(f)$, and will indicate the dependence on $n$ only when necessary.

Our next result (proved in Appendix~\ref{sec:proofof_prop_test_rank_finite}) establishes the concentration of $\hatrankECE(f)$ around its mean $\rankECE(f)$.
\begin{proposition}\label{prop:test_rank_finite}
    Take $n\geq 2$. For a predictor $f:\cX\ra[0,1]$, recall $\hatrankECE(f)$ and $\rankECE(f)$ from \eqref{eq:rank_ece_sample} and \eqref{eq:rank_ECE}, respectively. Then for any $\delta>0$, with probability at least $1-\delta$,
    \begin{align*}
        \left|\hatrankECE(f) - \rankECE(f)\right|\leq \sqrt{\frac{81\log (2/\delta)}{32n}}
    \end{align*}
\end{proposition}
\noindent In particular, this result places no assumptions on $f$, and thus establishes that $\rankECE$ (which is a population quantity) can be estimated consistently in an assumption-free regime (via the sample estimator $\hatrankECE$).

\section{Properties of $\rankECE$}\label{sec:rank_prop}
Having defined our proposed measure of calibration $\rankECE$ (and established the assumption-free accuracy of its sample-based estimator $\hatrankECE$), we now return to our motivating question: can this new measure offer a better approximation to $\lECE$, as compared to existing alternatives?

In this section, we examine the properties of $\rankECE$, focusing on comparing it to $\lECE$. Afterwards (in Section~\ref{sec:compare}), we will then compare to the binning-based approximation, $\lbinECE$, to establish the benefits of our proposal.

\subsection{$\rankECE$ as an approximation to $\lECE$}
Given the challenges in estimating $\lECE(f)$ discussed in the introduction, our first aim is to explore the extent to which $\rankECE$ provides a reliable approximation to $\lECE$.

We begin by showing that miscalibration captured by $\rankECE$ implies miscalibration in terms of $\lECE$. The proof of this proposition is given in Appendix \ref{sec:proofof_prop_lower_bound}. 
\begin{proposition}\label{prop:lower_bound}
For all $n\ge 2$ and any predictor $f:\cX\ra[0,1]$,
\begin{align*}
    0 \leq \rankECE(f) \leq \lECE(f),
\end{align*}
where $\rankECE(f)$ and $\lECE(f)$ are defined in \eqref{eq:rank_ECE} and \eqref{eq:L2_ECE}, respectively.
\end{proposition}
\noindent At finite $n$, Proposition \ref{prop:lower_bound} establishes that $\rankECE(f)$ is a lower bound on $\lECE(f)$, but does not quantify the gap between the two. In other words, $\rankECE$ cannot systematically overestimate the $\lECE$, but we now need to examine whether $\rankECE$ might result in a drastic underestimate. To do so, we first consider an asymptotic result, showing that $\rankECE(f)$ offers consistent estimation of $\lECE(f)$ for any \emph{fixed} predictor $f$. 

\begin{proposition}\label{prop:consistency}
Fix a predictor $f:\cX\ra[0,1]$. The empirical estimate $\hatrankECE_n(f)$ converges in probability to $\lECE(f)$ as $n\ra\infty$. Consequently, $\rankECE_n(f)\to\lECE(f)$ as $n\to\infty$.
\end{proposition}
\noindent This asymptotic result holds for each fixed $f$, and does not provide any uniform control of the gap; indeed, the hardness results for estimating $\lECE$ (discussed in Section~\ref{sec:background}) imply that it would not be possible to provide an assumption-free (i.e., uniform over all $f$) bound on the gap between $\rankECE(f)$ and $\lECE(f)$. However,
the following result shows that this gap is typically small, except in pathological cases where the residual function $r(Z)=\E[Y-Z\mid Z]$ is highly irregular. To formalize this notion, we introduce the total variation of a function $r$:
\begin{align}\label{eq:def_TV}
    \|r\|_{\textnormal{TV}} = \sup\left\{\sum_{i=0}^{k-1}\left|r(s_{i+1}) - r(s_i)\right|:0 = s_0<s_1<\cdots<s_{n_P-1}<s_k = 1,~ k\geq 2\right\}.
\end{align}
The following result quantifies the approximation error of $\rankECE(f)$ relative to $\lECE(f)$, in terms of the total variation of the residual function $r(Z)$ (see Appendix \ref{sec:proofof_prop_rank_l2_finite_bd} for the proof).
\begin{proposition}\label{prop:rank_l2_finite_bd}
    Consider a predictor $f:\cX\ra[0,1]$, and let $r(Z) = \E[Y-Z\mid Z]$ where $Z=f(X)$. Then, for $n\geq 2$,
    \begin{align*}
        \left|\rankECE(f) - \lECE(f)\right|\leq \frac{1}{n}\left(\|r\|_{\textnormal{TV}} + 1\right)
    \end{align*}
\end{proposition}

Proposition \ref{prop:rank_l2_finite_bd} shows that $\rankECE(f)$ provides a tight approximation to $\lECE$ whenever the residual function has bounded variation. This assumption is mild in practice: prediction functions learned by standard classifiers are typically smooth or monotone, and the corresponding residual functions therefore have bounded variation \citep{ciosek2026measuring}; consequently, the approximation error decays at $O(1/n)$ rate in most practical settings. In Appendix \ref{appendix:example}, we complement this result with an example showing that the bounded variation assumption is essential. Specifically, we construct a sequence of residual functions whose total variation increases with $n$, for which $\rankECE(f)$ can be arbitrarily close to zero while $\lECE(f)$ remains bounded away from zero.

The results thus far establish $\rankECE(f)$ as a reasonably reliable approximation of $\seqsplit{\lECE(f)}$; this is complemented by the result of Proposition \ref{prop:test_rank_finite}, where we prove an assumption-free guarantee on the accuracy of its estimator $\hatrankECE(f)$. In other words, $\hatrankECE(f)$ is a sample-based, tuning-free estimator that cannot substantially overestimate $\lECE(f)$, and also will not substantially underestimate $\lECE(f)$ under a mild total variation condition. This is essentially the best we could hope for, given the hardness results preventing assumption-free consistent estimation of $\lECE(f)$ described in Section~\ref{sec:background} (see also Remark~\ref{rmk:hardness_result} below).
\begin{remark}
    The property of $\rankECE(f)$ established in Proposition \ref{prop:test_rank_finite} has been referred to in the literature as \textit{testability} \citep{rossellini2025can,okoroafor2025nearoptimal}. Informally, this means that one can empirically verify whether a predictor $f$ is approximately calibrated with respect to this measure (which holds for $\rankECE(f)$, via its estimator $\hatrankECE(f)$, whose accuracy is established in Proposition \ref{prop:test_rank_finite} with no assumptions on $f$). When a calibration measure is testable, it can be used to ensure that a predictor is suitable for downstream decision-making. More broadly, two properties are especially desirable for any calibration measure \citep{rossellini2025can}. The first is \textit{testability}, as described above. The second is \textit{actionability}, which requires that small calibration error under the given measure implies meaningful guarantees for downstream decision-making tasks. We study these properties for $\rankECE(f)$ in more detail in Appendix \ref{sec:test_action_rankece}.
\end{remark}

\begin{remark}\label{rmk:hardness_result}
    Above, we have mentioned past results on the fundamental hardness of estimating $\lECE$ in an assumption-free setting; these results explain why finite-sample statements (such as Proposition~\ref{prop:rank_l2_finite_bd}) cannot avoid assumptions (such as bounded total variation). Here we give a precise statement of such a hardness result, to clarify the challenge. The following result is adapted from \citet[Theorem 12.5]{angelopoulos2024theoretical} (see Appendix~\ref{app:hardness_result} for details). Let $\hat{E}_n(f)$ be any function of data $(X_1,Y_1),\dots,(X_n,Y_n)$ (intended to estimate $\lECE(f)$ for some pretrained predictor $f$). Assume the range of $f$ contains some positive-width interval, $\{f(x):x\in\cX\} \supseteq (c-\epsilon,c+\epsilon)$ for some $c\in(0,1),\epsilon>0$. Then 
    \[{\textstyle\sup_P} \,\left| \E[\hat{E}_n(f)] - \lECE(f)\right|\geq \frac{c(1-c)}{2},\]
    where the supremum is taken over all distributions $P$ on $(X,Y)\in\cX\times\{0,1\}$ (i.e., the expected value $\E[\hat{E}_n(f)]$ and the miscalibration value $\lECE(f)$ are computed with respect to $P$). Note that the right-hand side does not depend on the sample size $n$, and therefore does not vanish as $n\to\infty$, meaning that there is no possibility of a uniformly consistent estimator for $\lECE$.
\end{remark}

\subsection{Testing for Perfect Calibration}\label{sec:perfect_calib}
In the previous sections, we motivated $\rankECE(f)$ as a meaningful measure of calibration for predictive models in binary classification. Given this interpretation, a natural question is whether $\rankECE(f)$ can also be used to test whether a predictive model is perfectly calibrated. The calibration hypothesis requires that \eqref{eq:def_calibration} holds almost surely. Equivalently, we consider the hypothesis testing problem
\begin{align}\label{eq:null_calibration}
    \bm{H}_0: \P\Big(\E[Y\mid f(X)] = f(X)\Big) = 1
    \quad\text{vs.}\quad
    \bm{H}_1: \P\Big(\E[Y\mid f(X)] = f(X)\Big) < 1.
\end{align}
There has been a growing body of recent work on testing this hypothesis using a variety of notions of calibration and the corresponding calibration measures \citep{widmann2019calibration,glaser2023fast,gweon2023power,feng2024model,chatterjee2024kernel,lee2023t}. To leverage $\rankECE(f)$ for the purpose of testing for calibration, in the following result we first show that $\rankECE(f)$ provides an exact characterization of calibration.

\begin{theorem}\label{thm:null_calibration}
Let $n \geq 4$. Then, for any predictor $f:\cX \to [0,1]$, we have $\rankECE(f) = 0$ if and only if $f$ satisfies \eqref{eq:def_calibration}.
\end{theorem}

We refer the reader to Appendix \ref{sec:proofof_thm_null_calibration} for a proof of Theorem \ref{thm:null_calibration}. The property of $\rankECE(f)$ established in Theorem \ref{thm:null_calibration} mirrors the exact characterization of calibration provided by $\seqsplit{\lECE(f)}$ in \eqref{eq:def_calibration}. In contrast, $\lbinECE(f)$ does not satisfy this property: Example 3.2 of \citet{kumar2019verified} demonstrates that $\lbinECE(f)$ can be zero even when $f$ is miscalibrated. Such a characterization is a desirable property for any calibration measure, as it ensures that small values of the measure provide meaningful guarantees that $f$ satisfies the calibration condition in \eqref{eq:def_calibration}. 

In the remainder of this section, we develop two tests for the calibration hypothesis in \eqref{eq:null_calibration} based on the test statistic $\hatrankECE(f)$. The first is a finite-sample valid test, while the second is an asymptotic test.

\begin{proposition}\label{prop:finite_test}
    Fix $n\geq 4$ and $\alpha\in(0,1)$. For the observed predictions $\bZ=(Z_1,\ldots,Z_n)$, define
    \begin{align*}
        \sigma^2(\bZ)=\frac{1}{n}\sum_{i=1}^{n-1}
        Z_{(i)}(1-Z_{(i)})
        Z_{(i+1)}(1-Z_{(i+1)}),
    \end{align*}
    where $Z_{(1)}\leq \dots \leq Z_{(n)}$ are the order statistics of $Z_1,\dots,Z_n\in[0,1]$.
    Define the test function
    \begin{align*}
        \phi_{n,\mathrm{finite}}
        =
        \one\left\{
        \hatrankECE(f)>
        2\sqrt{\frac{\sigma^2(\bZ)\log(2/\alpha)}{n}}
        +\frac{2\log(2/\alpha)}{3n}
        \right\}.
    \end{align*}
    Then, under $\bm H_0$, $\P_{\bm H_0}\!\left(\phi_{n,\mathrm{finite}}=1\right)\leq\alpha$ and under $\bm H_1$, the test $\phi_{n,\mathrm{finite}}$ is consistent, i.e., $\P_{\bm H_1}\!\left(\phi_{n,\mathrm{finite}}=1\right)\to 1$.
\end{proposition}

The proof of Proposition \ref{prop:finite_test} is given in Appendix \ref{sec:proofof_prop_finite_test}. The validity of the test $\phi_{n,\mathrm{finite}}$ follows from an application of Bernstein's inequality conditional on $\bZ$. While this guarantees finite-sample control of the type I error, tests based on concentration inequalities are often conservative, which can lead to a loss of power. We therefore now develop an asymptotically valid test that achieves asymptotically exact level control while remaining consistent. The following result shows that, under the null hypothesis $\bm H_0$, the statistic $\hatrankECE(f)$ is asymptotically standard Gaussian after appropriate normalization.

\begin{theorem}\label{thm:null_asymptotic}
    For predictor $f:\cX\ra[0,1]$, under $\bm H_0$ from \eqref{eq:null_calibration}, it holds that
    \begin{align*}
        \frac{\sqrt{n}\ \hatrankECE_n(f)}{\sqrt{\E\left[Z^2(1-Z)^2\right]}}
        \dto
        \mathrm N(0,1).
    \end{align*}
\end{theorem}
The proof of Theorem \ref{thm:null_asymptotic} in Appendix \ref{sec:proofof_thm_null_asymptotic} proceeds by identifying $\hatrankECE(f)$ as the terminal value of a martingale sequence and applying the martingale CLT. The Gaussian limit established in Theorem \ref{thm:null_asymptotic} naturally leads to an asymptotically valid test, by replacing the normalizing constant $\E\left[Z^2(1-Z)^2\right]$ with its sample estimate (see Appendix \ref{appendix:proofof_cor_rank_test} for proof), as follows:

\begin{corollary}\label{cor:rank_test_calib}
    For any predictor $f:\cX\ra[0,1]$, and any $\alpha\in(0,1)$, define the test function
\begin{align}\label{eq:phi_asymptotic_test}
    \phi_n
    =
    \one\left\{
        \frac{\sqrt{n}\ \hatrankECE_n(f)}
        {\sqrt{\frac{1}{n}\sum_{i=1}^{n}Z_i^2(1-Z_i)^2}}
        > z_{\alpha}
    \right\},
\end{align}
where $z_{\alpha}$ denotes the $(1-\alpha)$-quantile of the standard Gaussian distribution $\mathrm N(0,1)$.
Then the test $\phi_n$  satisfies
    \begin{align*}
        \P_{\bm H_0}(\phi_n=1)\ra\alpha,
        \qquad
        \P_{\bm H_1}(\phi_n=1)\ra 1.
    \end{align*}
\end{corollary}

\section{Comparing $\rankECE$ and $\lbinECE$}\label{sec:compare}

As discussed in Section \ref{sec:background}, the predominant approach for estimating $\lECE(f)$ is based on binning. In this section, we formally define the binning-based estimator, study its statistical properties, and, in particular, characterize its inherent bias. We then use these results to demonstrate that $\rankECE(f)$ provides a more faithful approximation to $\lECE(f)$.
\subsection{$\lbinECE$: properties and estimation}

From Section \ref{sec:background} and Remark~\ref{rmk:hardness_result}, recall that $\lECE(f)$ faces unavoidable challenges in practice, due to hardness results that prevent assumption-free consistent estimation. These challenges necessitate the introduction of the binning approach \citep{zadrozny2001obtaining,zadrozny2002transforming,futami2024information}. Recall that the binning-based approximation $\lbinECE(f)$ in \eqref{eq:L2_binned_ECE} can be interpreted as the $\lECE$ of a binned approximation of $f$, i.e., $f_{\mathrm{bin}}(X) = \E[f(X) \mid f(X)\in \cB_j]$ for $f(X)\in\cB_j$, where $\cB_1,\dots,\cB_K$ is a prespecified partition of $[0,1]$. An immediate consequence via Jensen’s inequality (see also Proposition 3.3 in \cite{kumar2019verified}) is that
\begin{align*}
    0 \leq \lbinECE(f) \leq \lECE(f),
\end{align*}
which shows that $\lbinECE(f)$ is a weaker notion of calibration than $\lECE(f)$. In fact, $\seqsplit{\lbinECE(f)}$ can severely underestimate $\lECE(f)$ (see Example 3.2 in \cite{kumar2019verified}). Nevertheless, this approach remains the most widely used approximation to $\lECE(f)$. 

To estimate $\lbinECE$ with a finite sample, given observations $\{(Y_i, Z_i)\}_{i=1}^n$ let
\begin{align*}
\bar Y_j = \frac{1}{n_j}\sum_{i=1}^{n}Y_i\cdot \one\{Z_i \in \cB_j\},
\qquad
\bar Z_j = \frac{1}{n_j}\sum_{i=1}^{n}Z_i\cdot \one\{Z_i \in \cB_j\}
\end{align*}
denote the average values of $Y$ and of $Z$ within each bin $\cB_j$. We then define the binned estimator as
\begin{align}\label{eq:binhatECE}
\binhatECE(f) = \sum_{j=1}^{K}\hat p_j\cdot \left(\bar Y_j - \bar Z_j\right)^2,
\end{align}
for $\hat p_j = n_j/n$, where $n_j = \sum_{i=1}^{n}\one\{Z_i \in \cB_j\}$ denotes the number of samples falling in bin $\cB_j$. (If a bin $\cB_j$ has no samples, $n_j=0$, then this term's contribution to the sum is simply taken to be $0$.)

In order to determine whether $\lbinECE$ offers a practically useful notion of calibration, we need to verify that it can $\lbinECE$ be estimated with a finite sample. Towards this question, the next result shows that $\binhatECE(f)$ from \eqref{eq:binhatECE} concentrates around its expectation (see Appendix \ref{sec:proofof_conc_binhat} for a proof):

\begin{proposition}\label{prop:conc_binhat}
    For any predictor $f:\cX\ra[0,1]$, any prespecified partition $[0,1]=\cB_1\cup\dots\cup\cB_K$, and any $\delta\in (0,1]$,
    \begin{align*}
        \left|\binhatECE(f) - \E\left[\binhatECE(f)\right]\right|\leq \sqrt{\frac{18\log (2/\delta)}{n}}
    \end{align*}
    with probability at least $1-\delta$.
\end{proposition}

While Proposition \ref{prop:conc_binhat} shows that $\binhatECE$ concentrates around its expectation, the primary issue is that this expectation may not coincide with $\lbinECE(f)$. Indeed, it has been noted that $\ell_1$-based variants exhibit a statistical bias that cannot be eliminated without additional assumptions \citep{futami2024information}. The following results show that this bias persists in our setting as well (the proof is presented in Appendix \ref{sec:proofof_lemma_stat_bias_bound}).

\begin{lemma}\label{lemma:stat_bias_bound}
Consider any prediction function $f:\cX\ra[0,1]$. Then the empirical calibration error $\binhatECE(f)$ and its population counterpart $\lbinECE(f)$, defined in \eqref{eq:binhatECE} and \eqref{eq:L2_binned_ECE} respectively, satisfy
\begin{align*}
    \left|\E\left[\binhatECE(f)\right]-\lbinECE(f)\right|
    \leq
    \frac{K}{n}.
\end{align*}
\end{lemma}

\begin{remark}\label{rmk:stat_bias_bound}
    The bound is indeed tight up to constants independent of $K$ and $n$. To see this, consider the perfectly calibrated setting where $Z = f(X)\sim\textnormal{Unif}[0,1]$ and $Y\mid Z\sim\textnormal{Ber}(Z)$, so that $\lbinECE(f)=0$. Moreover for $1\leq K\leq n$ let $[0,1] = \cB_1\cup\cdots\cup\cB_K$ be a partition of $[0,1]$ into $K$ equal-width bins. Then, for $\binhatECE(f)$ with the above choice of bins, we can show (see Appendix \ref{sec:e_binhat_O1}) that
\begin{align*}
    \E\left[\binhatECE(f)\right]
    =
    \frac{K}{6n}\left(1-\left(1-\frac{1}{K}\right)^n\right)
    \geq
    \frac{1}{6}\left(1-\frac{1}{e}\right)\frac{K}{n}.
\end{align*}
\end{remark}

The key implication of the above result is that $\binhatECE$ possesses an irreducible statistical bias when estimating $\lbinECE$. However, when the number of bins satisfies $K = o(n)$, this bias becomes asymptotically negligible. Consequently, $\lbinECE$ can be estimated in a distribution-free manner only when $K = o(n)$ (for instance, \cite{futami2024information} propose $K\propto n^{1/3}$ as an optimal choice of the number of bins).

\subsection{$\rankECE$ as an alternative to $\lbinECE$}

In this section, we compare our proposed calibration measure $\rankECE(f)$ from \eqref{eq:rank_ECE} with the commonly used binning-based measure $\lbinECE(f)$.
In particular, we are interested in comparing the two in terms of their ability to approximate $\lECE(f)$: while no assumption-free estimation is possible (as discussed in Section~\ref{sec:background} and Remark~\ref{rmk:hardness_result}), how do the two compare in terms of the settings in which they do successfully provide an accurate approximation to $\lECE$? 

Proposition \ref{prop:lower_bound}, together with Proposition 3.3 from \cite{kumar2019verified}, shows that both measures $\rankECE(f)$ and $\lbinECE(f)$ induce strictly weaker notions of calibration than $\lECE(f)$, at least in finite samples. This naturally raises the question of whether one of these measures nevertheless provides a stronger notion of calibration than the other. In the following result, we answer this question in the affirmative by showing that, up to finite-sample correction terms, $\rankECE$ dominates $\lbinECE$. We refer the reader to Appendix \ref{sec:proofof_thm_rank_bin_comparision} for a proof.

\begin{theorem}\label{thm:rank_bin_comparison}
Let $f:\cX\ra[0,1]$ be any predictor. For $K\geq 1$ and $0=a_0<a_1<\cdots<a_K=1$, let $\cB_1=[a_0,a_1), \cB_2=[a_1,a_2),\ldots, \cB_K=[a_{K-1},a_K]$ be a partition of $[0,1]$ into $K$ bins. Consider $\lbinECE(f)$ as defined in \eqref{eq:L2_binned_ECE}, with the choice of bins $\cB_1,\ldots,\cB_K$. Then, for $n\geq 2$,
\begin{align*}
    \rankECE(f) \geq \lbinECE(f) - \frac{4K}{n}.
\end{align*}
\end{theorem}

Recalling that both $\rankECE$ and $\lbinECE$ are proposed as approximations to $\lECE$, we can interpret this result in that regard. In particular, both $\rankECE(f)$ and $\lbinECE(f)$ are  upper bounded by $\lECE(f)$, and therefore we cannot overestimate $\lECE(f)$ with either of these (population) measures; on the other hand, it is possible to substantially underestimate the $\lECE$, but the result above tells us that gaps must satisfy
\[\Big(\lECE(f) - \rankECE(f)\Big) \leq \Big(\lECE(f) - \lbinECE(f)\Big)  + \frac{4K}{n}.\]
That is, for $K=o(n)$, the binned approximation to ECE will underestimate $\lECE(f)$ (nearly) \emph{at least as much} as $\rankECE(f)$. (And, we are forced to choose $K=o(n)$ so that $\lbinECE(f)$ can be accurately estimated without assumptions, as established in Lemma~\ref{lemma:stat_bias_bound} and Remark~\ref{rmk:stat_bias_bound}.)

\begin{remark}\label{remark:rate_of_convg_compare}
    Recall that, when the residual function $r(z)=\E[Y-Z\mid Z=z]$ is $L$-Lipschitz, Proposition \ref{prop:rank_l2_finite_bd} shows that the approximation error of $\rankECE$ converges at the rate $O(L/n)$. See Appendix \ref{sec:remark_lip_rate} for an analogous result for $\lbinECE$.
\end{remark}

To further illustrate the benefit of $\rankECE$ in the assumption-free regime, the following example gives a construction showing a scenario where $\rankECE$ can detect miscalibration, while $\lbinECE$ can fail without careful tuning of the binning scheme based on properties of the unknown residual function.

\begin{figure}[!h]
        \centering
        \includegraphics[width = 0.8\textwidth]{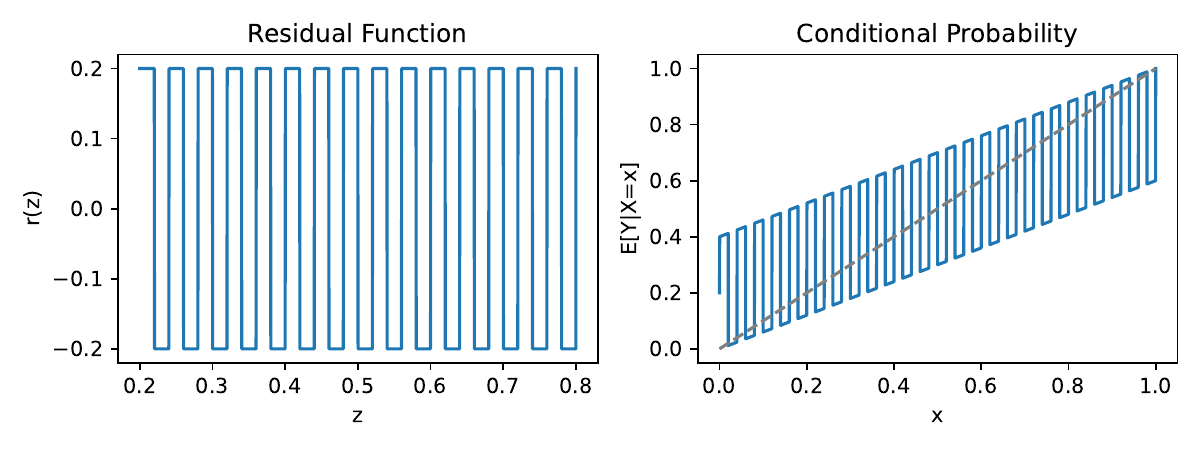}
        \caption{Residual function $r(z) $ and conditional probability $\E[Y\mid X=x] = c\phi_m(x) + x$ from Example \ref{example:compare} with $c = 0.2$ and $m = 25$.}
        \label{fig:function_plot_example_compare}
\end{figure}

\begin{example}\label{example:compare}
We compare our proposed measure $\rankECE$, with the binning-based measure $\lbinECE$. Fix $c\in (0, 1/2)$, and for $m\in \N$ and $x\in [0,1]$ define $\phi_m(x) = \textnormal{sign}(\sin(2\pi mx))$. Generate $X\sim \textnormal{Unif}[0,1]$, and $Y\mid X \sim \textnormal{Ber}\left(c(1-X) + (1-c)X + c\phi_m(X)\right)$; this mean function $X\mapsto c + (1-2c)X + c\phi_m(X)$ is a combination of a linear function ranging from $c$ to $1-c$, offset by a term $c\phi_m(X)$ that oscillates $m$ times. Consider the prediction function $Z = f(X)=c + (1-2c)X$. In this setting, the residual function is given by
\begin{align*}
    r(z) = \E[Y-Z\mid Z=z] = c\phi_m\left(\frac{z-c}{1-2c}\right)
\end{align*}
We plot both the residual function $r$ and conditional probability $\P(Y = 1\mid X) = c\phi_m(X) + X$ for $c = 0.2$ and $m = 25$ in Figure \ref{fig:function_plot_example_compare}. Hence, by definition,
    \[
        \lECE(f)=\E[r^2(Z)]=c^2.
    \]
    Moreover, a direct computation (see Appendix \ref{appendix:proofof_lemma_example_rank_bd}) shows that $\|r\|_{\textnormal{TV}}\leq 4mc$, and hence, by Proposition \ref{prop:rank_l2_finite_bd} and Proposition \ref{prop:lower_bound}, it follows that
    \begin{align*}
    c^2 - \frac{4mc + 1}{n}\leq \rankECE(f)\leq c^2.
    \end{align*}
    On the other hand, for $\lbinECE$, we consider an arbitrary choice of $K$ bins as in Theorem \ref{thm:rank_bin_comparison}. In particular, let $c = a_0<a_1<\cdots<a_{K} = 1-c$ and choose $\cB_1 = [a_0, a_1), \cB_2 = [a_1, a_2),\ldots, \cB_{K} = [a_{K-1}, a_K]$ to be an arbitrary partition of $[0, 1]$ into $K$ bins. It can be shown (see Lemma \ref{lemma:example_bin_bd}) that, for $\lbinECE(f)$ with any such natural choice of bins,
    \begin{align*}
    \lbinECE(f)\leq \frac{c^2K}{m}.
    \end{align*}
    Thus, when the number of bins satisfies $K\ll m\ll n$, it follows that $\rankECE(f)$ can detect non-trivial miscalibration, whereas $\lbinECE(f)$ fails to detect it. In particular, in this regime, $\rankECE(f)\gg\lbinECE(f)$. This example shows that when the residual function exhibits moderate oscillations, $\rankECE$, which intuitively behaves like a data-adaptive binning scheme with the finest possible resolution of $O(1/n)$, can adapt to the residual structure and detect miscalibration. In contrast, $\lbinECE$ can fail to capture such oscillations: we may be unable to choose an appropriate value of $K$ (i.e., here choosing too few bins, $K\ll m$, leads to the lack of detection of miscalibration) without access to knowledge of the (in practice unknown) properties of the residual function $r$.
\end{example}
\begin{remark}

Binning-based strategies have also been used to construct post-hoc transformations of existing predictors $f(X)$, that is, transformations of the form $h(f(X))$ that are guaranteed to satisfy a specified binning-based calibration criterion. In particular, the methods of \citet{gupta2021calibration,zadrozny2001obtaining} use a \emph{data-adaptive} choice of bins, rather than pre-specified bins as in the standard definition of binned ECE. In this sense, our work is related to this literature, since $\rankECE$ may be interpreted as using data-adaptive bins, each containing two data points. Our objective, however, is different: we aim to measure miscalibration rather than to construct a post-hoc calibrated predictor. Relatedly, several works have considered extensions of $\lbinECE$ that use data-adaptive binning schemes \citep{nguyen2015posterior,hendrycks2018deep,nixon2019measuring,futami2024information}. However, such methods have been found to be sensitive to the choice of tuning parameters, potentially leading to inconsistent conclusions \citep{Ashukha2020Pitfalls}.

\end{remark}

\section{Experiments}\label{sec:experiments}
In this section, we conduct experiments to complement our theoretical findings. In particular, in Section~\ref{sec:compare_sims} we consider synthetic settings and compare $\rankECE(f)$ and $\lbinECE(f)$ as approximations to $\lECE(f)$. We then extend this comparison to real, pretrained models for sentiment prediction on Amazon and Yelp reviews in Section~\ref{sec:sentiment_main}. Finally, to assess the practical implications of approximation quality in a downstream task, in Section~\ref{sec:empirical_test_calib} we evaluate hypothesis tests for perfect calibration based on $\rankECE(f)$ and compare them with existing kernel-based tests in terms of statistical power and computational cost.\footnote{Code for all experiments available at \url{https://github.com/anirbanc96/rankece}.}

\subsection{$\rankECE$ and $\lbinECE$ as approximations to $\lECE$}\label{sec:compare_experiments}

In this section, we empirically compare the proposed $\rankECE(f)$ and $\lbinECE(f)$ as approximations to $\lECE(f)$, complementing Theorem~\ref{thm:rank_bin_comparison} through controlled simulations and sentiment prediction experiments. In particular, these experiments provide empirical validation of the theorem's claim that $\rankECE(f)$ provides a closer approximation to $\lECE(f)$ than $\lbinECE(f)$. Both experiments follow a common evaluation protocol, which we describe first to avoid repetition before presenting the experimental details.

\paragraph{Evaluation Protocol.}
We evaluate $\rankECE(f)$ using its empirical estimator $\hatrankECE(f)$ from \eqref{eq:rank_ece_sample}, and $\lbinECE(f)$ using its empirical estimator $\binhatECE(f)$ from \eqref{eq:binhatECE}. For the number of bins, we consider $K\in\{10,\sqrt{n},n^{1/3},n/20\}$. Throughout this section, we assess the quality of these approximations using the ratios $\hatrankECE(f)/\lECE(f)$ and $\seqsplit{\binhatECE(f)/\lECE(f)}$. We describe how $\lECE(f)$ is evaluated for the simulated experiments in Section~\ref{sec:compare_sims} and for the sentiment prediction experiments in Section~\ref{sec:sentiment_main}. For each experiment, we repeat the procedure $T=100$ times and report the mean ratios with $\pm 1$ standard error.

\subsubsection{Simulations}\label{sec:compare_sims}
For simulated experiments to compare $\rankECE$ and $\lbinECE$ we consider conditional probability functions $g(z) = \E[Y\mid Z = z]$ with varying degrees of complexity and report the findings using the above described evaluation protocol. In this section we focus on the following three representative conditional probability functions (additional examples are presented in Appendix \ref{appendix:compare_sims}).

\begin{itemize}
    \item \texttt{Quadratic}: $g(z) = 0.8\,z^2 - 0.2\,z + 0.2$
    \item \texttt{Spikes}: $g(z) = z + \sum_{i=1}^{10} a_i \exp\!\left(-\frac{(z - c_i)^2}{2 (0.02)^2}\right)$ with centers $c_i = 0.01 + \dfrac{0.98\,(i-1)}{9}$ for $i=1,\dots,10$, and alternating amplitudes $a_i \in \{0.3, -0.3\}$.
    \item \texttt{Increasing Frequency}: $g(z) = 0.15 + 0.7 z + 0.5 \sin\!\big(40\pi z^2\big)$
\end{itemize}

\begin{figure}[!h]
    \centering
    \includegraphics[width = \textwidth]{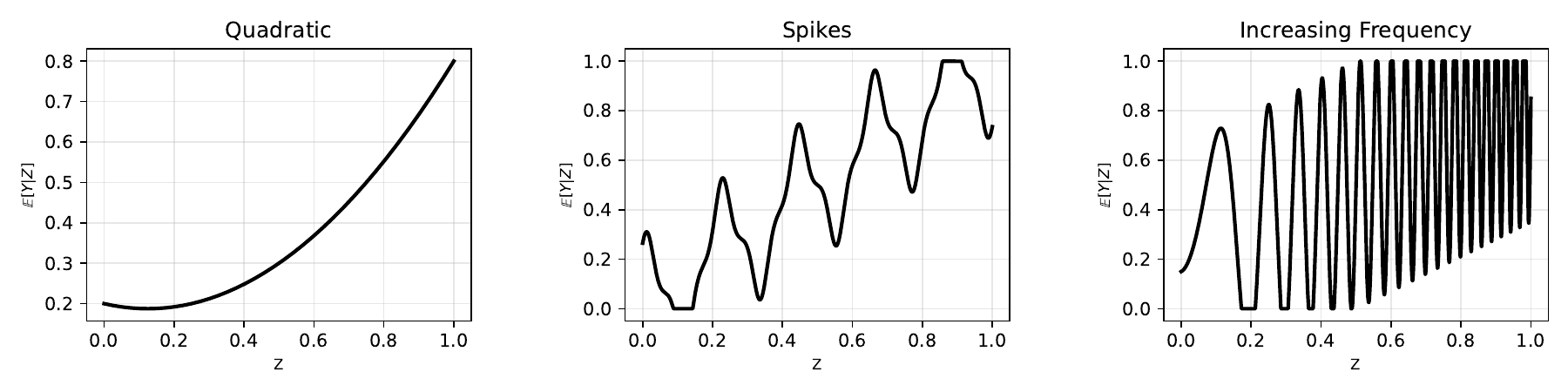}
    \caption{Representative conditional probability functions $g(z)$ \texttt{Quadratic}, \texttt{Spikes}, \texttt{Increasing Frequency} used in the simulation study in Section \ref{sec:compare_sims}.}
    \label{fig:prob_functions_main}
\end{figure}

The three conditional probability functions are shown in Figure~\ref{fig:prob_functions_main}. To ensure that they define valid conditional probabilities, each function is clipped to the interval $[0,1]$. For each conditional probability function $g$, we generate a dataset of size $n$ according to
\begin{align}\label{eq:dgp_compare_sims}
Z_i = f(X_i) \stackrel{\mathrm{i.i.d.}}{\sim} \mathrm{Unif}[0,1], \qquad
Y_i \mid Z_i \sim \mathrm{Ber}\left(g(Z_i)\right) \text{ for all } i=1,\ldots,n.
\end{align}

We follow the evaluation protocol described in Section \ref{sec:compare_experiments}, approximating $\lECE(f)$ numerically using a fine uniform grid on $[0,1]$. We present the results in Figure \ref{fig:compare_sims}, with results for additional probability functions deferred to Appendix \ref{appendix:compare_sims}.

\begin{figure}[h]
    \centering

    \begin{subfigure}{0.32\textwidth}
        \centering
        \includegraphics[width=\linewidth]{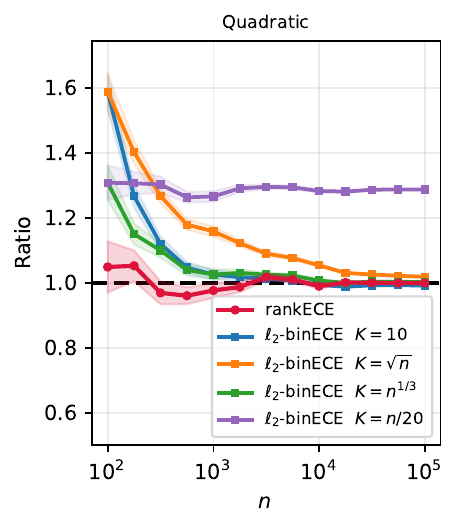}
        \label{fig:first_sims}
    \end{subfigure}
    \hfill
    \begin{subfigure}{0.32\textwidth}
        \centering
        \includegraphics[width=\linewidth]{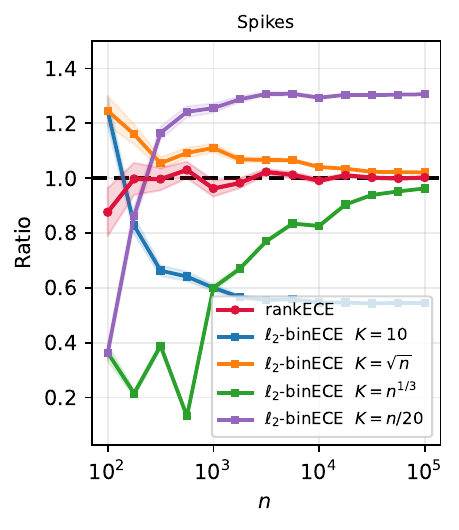}
        \label{fig:second_sims}
    \end{subfigure}
    \hfill
    \begin{subfigure}{0.32\textwidth}
        \centering
        \includegraphics[width=\linewidth]{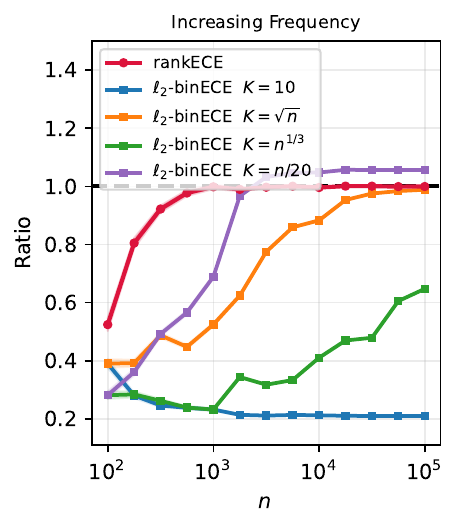}
        \label{fig:third_sims}
    \end{subfigure}

    \caption{Empirical comparison of $\rankECE(f)$ and $\lbinECE(f)$ as approximations to $\lECE(f)$ for the conditional probability functions \texttt{Quadratic}, \texttt{Spikes} and \texttt{Increasing Frequency} from Figure \ref{fig:prob_functions_main}. The horizontal axis shows the sample size $n$, while the vertical axis plots the ratios $\hatrankECE(f)/\lECE(f)$ and $\binhatECE(f)/\lECE(f)$ with $K=10$, $\sqrt{n}$, $n^{1/3}$, and $n/20$ bins.}
    \label{fig:compare_sims}
\end{figure}

The main takeaway from the comparisons in Figure \ref{fig:compare_sims} and Figure \ref{fig:compare_sims_appendix} is that $\rankECE(f)$ consistently provides a closer approximation to $\lECE(f)$ than $\lbinECE(f)$. Across all the considered settings, the ratio $\hatrankECE(f)/\lECE(f)$ converges rapidly to $1$ as the sample size $n$ increases and, moreover, remains substantially closer to $1$ than $\binhatECE(f)/\lECE(f)$ even for moderate sample sizes. In particular, Figure \ref{fig:compare_sims} shows that as the conditional probability function becomes increasingly oscillatory, $\binhatECE(f)$ with a fixed choice of $K=10$ bins fails to provide an accurate approximation of $\lECE(f)$. On the other hand, choosing $K=n/20 = O(n)$ exhibits a clear and persistent bias, in agreement with the discussion in Section \ref{sec:compare}. The intermediate choices $K=\sqrt{n}$ and $K=n^{1/3}$ do eventually provide good approximations as $n$ becomes large. However, as the complexity of the conditional probability function increases, these choices require substantially larger sample sizes before they accurately approximate $\lECE(f)$, whereas $\hatrankECE(f)$ continues to closely track $\lECE(f)$ throughout. Similar trends are observed in the additional simulation studies presented in Appendix \ref{appendix:compare_sims}.

\subsubsection{Sentiment Prediction}\label{sec:sentiment_main}
In this section, we complement the simulation study in Section~\ref{sec:compare_sims} with an empirical comparison of $\rankECE(f)$ and $\lbinECE(f)$ for evaluating the calibration of pretrained sentiment classifiers. We consider the binary-labeled Amazon Review Polarity corpus (\texttt{\seqsplit{https://registry.opendata.aws/fast-ai-nlp/}}), where the task is to predict whether a product review expresses positive or negative sentiment. We additionally consider the binary-labeled Yelp Review Polarity corpus, available at the same source, with further details provided in Appendix~\ref{appendix:sentiment_experiments}.

For sentiment prediction, we consider four pretrained models: \texttt{DistilBERT\_SST2}, \texttt{BERT\_SST2}, \texttt{RoBERTa\_Twitter\_Sentiment}, and \texttt{BERT\_Multilingual\_Stars}, all available through \texttt{Hugging Face}. Further details on the pretrained models are provided in Appendix~\ref{sec:sentiment_details}. For the multiclass models \texttt{RoBERTa\_Twitter\_Sentiment} and \texttt{BERT\_Multilingual\_Stars}, we convert the model outputs to binary predicted probabilities by adding the softmax probabilities assigned to the positive (upper-half) classes.

\begin{figure}[!t]
    \centering

    \begin{subfigure}{0.48\textwidth}
        \centering
        \includegraphics[width=0.9\linewidth]{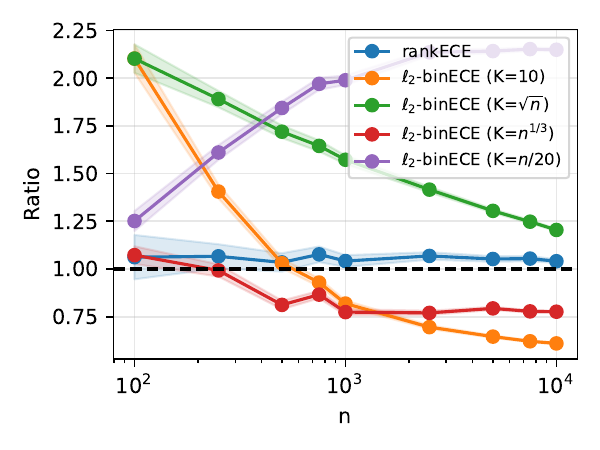}
        \caption{\texttt{BERT\_SST2}.}
        \label{fig:relative_error_bert_sst2}
    \end{subfigure}
    \hfill
    \begin{subfigure}{0.48\textwidth}
        \centering
        \includegraphics[width=0.9\linewidth]{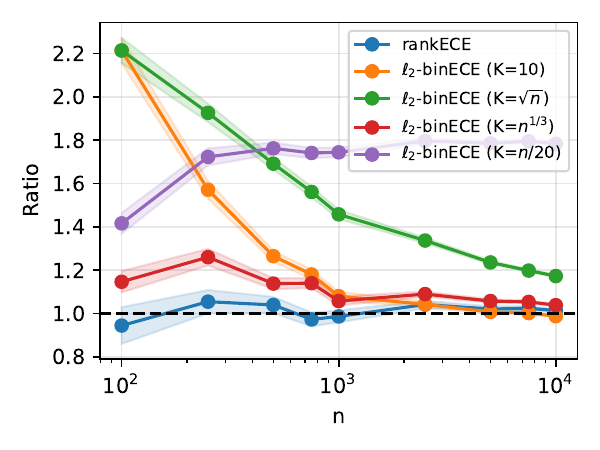}
        \caption{\texttt{BERT\_Multilingual\_Stars}.}
        \label{fig:relative_error_bert_multilingual_stars}
    \end{subfigure}

    \vspace{0.5em}

    \begin{subfigure}{0.48\textwidth}
        \centering
        \includegraphics[width=0.9\linewidth]{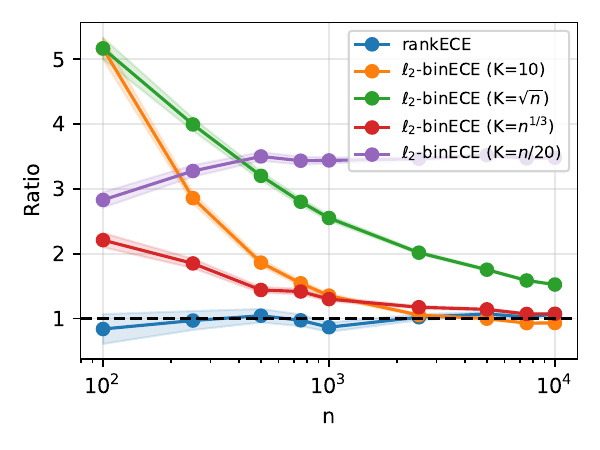}
        \caption{\texttt{RoBERTa\_Twitter\_Sentiment}.}
        \label{fig:relative_error_roberta_twitter}
    \end{subfigure}
    \hfill
    \begin{subfigure}{0.48\textwidth}
        \centering
        \includegraphics[width=0.9\linewidth]{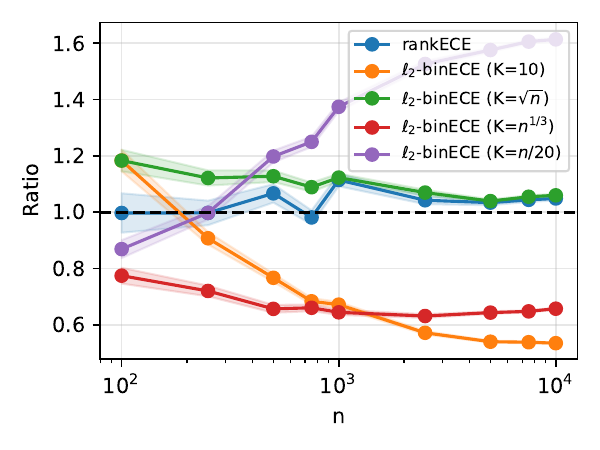}
        \caption{\texttt{DistilBERT\_SST2}.}
        \label{fig:relative_error_distilbert_sst2}
    \end{subfigure}

    \caption{The ratios $\hatrankECE(f)/\lECE(f)$ and $\binhatECE(f)/\lECE(f)$ for pretrained sentiment models: $f\in \{$\texttt{DistilBERT\_SST2}, \texttt{BERT\_SST2}, \texttt{RoBERTa\_Twitter\_Sentiment}, and \texttt{BERT\_Multilingual\_Stars}$\}$.}
    \label{fig:relative_error_models}
\end{figure}

We first randomly subsample $5\times10^4$ observations from the test set of the Amazon Polarity corpus and obtain prediction probabilities from each model on this entire subsample. For each model $f\in \{$\texttt{DistilBERT\_SST2}, \texttt{BERT\_SST2}, \texttt{RoBERTa\_Twitter\_Sentiment}, and \texttt{\seqsplit{BERT\_Multilingual\_Stars}}$\}$, we then repeatedly draw samples of sizes
\[
n\in\{100,250,500,750,1000,2500,5000,7500,10000\}.
\]
For each model $f$ and sample size $n$, we denote the resulting evaluation dataset by
\[
\mathcal{D}_{n,f}=\{(Y_i,Z_i)\}_{i=1}^n,
\]
where $Y_i\in\{0,1\}$ is the ground-truth sentiment label and $Z_i\in[0,1]$ is the binary prediction probability produced by model $f$ for the positive class. We apply the evaluation protocol described in Section~\ref{sec:compare_experiments} separately to each $\mathcal{D}_{n,f}$. We estimate $\lECE(f)$ using the debiased version of $\binhatECE(f)$ proposed by \cite{lee2023t}, computed on the full $5\times10^4$-observation subsample.

The results are presented in Figure \ref{fig:relative_error_models}. We also repeat the same experiment with Yelp Review Polarity corpus (also available at \texttt{https://registry.opendata.aws/fast-ai-nlp/}) in Appendix \ref{appendix:sentiment_experiments}. They are consistent with the findings of Section \ref{sec:compare_sims}. Across all four models, $\rankECE$ provides a substantially more accurate approximation of $\lECE$ than $\lbinECE$. In contrast, the behavior of $\lbinECE$ depends strongly on the choice of the number of bins. Using $K=n/20$ bins yields a persistent positive bias, while using a fixed number of bins ($K=10$) performs well for the \texttt{RoBERTa\_Twitter\_Sentiment} and \texttt{BERT\_Multilingual\_Stars} models but remains inconsistent for the two \texttt{SST2} models. Choosing $K=\sqrt{n}$ produces estimates that gradually approach $\lECE$ as the sample size increases, whereas $K=n^{1/3}$ still exhibits noticeable finite-sample bias for \texttt{BERT\_SST2} and \texttt{DistilBERT\_SST2}. Overall, these experiments highlight the importance of choosing the number of bins appropriately when using $\lbinECE(f)$ as an approximation to $\lECE(f)$. As discussed in Remark \ref{remark:rate_of_convg_compare} and Example \ref{example:compare}, the optimal choice of bins depends on the typically unknown residual function. In contrast, $\rankECE(f)$ consistently provides a robust approximation to $\lECE(f)$ across all models, demonstrating its stability and effectiveness over a diverse collection of real-world prediction models.

\subsection{Testing for Perfect Calibration}\label{sec:empirical_test_calib}

In this section, we empirically evaluate the proposed tests for perfect calibration based on $\rankECE$ from Section \ref{sec:perfect_calib} and compare their performance with tests based on the \emph{Squared Kernel Calibration Error (SKCE)} of \cite{widmann2019calibration}. Specifically, we compare both the statistical power and the computational cost of the competing procedures.

\begin{figure}[ht]
\centering

% First row
\begin{subfigure}{0.48\textwidth}
    \centering
    \includegraphics[width=0.8\linewidth]{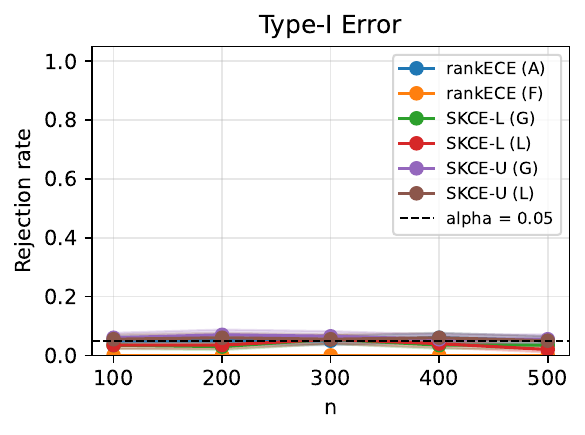}
    \caption{Type-I error vs.\ sample size.}
    \label{fig:first_typeI}
\end{subfigure}
\hfill
\begin{subfigure}{0.48\textwidth}
    \centering
    \includegraphics[width=0.8\linewidth]{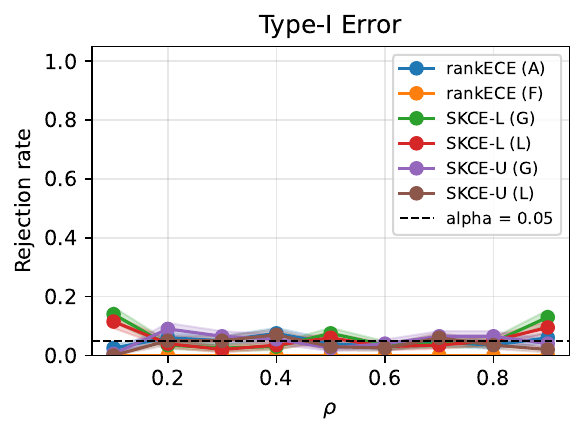}
    \caption{Type-I error vs.\ $\rho$.}
    \label{fig:second_typeI}
\end{subfigure}

\vspace{0.8em}

% Second row
\begin{subfigure}{0.48\textwidth}
    \centering
    \includegraphics[width=0.8\linewidth]{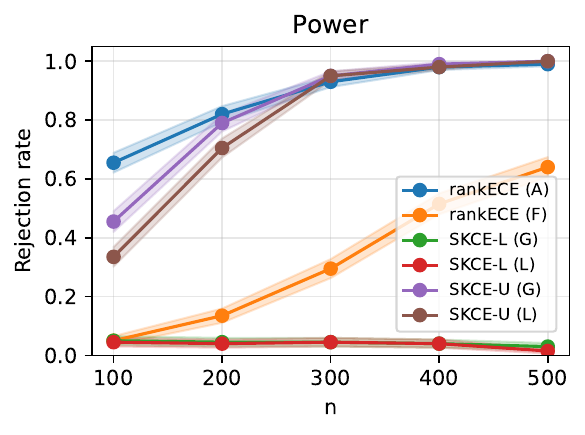}
    \caption{Power vs.\ sample size.}
    \label{fig:first_power}
\end{subfigure}
\hfill
\begin{subfigure}{0.48\textwidth}
    \centering
    \includegraphics[width=0.8\linewidth]{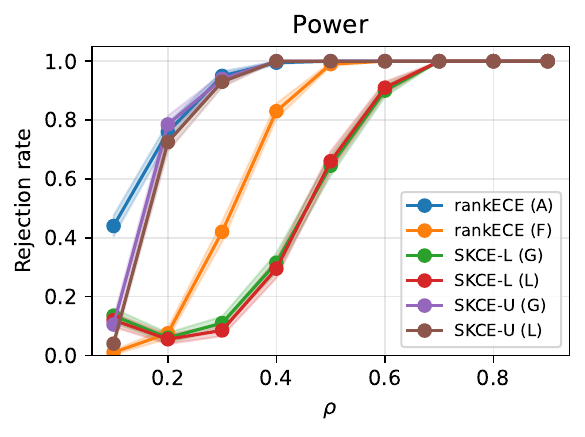}
    \caption{Power vs.\ $\rho$.}
    \label{fig:second_power}
\end{subfigure}

\caption{
Comparison of the $\rankECE$-based tests from Section~\ref{sec:perfect_calib} and the SKCE-based tests of \cite{widmann2019calibration}. The top row shows empirical type-I error, while the bottom row shows empirical power under varying sample sizes and values of $\rho$.
}
\label{fig:typeI_power_compare}

\end{figure}

We consider the following data generation process: Fix $\rho\in (0,1)$ and consider,
\begin{align*}
    Z = f(X)\sim \text{Beta}(\rho, 1-\rho)\text{ and }Y\mid Z\sim \text{Ber}(Z-Z^4).
\end{align*}
For each $\rho\in(0,1)$, given samples $(Y_i,Z_i)_{i=1}^{n}$ generated from the above data-generating process, we test the hypothesis of perfect calibration in \eqref{eq:null_calibration}. This construction induces a misspecified and, therefore, imperfectly calibrated prediction model.

We compare the finite-sample and asymptotic tests based on $\rankECE$ from Section \ref{sec:perfect_calib}, denoted by \texttt{rankECE (F)} and \texttt{rankECE (A)}, respectively, with the SKCE-based tests of \cite{widmann2019calibration}. For SKCE, we consider both the quadratic-time U-statistic estimator, denoted by \texttt{SKCE-U}, and the linear-time estimator, denoted by \texttt{SKCE-L}. For each estimator, we consider both the Gaussian (G) and Laplace (L) kernels, yielding the four variants \texttt{SKCE-U (G)}, \texttt{SKCE-U (L)}, \texttt{SKCE-L (G)}, and \texttt{SKCE-L (L)}. As noted in \cite{widmann2019calibration}, the asymptotic distribution of \texttt{SKCE-U} is intractable, requiring a resampling procedure to calibrate the test. Throughout our experiments, we implement the multiplier bootstrap proposed in Appendix G of \cite{widmann2019calibration} using $B = 100$ bootstrap resamples.

Our first experiment evaluates the type-I error control of all competing tests. To this end, we consider two complementary settings. In Figure \ref{fig:first_typeI}, we fix $\rho = 0.3$ and plot the empirical type-I error as the sample size $n$ increases. In Figure \ref{fig:second_typeI}, we fix $n = 100$ and plot the empirical type-I error as $\rho$ varies over $(0,1)$. The nominal significance level is $\alpha = 0.05$. Both figures demonstrate that all methods achieve the desired type-I error control across the range of experimental settings considered.

We next compare the empirical power and computational cost of the competing procedures. In Figure \ref{fig:first_power}, we fix $\rho = 0.15$ and plot the empirical power of each method as a function of the sample size $n$, estimated from $200$ independent repetitions. In Figure \ref{fig:second_power}, we instead fix $n = 100$ and plot the empirical power as $\rho$ varies over $(0,1)$, again using $200$ independent repetitions. Additional power comparisons under a wider range of experimental settings are provided in Appendix \ref{appendix:empirical_test_calib}. 

As expected, the asymptotic $\rankECE$ test (\texttt{rankECE (A)}) is substantially more powerful than its finite-sample counterpart (\texttt{rankECE (F)}), reflecting the conservativeness of the concentration inequalities used by the latter. The asymptotic test achieves power comparable to that of \texttt{SKCE-U}, while the finite-sample test achieves power comparable to, and in some settings exceeding, that of \texttt{SKCE-L}. From a computational perspective, both $\rankECE$-based tests are significantly faster than the SKCE-based procedures, achieving orders-of-magnitude speedups over \texttt{SKCE-U}. A detailed comparison of the running times of the $\rankECE$- and SKCE-based methods is provided in Appendix~\ref{appendix:empirical_test_calib}.

\section{Discussion}\label{sec:discussions}
In this work, we have developed a new measure of calibration error, the $\rankECE$, that provides a proxy for $\lECE$ by comparing data points with neighboring values of the predicted probabilities $Z=f(X)$. This approach avoids a core challenge of the binned approximation to $\ECE$, where the choice of the number of bins $K$ introduces a tradeoff between underestimating $\ECE$ due to excessive smoothing (when $K$ is low) versus overestimating $\ECE$ due to positive bias in the empirical binned $\ECE$ (when $K$ is large). Our theoretical results demonstrate that $\rankECE$ provides a better proxy for the $\lECE$, as is validated by our experiments; in addition, $\rankECE$ offers a computationally inexpensive test for perfect calibration, that empirically shows comparable or higher power relative to existing SKCE methods.

We next discuss several open questions that are raised by these findings. First, while we see in our experiments that $\rankECE$ offers a good proxy for $\lECE$, our finite-sample guarantee in Proposition~\ref{prop:rank_l2_finite_bd} requires bounded total variation of the underlying miscalibration error $r(X)$. While
the hardness results (see Remark~\ref{rmk:hardness_result}) for estimating $\ECE$ imply that some sort of assumption is unavoidable, it remains an open question as to whether the total variation assumption offers a sharp characterization of the types of miscalibration that can, or cannot, be detected with finite samples. Second, the construction of $\rankECE$ inherently rely on the predictions $Z=f(X)$ taking values in $[0,1]$, i.e., in a one-dimensional domain (so that we can construct a ranking of $Z_1,\dots,Z_n$ in order to calculate $\rankECE$). How can these ideas be extended to settings where miscalibration may be measured in more complex ways---for instance, when making a multivariate prediction due to the presence of multiple outcomes?

\bigskip\small

\noindent\textbf{Acknowledgments.}\\
R.F.B. was partially supported by the Office of Naval Research via grant N00014-24-1-2544.\\

\noindent
\textbf{Use of Generative AI.} During the preparation of this work, the authors used AI tools for coding assistance and experiment development, and for proofreading and improving clarity and presentation of the manuscript.

\bibliographystyle{abbrvnat}
\bibliography{ref}

\normalsize

\appendix
\renewcommand{\thefigure}{\thesection.\arabic{figure}}
\setcounter{figure}{0}

\newpage

\tableofcontents

\newpage

\section{Additional Experiments}
In this section we present additional experiments to complement the ones from Section \ref{sec:experiments}. This section is organised as follows: We present additional experiments for comparing $\rankECE$ and $\lbinECE$ as approximations to $\lECE$ in Appendix \ref{appendix:compare_sims} and Appendix \ref{appendix:sentiment_experiments} and present additional experiments for testing for perfect calibration in Appendix \ref{appendix:empirical_test_calib}.

\subsection{Comparing $\rankECE$ and $\lbinECE$ as approxmations to $\lECE$}\label{appendix:compare_sims}

In this section we recall the experimental setting in Section \ref{sec:compare_sims} for comparing $\rankECE$ and $\lbinECE$. For conditional probability functions we consider the following examples, in addition to the three studied in Section \ref{sec:compare_sims} :

\begin{figure}[!h]
    \centering
    \includegraphics[width = 0.9\textwidth]{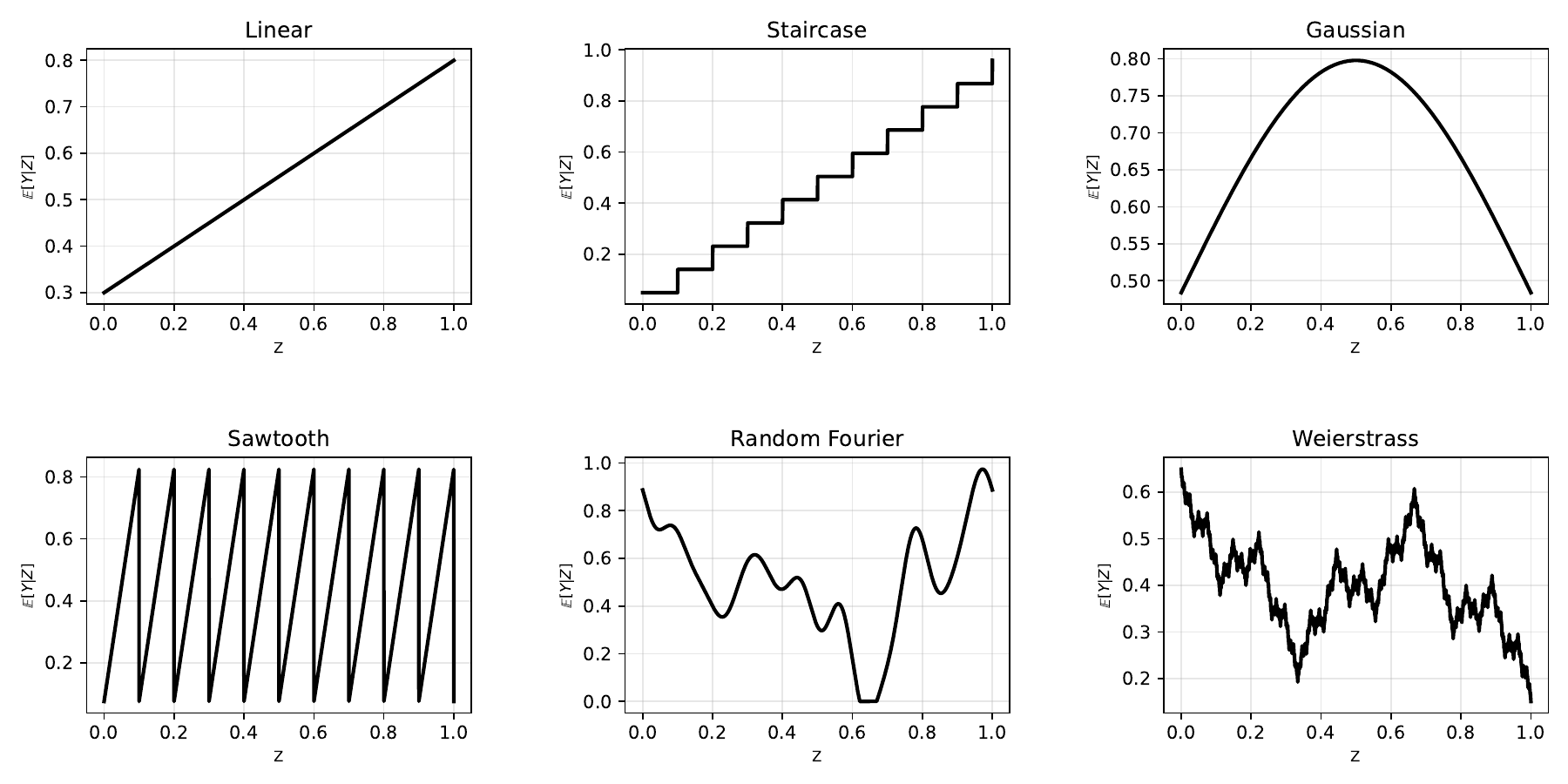}
    \caption{Conditional probability functions $\E[Y\mid Z = z]$ from Appendix \ref{appendix:compare_sims}}
    \label{fig:cond_prob_appendix}
\end{figure}

\begin{itemize}
    \item \texttt{Linear}: $\E[Y\mid Z = z] = 0.5\,z + 0.3$
    \item \texttt{Staircase}: $\E[Y\mid Z = z] = \frac{\lfloor 10 z \rfloor}{11} + 0.05$
    \item \texttt{Gaussian}: $\E[Y\mid Z = z] = \frac{1}{\sqrt{2\pi\sigma^2}}\exp\!\left(-\frac{(z-\mu)^2}{2\sigma^2}\right)$ with $\mu = \sigma = 0.5$.
    \item \texttt{Sawtooth}: $\E[Y\mid Z = z] = 0.75\big(\{10z\} + 0.1\big)$ where $\{\cdot\}$ denotes fractional part of the argument.
    \item \texttt{Random Fourier}: $\E[Y\mid Z = z] = 0.5 + 0.25\sum_{k=1}^{10}\Big[a_k \sin(2\pi k z) + b_k \cos(2\pi k z)\Big]$, where the coefficients are drawn independently as $a_k, b_k \sim \mathcal{N}(0, 1/k^2)$ once and then kept fixed throughout the experiment.
    \item \texttt{Weierstrass}: $\E[Y\mid Z = z] = 0.15 + 0.5 z + 0.25\, W(z)$ where $W(z) = \sum_{k=0}^{9} \left(\frac{1}{2}\right)^{k} \cos\!\big(3^{k}\pi z\big)$
\end{itemize}
The above $6$ conditional probability functions are plotted in Figure \ref{fig:cond_prob_appendix}. We repeat the data generation procedure from \eqref{eq:dgp_compare_sims} with the above choice of conditional probability functions $\E[Y\mid Z = z]$. For evaluation we repeat the evaluation protocol from Section \ref{sec:compare_experiments} and report the comparisions in Figure \ref{fig:compare_sims_appendix}. The empirical results in Figure \ref{fig:compare_sims_appendix} corroborate the conclusions of Section \ref{sec:compare_sims}: across a range of settings, $\rankECE$ provides a more faithful approximation to $\lECE$ than $\lbinECE$.

\begin{figure}[!h]
    \centering

    % Row 1
    \begin{subfigure}{0.32\textwidth}
        \centering
        \includegraphics[width=\linewidth]{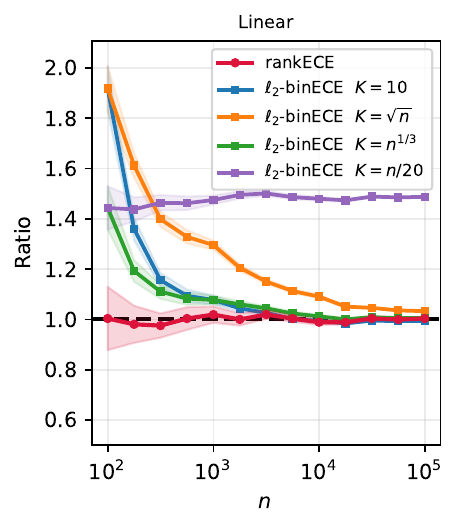}
        \caption{\texttt{Linear}.}
        \label{fig:first_sims_appendix}
    \end{subfigure}
    \hfill
    \begin{subfigure}{0.32\textwidth}
        \centering
        \includegraphics[width=\linewidth]{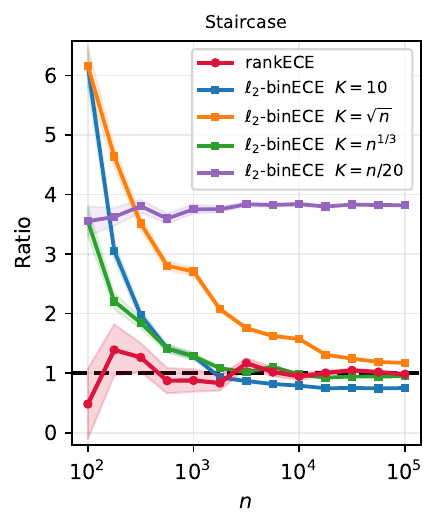}
        \caption{\texttt{Staircase}.}
        \label{fig:second_sims_appendix}
    \end{subfigure}
    \hfill
    \begin{subfigure}{0.32\textwidth}
        \centering
        \includegraphics[width=\linewidth]{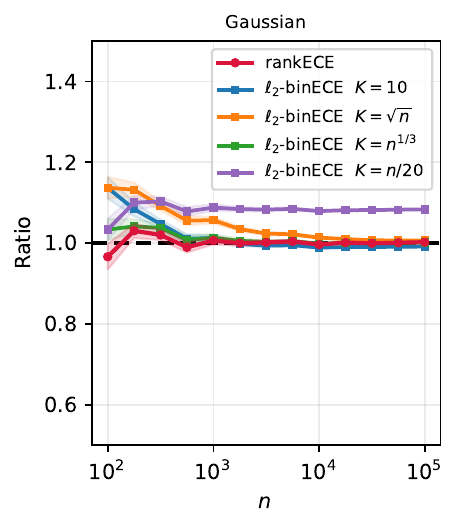}
        \caption{\texttt{Gaussian}.}
        \label{fig:third_sims_appendix}
    \end{subfigure}

    \vspace{0.5em}

    % Row 2
    \begin{subfigure}{0.32\textwidth}
        \centering
        \includegraphics[width=\linewidth]{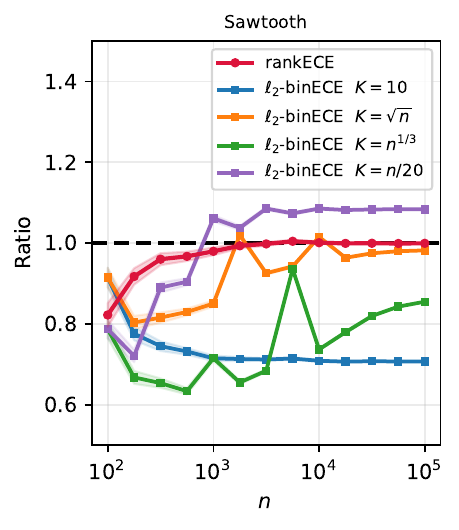}
        \caption{\texttt{Sawtooth}.}
        \label{fig:fourth_sims_appendix}
    \end{subfigure}
    \hfill
    \begin{subfigure}{0.32\textwidth}
        \centering
        \includegraphics[width=\linewidth]{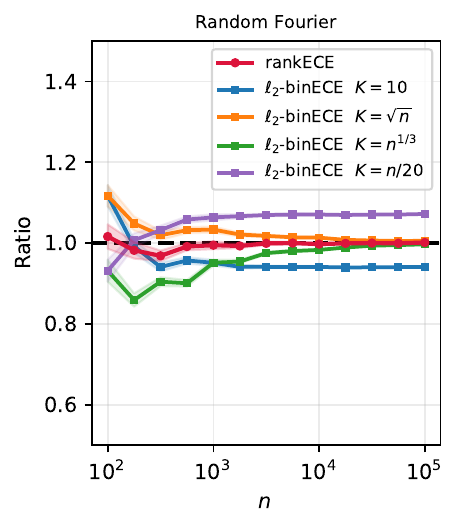}
        \caption{\texttt{Random Fourier}.}
        \label{fig:fifth_sims_appendix}
    \end{subfigure}
    \hfill
    \begin{subfigure}{0.32\textwidth}
        \centering
        \includegraphics[width=\linewidth]{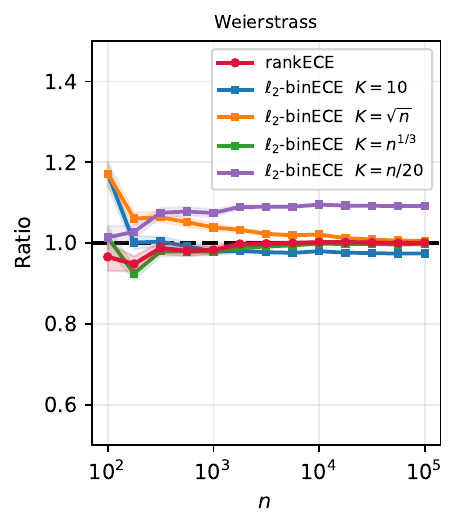}
        \caption{\texttt{Weierstrass}.}
        \label{fig:sixth_sims_appendix}
    \end{subfigure}

    \caption{Empirical comparison of $\rankECE$ and $\lbinECE$ as approximations to $\lECE$ for the conditional probability functions \texttt{Linear}, \texttt{Staircase}, \texttt{Gaussian}, \texttt{Sawtooth}, \texttt{Spikes}, and \texttt{Weierstrass} shown in Figure~\ref{fig:cond_prob_appendix}.}
    \label{fig:compare_sims_appendix}
\end{figure}

\subsection{Comparing $\rankECE$ and $\lbinECE$ in Sentiment Prediction}\label{appendix:sentiment_experiments}
In this section we repeat the empirical comparision between $\rankECE$ and $\lbinECE$ from Section \ref{sec:sentiment_main} for evaluating calibration of pre-trained senitment classifiers using the binary labeled Yelp Review Polarity corpus (\texttt{\seqsplit{https://registry.opendata.aws/fast-ai-nlp/}}). As in Section \ref{sec:sentiment_main} once again the prediction objective is to predict whether a Yelp review expresses positive or negative senitment. We consider the same four pre-trained models from Section \ref{sec:sentiment_main}, namely, \texttt{DistilBERT\_SST2}, \texttt{BERT\_SST2}, \texttt{RoBERTa\_Twitter\_Sentiment}, and \texttt{BERT\_Multilingual\_Stars}. 

\begin{figure}[ht]
    \centering

    % Row 1
    \begin{subfigure}{0.48\textwidth}
        \centering
        \includegraphics[width=0.9\linewidth]{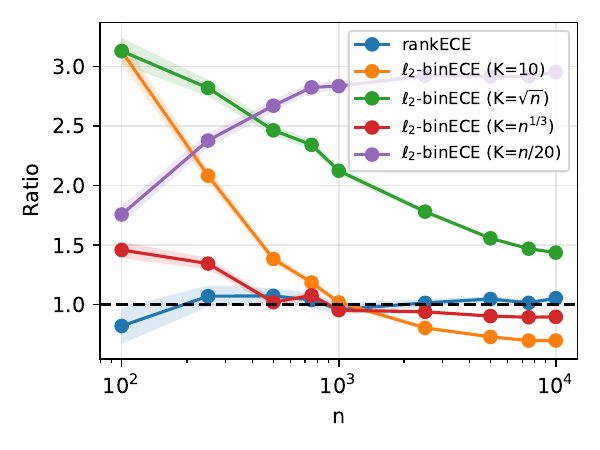}
        \caption{\texttt{BERT\_SST2}.}
        \label{fig:relative_error_yelp_bert_sst2}
    \end{subfigure}
    \hfill
    \begin{subfigure}{0.48\textwidth}
        \centering
        \includegraphics[width=0.9\linewidth]{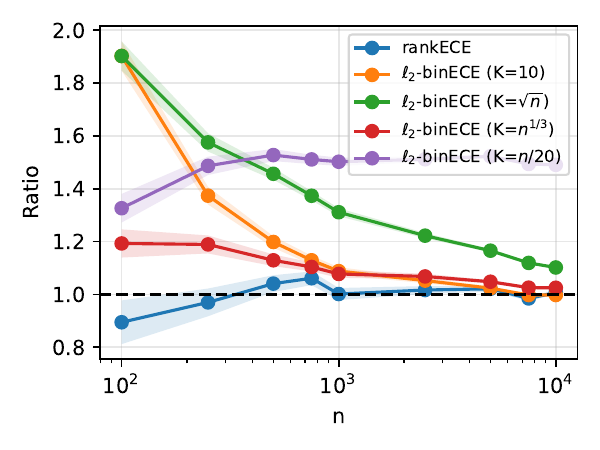}
        \caption{\texttt{BERT\_Multilingual\_Stars}.}
        \label{fig:relative_error_yelp_bert_multilingual_stars}
    \end{subfigure}

    \vspace{0.5em}

    % Row 2
    \begin{subfigure}{0.48\textwidth}
        \centering
        \includegraphics[width=0.9\linewidth]{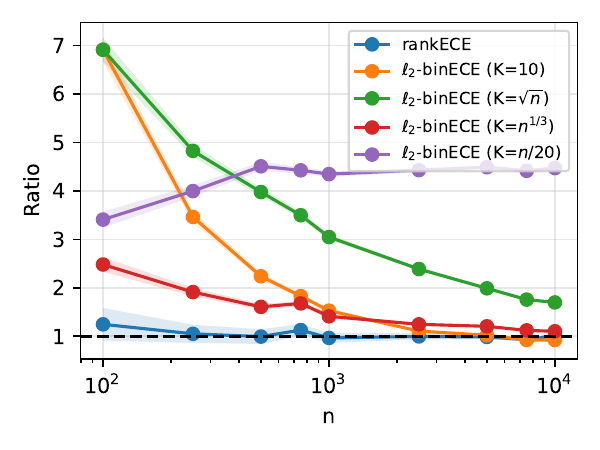}
        \caption{\texttt{RoBERTa\_Twitter\_Sentiment}.}
        \label{fig:relative_error_yelp_roberta_twitter}
    \end{subfigure}
    \hfill
    \begin{subfigure}{0.48\textwidth}
        \centering
        \includegraphics[width=0.9\linewidth]{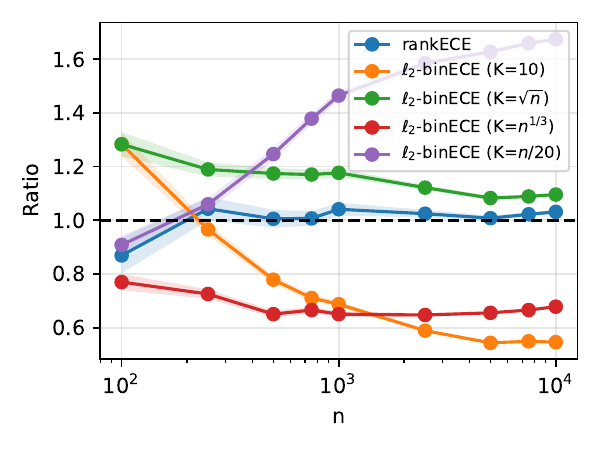}
        \caption{\texttt{DistilBERT\_SST2}.}
        \label{fig:relative_error_yelp_distilbert_sst2}
    \end{subfigure}

    \caption{The ratios $\hatrankECE(f)/\lECE(f)$ and $\binhatECE(f)/\lECE(f)$ for four pretrained sentiment models on the Yelp Review Polarity corpus.}
    \label{fig:relative_error_yelp}
\end{figure}

For this experiment, we use the full training and test sets from the Yelp Review Polarity corpus. We subsample $5\times 10^4$ observations from the combined corpus and repeat the experiment from Section \ref{sec:sentiment_main} using the evaluation protocol described in Section \ref{sec:compare_experiments}. The results are shown in Figure \ref{fig:relative_error_yelp}. The four sentiment models exhibit behavior similar to that in Figure \ref{fig:relative_error_models}, and the main conclusion of Section \ref{sec:sentiment_main} remains unchanged: $\rankECE$ provides a more robust approximation to $\lECE$ than $\lbinECE$. In particular, the performance of $\lbinECE$ depends on the number of bins and the underlying prediction model, whereas $\rankECE$ remains invariant to the choice of prediction function.

\subsection{Testing for perfect calibration}\label{appendix:empirical_test_calib}

In this section, we present additional experiments comparing the proposed $\rankECE$-based tests for perfect calibration with the SKCE-based tests of \cite{widmann2019calibration}. These experiments complement those presented in Section \ref{sec:empirical_test_calib}. We consider the following two conditional probability functions (clipped to lie in $[0,1]$):
\begin{itemize}
\item $g(z) = \E[Y\mid Z = z] = z - z^{15}$.
\item $g(z) = \E[Y\mid Z = z] = z + \frac{1}{4}\operatorname{sign}\left(\sin(10\pi z)\right)$.
\end{itemize}

\begin{figure}[ht]
\centering
\vspace{-5pt}
\begin{subfigure}{0.32\textwidth}
\centering
\includegraphics[width=\linewidth]{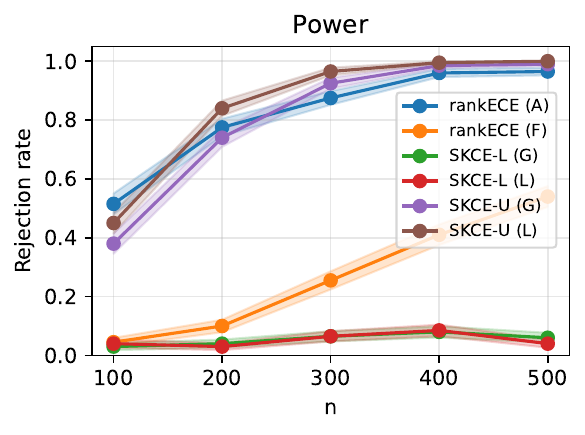}
\caption{}
\label{fig:first_power_15}
\end{subfigure}
\hfill
\begin{subfigure}{0.32\textwidth}
\centering
\includegraphics[width=\linewidth]{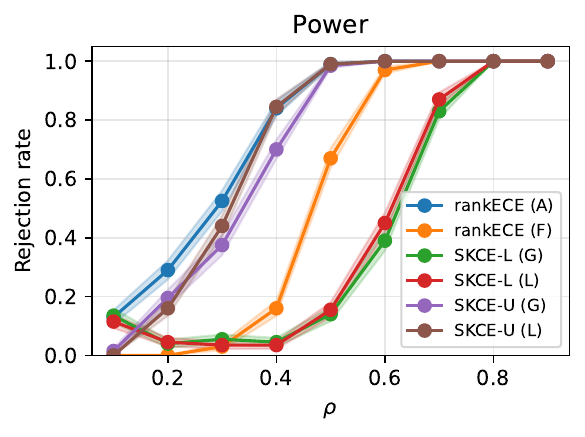}
\caption{}
\label{fig:second_power_15}
\end{subfigure}
\hfill
\begin{subfigure}{0.32\textwidth}
\centering
\includegraphics[width=\linewidth]{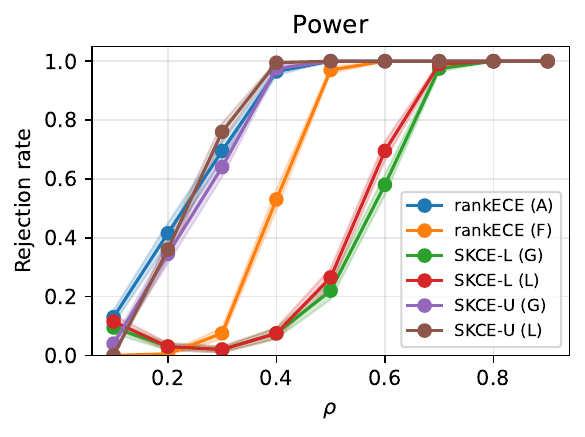}
\caption{}
\end{subfigure}

\caption{Empirical power of the $\rankECE$-based and SKCE-based tests for $g(z)=z-z^{15}$.}
\label{fig:power_compare_15}

\end{figure}

For each choice of $g$, we repeat the experimental setup from Section \ref{sec:empirical_test_calib}. Specifically, for $\rho\in (0,1)$, we generate $Z\sim \textnormal{Beta}(\rho, 1-\rho)$ and subsequently generate $Y\mid Z\sim \textnormal{Ber}(g(Z))$. We compare the empirical power of the $\rankECE$-based tests (\texttt{rankECE (A)} and \texttt{rankECE (F)}) with the SKCE-based tests (\texttt{SKCE-U (G)}, \texttt{SKCE-U (L)}, \texttt{SKCE-L (G)}, and \texttt{SKCE-L (L)}) in two settings: (i) fixing $\rho = 0.3$ while varying the sample size $n$, and (ii) fixing $n = 100, 200$ while varying $\rho$ over $(0,1)$.

\begin{figure}[!h]
    \centering

    \begin{subfigure}{0.32\textwidth}
        \centering
        \includegraphics[width=\linewidth]{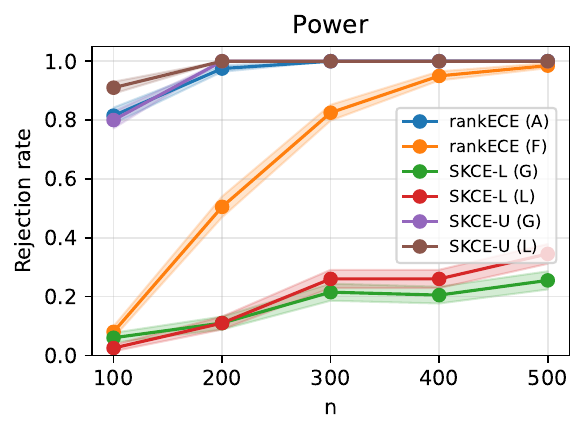}
        \label{fig:first_power_sin}
    \end{subfigure}
    \hfill
    \begin{subfigure}{0.32\textwidth}
        \centering
        \includegraphics[width=\linewidth]{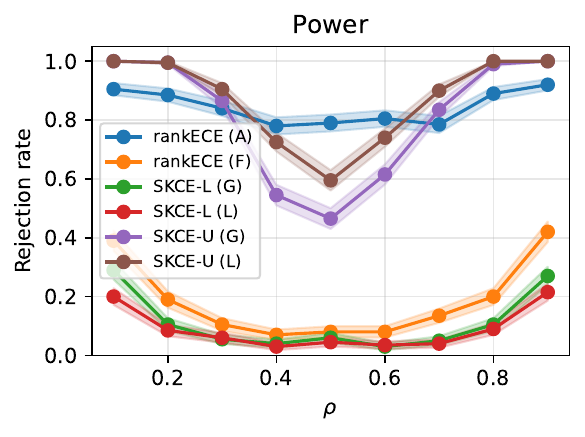}
        \label{fig:second_power_sin}
    \end{subfigure}
    \hfill
    \begin{subfigure}{0.32\textwidth}
        \centering
        \includegraphics[width=\linewidth]{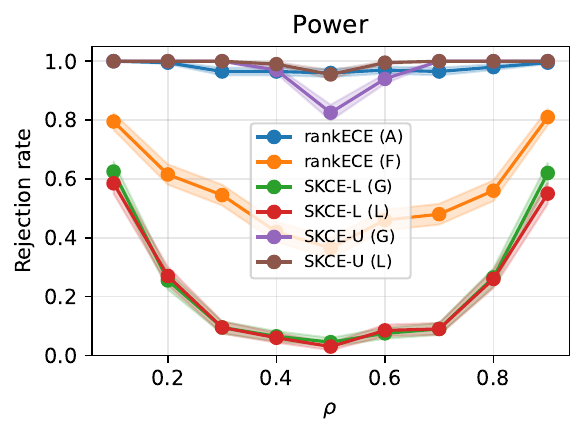}
        \label{fig:third_power_sin}
    \end{subfigure}

    \caption{Empirical power of the $\rankECE$-based and SKCE-based tests for $g(z)=z+\frac{1}{4}\operatorname{sign}\left(\sin(10\pi z)\right)$.}
    \label{fig:power_compare_sin}
\end{figure}

The results are shown in Figures \ref{fig:power_compare_15} and \ref{fig:power_compare_sin}. Consistent with the findings in Section \ref{sec:empirical_test_calib}, the asymptotic $\rankECE$ test (\texttt{rankECE (A)}) achieves power comparable to that of \texttt{SKCE-U}, while the finite-sample test (\texttt{rankECE (F)}) exhibits better power to that of \texttt{SKCE-L}. These experiments demonstrate that the conclusions of Section \ref{sec:empirical_test_calib} remain robust across substantially different forms of model misspecification.

\begin{figure}[!h]
    \centering

    \begin{subfigure}{0.32\textwidth}
        \centering
        \includegraphics[width=\linewidth]{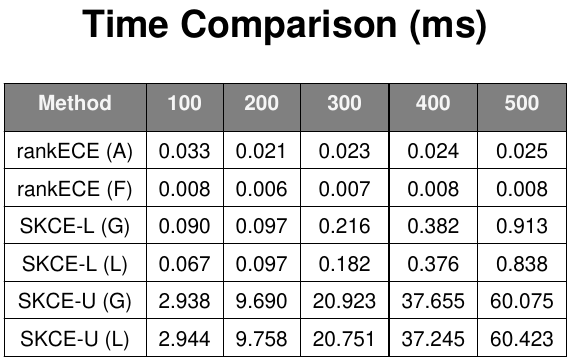}
        \caption*{\small (a)}
    \end{subfigure}
    \hfill
    \begin{subfigure}{0.32\textwidth}
        \centering
        \includegraphics[width=\linewidth]{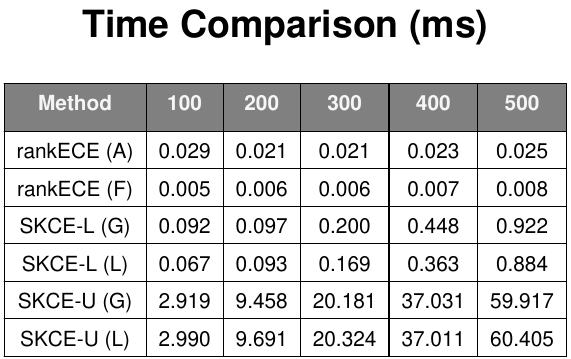}
        \caption*{\small (b)}
    \end{subfigure}
    \hfill
    \begin{subfigure}{0.32\textwidth}
        \centering
        \includegraphics[width=\linewidth]{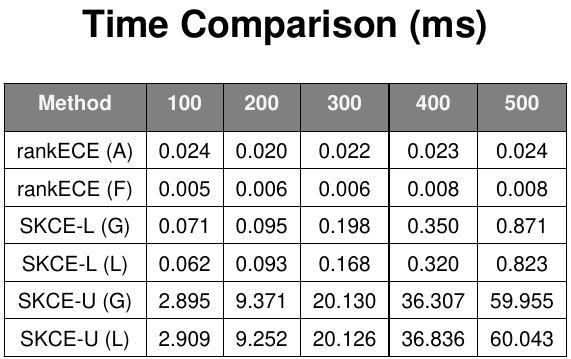}
        \caption*{\small (c)}
    \end{subfigure}

    \caption{Average running time of $\rankECE$-based and SKCE-based tests for (a) $g(z)=z-z^4$, (b) $g(z) = z - z^{15}$ and (c) $g(z) = z+\frac{1}{4}\operatorname{sign}\left(\sin(10\pi z)\right)$}
    \label{fig:time_compare}
\end{figure}

In addition to comparing power, for the same choice of conditional probability functions with $\rho = 0.5$, we compare the average running times of the $\rankECE$-based and SKCE-based tests, computed over $200$ repetitions, as shown in Figure~\ref{fig:time_compare}. These experiments were conducted on a MacBook Pro with an M4 chip. In all considered settings, both \texttt{rankECE (A)} and \texttt{rankECE (F)} are significantly faster than the SKCE-based tests. In particular, they are orders of magnitude faster than \texttt{SKCE-U}, whose quadratic-time test statistic requires resampling to determine the test threshold. This is due to the lack of a tractable asymptotic null distribution; in our experiments, we use the multiplier bootstrap following \cite{widmann2019calibration}.

\section{Testability and Actionability of $\rankECE$}\label{sec:test_action_rankece}
Two key desiderata for a calibration measure are \textit{testability} and \textit{actionability} \citep{rossellini2025can}. A measure should be testable, meaning that it can be reliably estimated from finite samples and used to assess whether a predictor is approximately calibrated. It should also be actionable, meaning that a small calibration error should provide meaningful guarantees for downstream decision-making. While these to requirements are usually in tension, in this section we will show that the proposed $\rankECE$ satisfies both. 

\subsection{Testability of $\rankECE$}
In order for a calibration metric $\Delta$ to be useful for decision-making in practice, it is important to be able to empirically verify whether a predictive model $f$ has small miscalibration $\Delta(f)$, and consequently determine whether the model can be reliably used in downstream tasks. In other words, for a fixed threshold $\tau>0$, the question of \textit{testability} of a calibration measure asks whether, given a predictive function $f$, one can test whether $\Delta(f)\leq\tau$.

The concentration result established in Proposition \ref{prop:test_rank_finite} immediately implies that $\rankECE$ is \textit{testable}. (In fact, a key advantage of having a testable calibration measure is that it enables the design of procedures that return a predictive model $f$ which, with high probability, is guaranteed to have low miscalibration with respect to the given calibration measure. Following similar constructions in \citet{blasiok2023smooth,rossellini2025can}, one can construct a preliminary trained model $f_n$, then use a holdout data set to test whether $\rankECE(f_n)$ is sufficiently small and so it is safe to output $f_n$---if not, then replace $f_n$ with a simpler constant model. We omit the details for brevity.)

\subsection{Actionability of $\rankECE$}
We now examine $\rankECE$ from a decision-theoretic perspective, focusing on the guarantees it provides for forecast quality in downstream decision-making. In a binary prediction setting, let $\tau \in [0,1]$ denote the decision maker's relative tolerance toward false positives and false negatives, and consider the loss function
\begin{align}\label{eq:tau_binary_loss}
    \ell(Y, \hat Y; \tau) = \tau(1-Y)\hat Y + (1-\tau)Y(1-\hat Y),
\end{align}
where $\hat Y$ is the predicted label. Under this loss, the Bayes-optimal decision rule is given by $\one\{\E[Y\mid f(X)]\geq\tau\}$. In practice, however, the conditional expectation $\E[Y\mid f(X)]$ is unknown, and one instead employs the plug-in decision rule $\one\{f(X)\geq \tau\}$. While the plug-in rule may be suboptimal, it is natural to ask whether a transformation $h$ can improve its decision-theoretic risk. Following \cite{rossellini2025can}, we investigate whether a small calibration error, as measured by $\rankECE(f)$, guarantees that the plug-in decision rule is close to optimal, in the sense that transformations from a class $\cH$ cannot substantially improve its risk. To that end consider the risk associated with the loss function $\ell$ from \eqref{eq:tau_binary_loss} for any predictor $f$:
\begin{align*}
    \cR(f; \tau) = \E\left[\ell\left(Y, \one\{f(X)\geq \tau\}; \tau\right)\right].
\end{align*}

In the following result, we show that imposing a shape constraint on the class of transformations yields a direct connection between calibration and decision-theoretic performance. In particular, when $\cH$ is taken to be the class of piecewise monotonic transformations of $f$, the excess risk incurred by the plug-in decision rule relative to the best transformed decision rule is controlled by $\rankECE(f)$.

\begin{theorem}\label{thm:actionable_rankece}
    Let $f:\cX\ra[0,1]$ be a fixed predictor and fix $\tau\in (0,1)$ and $M\in \N$. Let $\cH = \{h:[0,1]\ra[0,1]\text{ such that }h\text{ is $M$-piecewise monotone}\}$. Then for any $n\geq 2$,
    \begin{align*}
        \cR(f;\tau) - \inf_{h\in \cH}\cR(h\circ f;\tau)\leq \sqrt{\rankECE(f) + \frac{8(M+1)}{n}}
    \end{align*} 
\end{theorem}

Proof of Theorem \ref{thm:actionable_rankece} is deferred to Appendix \ref{sec:proofof_thm_actionable_rankece}. Similar guarantees to those in Theorem \ref{thm:actionable_rankece} also hold for $\lECE(f)$ and the related measure of \textit{cutoff calibration} introduced in \citep{rossellini2025can}. In particular, the guarantee for $\lECE(f)$ does not require any restrictions on the class of transformation functions $h$. Similar guarantees have also been proposed in terms of regret through the lens of omniprediction in \cite{okoroafor2025nearoptimal}.

\begin{remark}
    The actionability guarantees in Theorem \ref{thm:actionable_rankece} are stated in terms of the risk under the loss function $\ell$ in \eqref{eq:tau_binary_loss}. However, analogous guarantees can be established for arbitrary bounded proper scoring rules, following the argument in Appendix B.3 of \cite{rossellini2025can}. This follows from the Schervish representation \citep{schervish1989general}, which expresses every proper scoring rule as a mixture of binary decision losses over different values of $\tau$. We omit the details for brevity.
\end{remark}

\subsection{Proof of Theorem \ref{thm:actionable_rankece}}\label{sec:proofof_thm_actionable_rankece}
To prove Theorem \ref{thm:actionable_rankece} notice that it is enough to prove that for any $h\in \cH$,
\begin{align*}
    \cR(f;\tau) - \cR(h\circ f;\tau)\leq \sqrt{\rankECE(f) + \frac{8(M+1)}{n}}.
\end{align*}
To that end fix $h\in \cH$. Now by Lemma 3.1 from \cite{rossellini2025can} we have the decomposition,
\begin{align*}
    \cR(f;\tau) - \cR(h\circ f; \tau) 
    & = \E\left[(Y-\tau)\one\{Z<\tau, h(Z)\geq \tau\}\right] + \E\left[(\tau-Y)\one\{Z\geq \tau, h(Z)< \tau\}\right]\\
    & \leq \E[(Y-Z)\omega_\tau(Z)]
\end{align*}
where recall that $Z = f(X)$ and $\omega_\tau(\cdot)$ is the weight function defined as,
\begin{align}\label{eq:def_omega}
    \omega_\tau(z) = \one\{z<\tau, h(z)\geq \tau\} - \one\{z\geq\tau, h(z)<\tau\}\text{ for all }z\in [0,1].
\end{align}
Then by the tower property of conditional expectations, $\E[(Y-Z)\omega_\tau(Z)] = \E[r(Z)\omega_\tau(Z)]$, where $r(\cdot)$ is the residual function from Proposition \ref{prop:rank_l2_finite_bd}. Thus we have the bound,
\begin{align*}
    \cR(f;\tau) - \cR(h\circ f; \tau) \leq \E[r(Z)\omega_\tau(Z)].
\end{align*}
In the next result we claim that there exists a partition of $[0,1]$ such that $\omega_\tau$ is piecewise constant on that partition. 

\begin{lemma}\label{lemma:partition_omega}
    There exists an integer $K_\tau \leq 2(M+1)$ and a partition $\cP_{K_\tau} = \{J_1,\ldots,J_{K_\tau}\}$ of $[0,1]$ consisting of subintervals $J_t$ such that $\omega_\tau$ is constant on each $J_t$, $1 \leq t \leq K_\tau$.
\end{lemma}

We postpone the proof of Lemma \ref{lemma:partition_omega} to Section \ref{appendix:proofof_partition_omega}. In the following we use this result to complete the proof of Theorem \ref{thm:actionable_rankece}. Considering the partition $\cP_{K_\tau}$ from Lemma \ref{lemma:partition_omega} we can write,
\begin{align*}
    \E[r(Z)\omega_\tau(Z)] = \sum_{k=1}^{K_\tau}p_k\omega_k\mu_k
\end{align*}
where $p_k = \P(Z\in J_k)$, $\mu_k = \E[r(Z)|Z\in J_k]$ and $\omega_k\in \{-1, 0,1\}$ is the constant value of $\omega_{\tau}$ on $J_k$. By Cauchy Schwartz inequality it follows that,
\begin{align}\label{eq:actionable_bin_bd}
    \left|\E[r(Z)\omega_\tau(Z)]\right|\leq \sqrt{\sum_{k=1}^{K_\tau}p_k\omega_k^2}\sqrt{\sum_{k=1}^{K_\tau}p_k\mu_k^2}\leq \sqrt{\sum_{k=1}^{K_\tau}p_k\mu_k^2},
\end{align}
where the last inequality follows by noting that $\omega_k^2\in \{0,1\}$ and $\sum_{k=1}^{K_\tau}p_k = 1$. Now note that $\sum_{k=1}^{K_\tau}p_k\mu_k^2$ is the $\lbinECE$ of $f$ with respect to the partition $\cP_{K_\tau}$. With this observation we can now repeat the proof of Theorem \ref{thm:rank_bin_comparison} with the choice of bins $\cP_{K_\tau}$ and conclude,
\begin{align}\label{eq:actionable_bin_rank_bd}
    \sum_{k=1}^{K_\tau}p_k\mu_k^2\leq \rankECE(f) + \frac{4K_\tau}{n}\leq \rankECE(f) + \frac{8(M+1)}{n}.
\end{align}
The proof of Theorem \ref{thm:actionable_rankece} is now completed by substituting the bound from \eqref{eq:actionable_bin_rank_bd} in the right hand side of \eqref{eq:actionable_bin_bd}.

\subsubsection{Proof of Lemma \ref{lemma:partition_omega}}\label{appendix:proofof_partition_omega}
Consider the function $\phi_\tau(z) = \one\{h(z)\geq \tau\}$ and define,
\begin{align*}
    C_h(\tau) = \sup\left\{\sum_{i=0}^{m}\one\{\phi_\tau(z_i)\neq \phi_\tau(z_{i+1})\}: m\geq 1, z_0 = 0\leq z_1\leq\cdots\leq z_m\leq z_{m+1} = 1\right\}.
\end{align*}
By definition $C_h(\tau)$ counts the number of times $\phi_\tau$ changes values in $[0,1]$. Note that $C_h(\tau)\in \N\cup\{+\infty\}$. Now recall that $h$ is piecewise monotone and consider the partition $0 = a_0<a_1<\cdots<a_M=1$ such that $h$ is monotone in each sub-interval $[a_j, a_{j+1}]$ for all $0\leq j\leq M-1$. Now fix a collection of points $z_0 = 0\leq z_1\leq\cdots\leq z_m\leq z_{m+1} = 1$ and notice,
\begin{align*}
    \sum_{i=0}^{m}\one\{\phi_\tau(z_i)\neq \phi_\tau(z_{i+1})\}\leq \sum_{j=1}^{M}(S_j + 1)
\end{align*}
where $S_j = \sum_{i\in I_j}\one\{\phi_\tau(z_i)\neq \phi_\tau(z_{i+1})\}$ with $I_j = \{i\in [m]: z_i\in [a_{j-1}, a_j)\}$ for all $1\leq j\leq M$. Fix $j\in [M]$. Then recall that by definition $h$ is monotone in $[a_{j-1}, a_j]$ and hence $S_j\leq 1$. Thus,
\begin{align*}
    \sum_{i=0}^{m}\one\{\phi_\tau(z_i)\neq \phi_\tau(z_{i+1})\}\leq 2M.
\end{align*}
Moreover, recall that the collection of points $z_0 = 0\leq z_1\leq\cdots\leq z_m\leq z_{m+1} = 1$ was arbitrary and hence it follows that $C_h(\tau)\leq 2M$. Now to complete the proof of Lemma \ref{lemma:partition_omega} note that by definition from \eqref{eq:def_omega} the function $\omega_\tau$ can change value only at points where either $\one\{z\geq \tau\}$ or $\phi_\tau(z) = \one\{h(z)\geq \tau\}$ changes value. We can find a partition $\{\bar J_1,\bar J_2,\ldots, \bar J_{C_h(\tau)+1}\}$ of $[0,1]$ where each $\bar J_t$ is a sub-interval of $[0,1]$ such that $\phi_\tau$ is constant in each $\bar J_t$ for all $t\in [C_h(\tau) + 1]$. Moreover, the function $\one\{z\geq \tau\}$ has a single jump at $z = \tau$ and hence we can refine the above partition to consider $\cP_{K_\tau} := \{J_1,J_2,\ldots, J_{K_\tau}\}$ with $K_\tau = C_h(\tau) + 2$ such that $\omega_\tau$ is constant on each $J_t$. The proof is completed by noting that $K_\tau\leq 2(M+1)$.

\section{An example comparing $\rankECE$ and $\lbinECE$}\label{appendix:example}

In this section, we formally discuss the setting of Example \ref{example:compare} to compare the proposed measure of calibration $\rankECE$ with $\lbinECE$ (recall \eqref{eq:L2_binned_ECE}). To that end, we begin by recalling the notation. Fix $c\in (0,1/2)$ and, for $m\in \N$, let $\phi_m(z) = \textnormal{sign}(\sin(2\pi mz))$ be a square wave on $[0,1]$. The data-generating process is as follows:
\begin{align*}
    X\sim \textnormal{Unif}[0,1] \text{ and }Y\mid X\sim \textnormal{Ber}(c\phi_m(X) + (1-c)X + c(1-X)).
\end{align*}
We take the prediction function to be the map
\begin{align}\label{eq:example_f}
    Z = f(X) = c(1-X) + (1-c)X = c + (1-2c)X.
\end{align}
With this setting in place, we compare $\rankECE(f)$ and $\lbinECE(f)$ as measures of calibration. To that end we first evaluate $\lECE(f)$. By construction,
\begin{align*}
    \E[Y\mid Z] = Z + c\phi_m\left(\frac{Z - c}{1-2c}\right)
\end{align*}
where $|\phi_m| = 1$. Hence, by definition (recall \eqref{eq:L2_ECE}), it immediately follows that $\lECE(f) = c^2$. Next, in our first result, we analyze $\lbinECE$ as a measure of calibration in this example. For $K\geq 1$ and $0=a_0<a_1<\cdots<a_{K-1}<a_K=1$, let
\begin{align}\label{eq:def_bins_example}
\cB_1=[a_0,a_1),\ \cB_2=[a_1,a_2),\ \ldots,\ \cB_K=[a_{K-1},a_K]
\end{align}
be an arbitrary partition of $[0,1]$ into $K$ bins.
\begin{lemma}\label{lemma:example_bin_bd}
Consider the setting defined above and let $f(x) = c + (1-2c)x$ be the prediction function from \eqref{eq:example_f}. Define $\lbinECE(f)$ with respect to the bins $\cB_1,\ldots,\cB_K$ from \eqref{eq:def_bins_example}. Then, for $m\geq 1$,
\begin{align*}
\lbinECE(f) \leq \frac{c^2K}{m}.
\end{align*}
\end{lemma}

We present the proof of Lemma \ref{lemma:example_bin_bd} in Appendix \ref{sec:proofof_lemma_example_bin_bd}. The result of Lemma \ref{lemma:example_bin_bd} confirms that when the residual function $r(Z) = \E[Y-Z\mid Z] = c\phi_m\left(\frac{Z - c}{1-2c}\right)$ exhibits high oscillation, specifically when the oscillation frequency satisfies $m\gg K$, then $\lbinECE(f)$ fails to capture the miscalibration (recall that $\lECE(f) = c^2>0$). In contrast, the next result shows that, in the same example, $\rankECE$ provides a more robust alternative for detecting miscalibration.

\begin{lemma}\label{lemma:example_rank_bd}
    Consider the setting defined above. Then for the prediction function $f(x) = c + (1-2c)x$ from \eqref{eq:example_f},
    \begin{align*}
        c^2 - \frac{4mc+1}{n}\leq \rankECE(f)\leq c^2 \min\left\{ 1, \frac{(n-1)^2}{2m^2}\right\}
    \end{align*}
    for all $n\geq 2$ and $m\geq 1$.
\end{lemma}

We defer the proof of the above result to Appendix \ref{appendix:proofof_lemma_example_rank_bd}. Indeed, by Lemma \ref{lemma:example_rank_bd}, we see that $\rankECE$ is able to detect miscalibration when $m\ll n$. In particular, Lemma \ref{lemma:example_bin_bd} shows that for $\lbinECE$, for which the optimal choice of the number of bins is almost always $o(n)$, one must choose the number of bins carefully and in a way that adapts to the oscillations of the (unknown) underlying residual function $r$. By contrast, $\rankECE$ adapts to these oscillations and detects miscalibration whenever they are reasonably small compared to the sample size. Of course, Lemma \ref{lemma:example_rank_bd} also shows that $\rankECE$ fails to detect miscalibration when $m\gg n$, but this is not surprising, since $O(1/n)$ is the finest resolution that can be resolved from $n$ observations.

\begin{figure}[ht]
    \centering

    \begin{subfigure}{0.48\textwidth}
        \centering
        \includegraphics[width=\linewidth]{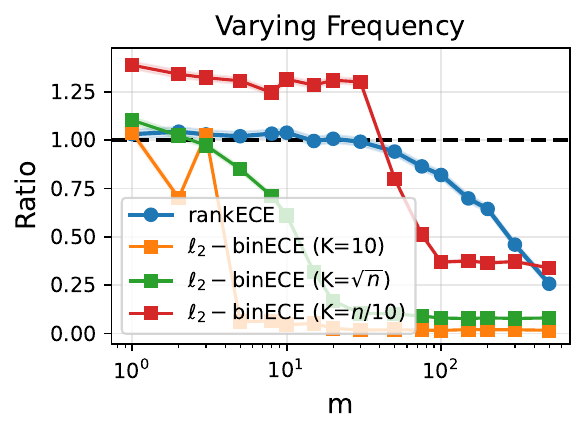}
        \label{fig:first}
    \end{subfigure}
    \hfill
    \begin{subfigure}{0.48\textwidth}
        \centering
        \includegraphics[width=\linewidth]{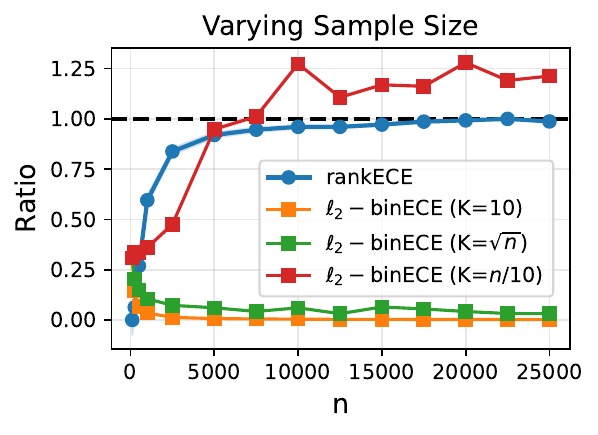}
        \label{fig:second}
    \end{subfigure}

    \caption{\small Empirical comparison of $\rankECE(f)$ and $\lbinECE(f)$ in the setting of Appendix~\ref{appendix:example} (see also Example~\ref{example:compare}) with $c=0.2$. The vertical axes show $\hatrankECE(f)/\lECE(f)$ and $\binhatECE(f)/\lECE(f)$ for $K\in\{10,\sqrt{n},n/10\}$. (Left) Comparison as the frequency $m$ increases, with $n=2000$ fixed. (Right) Comparison as the sample size $n$ increases, with $m=100$ fixed.}
    \label{fig:example_compare}
\end{figure}
\subsection{Empirical Comparison of $\lbinECE$ and $\rankECE$ in Example \ref{example:compare}}
We validate the above comparisons empirically in Figure \ref{fig:example_compare}. Throughout, recall that $f$ is the prediction function from \eqref{eq:example_f}. We fix $c=0.2$ and plot the ratios $\rankECE(f)/\lECE(f)$ and $\lbinECE(f)/\lECE(f)$, where $\rankECE(f)$ and $\lbinECE(f)$ are estimated using their empirical counterparts, $\hatrankECE(f)$ and $\binhatECE(f)$, respectively. Reported values are averaged over $50$ repetitions. For the binning estimator, we consider three choices for the number of bins: $K=10$, $K=\sqrt{n}$, and $K=n/10$.

In Figure \ref{fig:example_compare} (a), we fix $n=2000$ and increase the oscillation frequency $m$. In contrast, Figure \ref{fig:example_compare} (b) fixes the oscillation frequency at $m=100$ and increases the sample size $n$. Figure \ref{fig:example_compare} (a) shows that $\rankECE(f)$ provides a substantially better approximation to $\lECE(f)$ for moderate values of $m$. As $m$ increases, however, the quality of the approximation gradually deteriorates, consistent with the prediction of Lemma \ref{lemma:example_rank_bd}. By contrast, Figure \ref{fig:example_compare} (b) shows that, for a fixed oscillation frequency, $\rankECE(f)$ rapidly becomes a close approximation to $\lECE(f)$ as $n$ increases, whereas the approximations produced by $\lbinECE(f)$ remain noticeably worse even for large sample sizes. The performance of $\lbinECE(f)$ with $K=10$ and $K = \sqrt{n}$ remains consistently poor, while the estimator with $K=n/10$ exhibits a persistent bias, as predicted by \cite{futami2024information} and discussed in Section \ref{sec:compare}.

\subsection{Proof of Lemma \ref{lemma:example_bin_bd}}\label{sec:proofof_lemma_example_bin_bd}
By definition from \eqref{eq:L2_binned_ECE} and recalling that here the prediction function is the map $Z = f(X) = (1-c)X + c(1-X)$,
\begin{align}\label{eq:example_binece_decomp}
\lbinECE(f) = \sum_{j=1}^{K}p_j\mu_j^2
\end{align}
where for all $j\in [K]$, $p_j = \P(Z\in \cB_j)$ and $\mu_j = \E[r(Z)\mid Z\in \cB_j]$ with $r(Z) = \E[Y-Z\mid Z] = c\phi_m\left(\frac{Z-c}{1-2c}\right)$ denoting the residual function. Moreover $Z\sim \textnormal{Unif}[c,1-c]$ and,
\begin{align*}
\mu_j = \frac{\E\left[r(Z)\one\{Z\in \cB_j\}\right]}{\P(Z\in \cB_j)} = \frac{c}{(1-2c)p_j}\int_{\cB_j}\phi_m\left(\frac{z-c}{1-2c}\right)\rmd z = \frac{c}{p_j}\int_{\bar{\cB}_j}\phi_m(x)\rmd x
\end{align*}
where $\bar{\cB}_j = \left[\frac{a_{j-1} - c}{1-2c}, \frac{a_j-c}{1-2c}\right)$. We can decompose $|\bar{\cB}_j| = q\frac{1}{m} + s$ for some $q\in\N$ and $s\in [0, 1/m)$. Since the function $\phi_m$ is a square wave with period $1/m$, it follows from the above decomposition that there exists an interval $R_j\subseteq \cB_j$ such that $|R_j| = s$ and
\begin{align*}
|\mu_j| = \left|\frac{c}{p_j}\int_{R_j}\phi_m(x)\rmd x\right|\leq \min\left\{\frac{c}{p_j}|R_j|, c\right\}\leq \min\left\{\frac{c}{mp_j}, c\right\}
\end{align*}
where the first inequality follows by noticing that $|r|\leq c$. Substituting the upper bound in \eqref{eq:example_binece_decomp},
\begin{align*}
    \lbinECE(f)\leq \sum_{j=1}^{K}\min\left\{\frac{c^2}{m^2p_j}, c^2p_j\right\} 
    & = \frac{c^2}{m}\sum_{j=1}^{K}\min\left\{\frac{1}{mp_j}, mp_j\right\}\\
    &\leq \frac{c^2K}{m}
\end{align*}
which completes the proof.

\subsection{Proof of Lemma \ref{lemma:example_rank_bd}}\label{appendix:proofof_lemma_example_rank_bd}
We divide the proof in two parts for the upper and lower bounds. We begin the proof with the lower bound and present the proof of the upper bound later. 

\paragraph{\textbf{Proof for the lower bound.}} The lower bound follows by recalling that $\lECE(f) = c^2$ and applying Proposition \ref{prop:rank_l2_finite_bd}. Indeed, to complete the proof it is enough to show that,
\begin{align*}
    \|r\|_{\textnormal{TV}}\leq 4mc,
\end{align*}
where $r(z) = \E[Y-Z\mid Z = z] = c\phi_m\left(\frac{z-c}{1-2c}\right)$ is the residual function. To that end, notice that the domain of $r$ can be restricted to $[c, 1-c]$. The total variation of $r$ (recall \eqref{eq:def_TV}) is defined as,
\begin{align*}
    \|r\|_{\textnormal{TV}} = \sup\left\{\sum_{i = 0}^{j-1}\left|r(t_{i+1}) - r(t_i)\right|: c = t_0<t_1<\cdots<t_j = 1-c, j\geq 1\right\}.
\end{align*}
Since $r$ is a square wave, the above supremum is attained by choosing the partition points to be the consecutive sign-change points. By definition, $r(z)=c\phi_m\left(\frac{z-c}{1-2c}\right)$, and hence sign changes can occur only at the zeros of $\phi_m$, which are given by $t/2m$ for all $t\in\mathbb{Z}$. Since $\frac{z-c}{1-2c}\in[0,1]$ for $z\in[c,1-c]$, the number of zeros of $\phi_m$ in $[0,1]$, and therefore the number of sign changes of $r$, is bounded above by $2m$. It follows immediately that,
\begin{align*}
    \|r\|_{\textnormal{TV}}\leq 2m\cdot 2c = 4mc.
\end{align*}

\paragraph{\textbf{Proof of the upper bound.}}
First, we have $\rankECE(f) \leq \lECE(f) = c^2$, which verifies the first term of the minimum. Now we prove that the second term also provides an upper bound.

We begin by recalling the expression of $\rankECE(f)$ from \eqref{eq:rank_ECE}. Take observations $\{(Y_i, Z_i): 1\leq i\leq n\}$ where $Z_i = f(X_i)$ with $f$ defined in \eqref{eq:example_f}. Then by tower property of conditional expectations (also see \eqref{eq:rank_r_equal})
\begin{align*}
    \rankECE(f) = \frac{1}{n}\sum_{i=1}^{n-1}\E\left[r\left(Z_{\pi(i)}\right)r\left(Z_{\pi(i+1)}\right)\right] = \frac{c^2}{n}\sum_{i=1}^{n-1}\E\left[\phi_m\left(\frac{Z_{\pi(i)}-c}{1-2c}\right)\phi_m\left(\frac{Z_{\pi(i+1)}-c}{1-2c}\right)\right]
\end{align*}
where $\pi$ is the permutation from \eqref{eq:rank_ece_sample}. By definition of $f$ from \eqref{eq:example_f} note that,
\begin{align*}
    X_j = \frac{Z_j-c}{1-2c}\text{ for all }j \in [n],
\end{align*}
where $X_1,\ldots, X_n$ are independent samples from $\textnormal{Unif}[0,1]$. Hence,
\begin{align*}
    \rankECE(f) = \frac{c^2}{n}\sum_{i=1}^{n-1}\E\left[\phi_m\left(X_{(i)}\right)\phi_m\left(X_{(i+1)}\right)\right]
\end{align*}
where $X_{(1)}\leq X_{(2)}\leq \cdots\leq X_{(n)}$ are the order statistics of $X_1, X_2,\ldots, X_n$. By the tower property of conditional expectations,
\begin{align}\label{eq:example_ub_tower}
    \E\left[\phi_m\left(X_{(i)}\right)\phi_m\left(X_{(i+1)}\right)\right] = \E\left[\E\left[\phi_m\left(X_{(i)}\right)\phi_m\left(X_{(i+1)}\right) \mid X_{(i-1)}, X_{(i+2)}\right]\right]
\end{align}
with the boundary convention $X_{(0)} = 0$ and $X_{(n+1)} = 1$. Fix $1\leq i\leq n-1$. Applying Lemma 1 from \cite{gupta2021distribution} we know that,
\begin{align*}
    X_{(i)}, X_{(i+1)}\mid X_{(i-1)}, X_{(i+2)} \overset{d}{=} W_{(1)}, W_{(2)}
\end{align*}
where $W_{(1)}\leq W_{(2)}$ are the order statistics of $W_1, W_2\sim \textnormal{Unif}[X_{(i-1)}, X_{(i+2)}]$. Hence
\begin{align}\label{eq:example_cond_exp_bd_1}
    \E\left[\phi_m\left(X_{(i)}\right)\phi_m\left(X_{(i+1)}\right) \mid X_{(i-1)}, X_{(i+2)}\right] 
    & = \E\left[\phi_m(W_{(1)})\phi_m(W_{(2)})\right] = \E\left[\phi_m(W_1)\right]^2\nonumber\\
    & = \left(\frac{1}{X_{(i+2)} - X_{(i-1)}}\int_{X_{(i-1)}}^{X_{(i+2)}}\phi_m(t)\rmd t\right)^2.
\end{align}
To further upper bound the right hand side, recall that by construction $\phi_m$ is a square wave with period $1/m$. Hence there exists an interval $I\subseteq [X_{(i-1)}, X_{(i+2)}]$ such that $|I|<1/m$ and,
\begin{align*}
    \left|\int_{X_{(i-1)}}^{X_{(i+2)}}\phi_m(t)\rmd t\right| = \left|\int_{I}\phi_m(t)\rmd t\right|\leq \frac{1}{m}.
\end{align*}
Substituting this upper bound in \eqref{eq:example_cond_exp_bd_1} along with \eqref{eq:example_ub_tower},
\begin{align}\label{eq:example_ub_tower_2}
    \rankECE(f)\leq \frac{c^2}{nm^2}\sum_{i=1}^{n-1}\E\left[\frac{1}{(X_{(i+2)} - X_{(i-1)})^2}\right].
\end{align}
To complete the proof, recall that $X_1,\ldots, X_n\sim\textnormal{Unif}[0,1]$ implying $X_{(i+2)} - X_{(i-1)}\sim \textnormal{Beta}(3, n-2)$. Hence,
\begin{align*}
    \E\left[\frac{1}{(X_{(i+2)} - X_{(i-1)})^2}\right] = \frac{n(n-1)}{2}.
\end{align*}
The proof is now completed by substituting the above identity in \eqref{eq:example_ub_tower_2}.

\section{Proof of Proposition \ref{prop:test_rank_finite} and results from Section \ref{sec:rank_prop}}

\subsection{Proof of Proposition \ref{prop:test_rank_finite}}\label{sec:proofof_prop_test_rank_finite}
The proof proceeds by a careful application of McDiarmid's bounded difference inequality. To that end, for pairs $\bw_i = (y_i, z_i)\in \{0,1\}\times [0,1]$ we define
\begin{align*}
    g\left(\bw_1,\ldots, \bw_n\right) = \frac{1}{n}\sum_{i=1}^{n-1}\left(y_{\pi(i)} - z_{\pi(i)}\right)\left(y_{\pi(i+1)} - z_{\pi(i+1)}\right),
\end{align*}
where $\pi:[n]\ra[n]$ is a permutation such that $z_{\pi(1)}\leq z_{\pi(2)}\leq\cdots\leq z_{\pi(n)}$. We can observe that
\begin{align}\label{eq:Rnf=g}
    \hatrankECE(f) = g\left(\bW_1,\ldots,\bW_i,\ldots, \bW_n\right)
\end{align}
where $\bW_i = (Y_i, Z_i)$ for all $1\leq i\leq n$. To apply McDiarmid's bounded difference inequality, consider a new pair $\bw_i^\prime\in \{0,1\}\times [0,1]$. Then, by Lemma~\ref{lem:g_diff_mcdiarmid} below,
\begin{align*}
    \left|g\left(\bw_1,\ldots,\bw_i,\ldots, \bw_n\right) - g\left(\bw_1,\ldots,\bw_i^\prime,\ldots, \bw_n\right)\right|\leq \frac{2.25}{n}.
\end{align*}
Hence, an application of McDiarmid's bounded difference inequality now shows that for all $t>0$,
\begin{align*}
    \P\left(\left|g\left(\bW_1,\ldots,\bW_i,\ldots, \bW_n\right) - \E\left[g\left(\bW_1,\ldots,\bW_i,\ldots, \bW_n\right)\right]\right|>t\right)\leq 2\exp\left(-\frac{32}{81}nt^2\right).
\end{align*}
The proof is now completed by recalling the identity from \eqref{eq:Rnf=g} and by definition $\rankECE(f) = \E[\hatrankECE(f)]$.

\begin{lemma}\label{lem:g_diff_mcdiarmid}
    Let $g$ be defined as above. Then, for any $i\in[n]$ and any $\bw_1,\dots,\bw_n,\bw'_i\in\{0,1\}\times[0,1]$,
    \[\left|g\left(\bw_1,\ldots,\bw_i,\ldots, \bw_n\right) - g\left(\bw_1,\ldots,\bw_i^\prime,\ldots, \bw_n\right)\right|\leq \frac{2.25}{n}.\]
\end{lemma}
\begin{proof}
    Without loss of generality we can assume $z_1 \leq  \dots \leq z_n$. Find index $j$ such that $z_j\leq z'_i \leq z_{j+1}$. (Recall that we are assuming no ties among $Z$ values, almost surely, throughout for simplicity.) Then
    \begin{multline*}
    g\left(\bw_1,\ldots,\bw_i,\ldots, \bw_n\right) - g\left(\bw_1,\ldots,\bw_i^\prime,\ldots, \bw_n\right)\\
        = \left[(y_{i-1}-z_{i-1})(y_i-z_i) + (y_i-z_i)(y_{i+1}-z_{i+1}) - (y_{i-1}-z_{i-1})(y_{i+1}-z_{i+1})\right]\\
        - \left[(y_j-z_j)(y'_i-z'_i) + (y'_i - z'_i)(y_{j+1}-z_{j+1}) - (y_j-z_j)(y_{j+1}-z_{j+1})\right].
    \end{multline*}
    Finally, for each of the two terms in square brackets, the value must lie in $[-1.25,1]$, since it holds that
    \[-1.25\leq ab + bc - ac\leq 1\textnormal{ for all $a,b,c\in[-1,1]$ with $c-1\leq b\leq a+1$}\]
    (where for instance, if we take $a=y_{i-1}-z_{i-1}$, $b=y_i-z_i$, and $c=y_{i+1}-z_{i+1}$, we must have $c-1\leq b\leq a+1$ since $z_{i-1}\leq z_i\leq z_{i+1}$).
\end{proof}

\subsection{Proof of Proposition \ref{prop:lower_bound}}\label{sec:proofof_prop_lower_bound}
Recall the residual function $r(z) = \E[Y-Z\mid Z = z]$ defined in Proposition \ref{prop:rank_l2_finite_bd}. With this notation in place, we now proceed to prove the upper bound.

\paragraph{\textbf{Proof of upper bound.}}
Note that, conditioning on the order statistics $Z_{\pi(1)},Z_{\pi(2)},\ldots,Z_{\pi(n)}$ and using the conditional independence of concomitant pairs, one immediately obtains
\begin{align}
    \rankECE(f) 
    & = \E\left[\frac{1}{n}\sum_{i=1}^{n-1}r\left(Z_{\pi(i)}\right)r\left(Z_{\pi(i+1)}\right)\right]\label{eq:rank_r_equal}\\
    &\leq \E\left[\frac{1}{n}\sum_{i=1}^{n-1}\left|r\left(Z_{\pi(i)}\right)\right|\left|r\left(Z_{\pi(i+1)}\right)\right| + \frac{1}{n}\left|r(Z_{\pi(n)})\right|\left|r(Z_{\pi(1)})\right|\right]\nonumber\\
    &\leq \E\left[\frac{1}{n}\sum_{i=1}^{n}r^2(Z_i)\right]\nonumber,
\end{align}
where the final inequality follows from the rearrangement inequality \citep{hardy1952inequalities}. The proof of the upper bound is now completed by observing that $\E[r^2(Z)] = \lECE(f)$. 

\paragraph{\textbf{Proof of lower bound.}}
To being with we define the hypothetical order statistics $Z_{\pi(0)} = 0$ and $Z_{\pi(n+1)} = 1$ to make notations easier. Then using the equality from \eqref{eq:rank_r_equal} and tower property of conditional expectations we get,
\begin{align}\label{eq:rank_lb_1}
    \rankECE(f) = \frac{1}{n}\sum_{i=1}^{n-1}\E\left[\E\left[r\left(Z_{\pi(i)}\right)r\left(Z_{\pi(i+1)}\right)\mid Z_{\pi(i-1)}, Z_{\pi(i+2)}\right]\right].
\end{align}
Fix $i\in [n-1]$. Then note that conditional on $Z_{\pi(i-1)}$ and $Z_{\pi(i+2)}$ the pair $\left(Z_{\pi(i)}, Z_{\pi(i+1)}\right)$ is distributed as the order statistics of two samples $\bar Z_1,\bar Z_2$ generated independently from the conditional distribution $P_{Z\mid Z\in \left(Z_{\pi(i-1)}, Z_{\pi(i+2)}\right)}$ (see Lemma 1 from \cite{gupta2021distribution}). Then,
\begin{align}\label{eq:rank_lb_2}
    \E\left[r\left(Z_{\pi(i)}\right)r\left(Z_{\pi(i+1)}\right)\mid Z_{\pi(i-1)}, Z_{\pi(i+2)}\right] = \E\left[r\left(\bar Z_{(1)}\right)r\left(\bar Z_{(2)}\right)\right]
\end{align}
where $\bar Z_{(1)}\leq \bar Z_{(2)}$ are the order statistics of $\bar Z_1, \bar Z_2$. Next, recall that $\bar Z_1,\bar Z_2$ are independent samples from $P_{Z\mid Z\in \left(Z_{\pi(i-1)}, Z_{\pi(i+2)}\right)}$. Hence,
\begin{align}\label{eq:rank_lb_3}
    \E\left[r\left(\bar Z_{(1)}\right)r\left(\bar Z_{(2)}\right)\right] = \E\left[r(\bar Z_1)r(\bar Z_2)\right] = \E[r(\bar Z_1)]^2\geq 0.
\end{align}
The proof is now completed by combining the conclusions from \eqref{eq:rank_lb_1}, \eqref{eq:rank_lb_2}, \eqref{eq:rank_lb_3} and recalling that $i\in [n-1]$ was chosen arbitrarily.

\subsection{Proof of Proposition \ref{prop:consistency}}\label{sec:proofof_prop_consistency}
By the concentration result from Proposition \ref{prop:test_rank_finite}, it is enough to show that
\begin{align*}
    \rankECE(f)\ra \lECE(f).
\end{align*}
Recalling the residual function $r(z) = \E[Y-Z\mid Z = z]$ and the identity in \eqref{eq:rank_r_equal}, we know that
\begin{align*}
    \rankECE(f) = \E\left[\frac{1}{n}\sum_{i=1}^{n-1}r\left(Z_{\pi(i)}\right)r\left(Z_{\pi(i+1)}\right)\right].
\end{align*}
The proof now proceeds by replacing the residual function $r$ with a continuous approximation. To that end, fix $\vep>0$. Applying Lusin's theorem together with the Tietze extension theorem \citep{dugundji1951extension,mcshane1934extension}, there exists a continuous function $\bar r_\vep$ and a compact set $A_\vep\subseteq[0,1]$ such that $r = \bar r_\vep$ on $A_\vep$ and $\P(Z\in [0,1]\setminus A_\vep)<\vep$. Moreover, since the interval $[0,1]$ is compact, $\bar r_\vep$ is uniformly continuous. By a simple decomposition,
\begin{align}\label{eq:rankECE_decomp_consistency}
    \rankECE(f) = \E\left[r(Z)^2\right] + \underbrace{\frac{1}{n}\E\left[\sum_{i=1}^{n-1}r\left(Z_{\pi(i)}\right)\left(r\left(Z_{\pi(i+1)}\right) - r\left(Z_{\pi(i)}\right)\right)\right]}_{S_n} + O\left(\frac{1}{n}\right).
\end{align}
Note that $\E[r(Z)^2] = \lECE(f)$, and hence to complete the proof it suffices to control the second term $S_n$. To that end, using the fact that $|r|\leq 1$, we obtain the upper bound
\begin{align}\label{eq:S_n_bd_1}
    |S_n|
    & \leq \frac{1}{n}\E\left[\sum_{i=1}^{n-1}\left|r\left(Z_{\pi(i+1)}\right) - r\left(Z_{\pi(i)}\right)\right|\right].
\end{align}
Now to further decompose the above upper bound consider the collection of indices $\cA_\vep\subseteq [n-1]$ such that,
\begin{align*}
    \cA_\vep = \left\{i\in [n-1]: Z_{\pi(i)},Z_{\pi(i+1)}\in A_\vep\right\}.
\end{align*}
With the above notation we can write \eqref{eq:S_n_bd_1} as,
\begin{align*}
    |S_n|
    &\leq \frac{1}{n}\E\left[\sum_{i\in \cA_\vep}\left|r\left(Z_{\pi(i+1)}\right) - r\left(Z_{\pi(i)}\right)\right| + \sum_{i\not\in\cA_\vep}\left|r\left(Z_{\pi(i+1)}\right) - r\left(Z_{\pi(i)}\right)\right|\right]\\
    &\leq \frac{1}{n}\E\left[\sum_{i\in \cA_\vep}\left|\bar r_\vep\left(Z_{\pi(i+1)}\right) - \bar r_\vep\left(Z_{\pi(i)}\right)\right| + \left|\cA_\vep^c\right|\right]
\end{align*}
where $|\cA^c_\vep|$ denotes the number of elements in $\cA^c_\vep$ and the last inequality follows by recalling that $|r|\leq 1$, the definition of $\cA_\vep$ and noting that $r = \bar r_\vep$ on $A_\vep$. Next consider the set $\cN_\vep = \{i\in [n]: Z_i\in A_\vep^c\}$. Then recalling that $\pi$ is a permutation of $[n]$, it immediately follows that,
\begin{align*}
    |\cA_\vep^c|\leq 2|\cN_\vep|.
\end{align*}
Hence,
\begin{align}\label{eq:S_n_bd_Sn_bar}
    |S_n|
    &\leq \underbrace{\frac{1}{n}\E\left[\sum_{i=1}^{n-1}\left|\bar r_\vep\left(Z_{\pi(i+1)}\right) - \bar r_\vep\left(Z_{\pi(i)}\right)\right|\right]}_{\bar S_n} + \vep
\end{align}
To further bound the right-hand side, we now leverage the uniform continuity of $\bar r_\vep$. By uniform continuity, given $\vep>0$, there exists $\delta>0$ such that
\begin{align*}
    \left|\bar r_\vep(x) - \bar r_\vep(y)\right|< \vep\textnormal{ whenever }|x-y|< \delta.
\end{align*}
Define the set of indices
\[
\Delta = \left\{i\in [n-1]: \left|Z_{\pi(i+1)} - Z_{\pi(i)}\right|<\delta\right\}.
\]
Then for any $i\in \Delta$, by the uniform continuity of $\bar r_\vep$ we immediately have
\begin{align*}
    \left|\bar r_\vep\left(Z_{\pi(i+1)}\right) - \bar r_\vep\left(Z_{\pi(i)}\right)\right|<\vep.
\end{align*}
Hence we may upper bound $\bar S_n$ as
\begin{align}\label{eq:upper_bd_barSn_1}
    \bar S_n \leq \vep + \E\left[\frac{1}{n}\sum_{i\not\in \Delta}\left|\bar r_\vep\left(Z_{\pi(i+1)}\right) - \bar r_\vep\left(Z_{\pi(i)}\right)\right|\right].
\end{align}
To complete the proof, we now bound the cardinality of $\{i\in [n-1]:i\not\in \Delta\}$. Observe that
\begin{align*}
    \sum_{i=1}^{n-1}\left|Z_{\pi(i+1)} - Z_{\pi(i)}\right| = \sum_{i=1}^{n-1}Z_{\pi(i+1)} - Z_{\pi(i)}\leq 1.
\end{align*}
Hence it follows that
\begin{align*}
    \left|\{i\in [n-1]:i\not\in \Delta\}\right| = |\{i\in [n-1]: \left|Z_{\pi(i+1)} - Z_{\pi(i)}\right|\geq \delta\}|\leq 1/\delta.
\end{align*}
Recall that the existence of $\bar r_\vep$ was justified using the Tietze extension theorem, and hence $\|\bar r_\vep\|_\infty\leq 1$. Therefore, recalling \eqref{eq:upper_bd_barSn_1}, we obtain the upper bound
\begin{align*}
    \bar S_n \leq \vep + \frac{1}{n\delta}. 
\end{align*}
The proof is now completed by recalling the identity $\E[r(Z)^2] = \lECE(f)$, together with the identities and bounds from \eqref{eq:rankECE_decomp_consistency}, \eqref{eq:S_n_bd_Sn_bar}, and \eqref{eq:upper_bd_barSn_1}, and observing that $\vep>0$ was chosen arbitrarily.

\subsection{Proof of Proposition \ref{prop:rank_l2_finite_bd}}\label{sec:proofof_prop_rank_l2_finite_bd}
The bound trivially holds true if $\|r\|_{\textnormal{TV}} = \infty$. Hence for the rest of the proof we assume $\|r\|_{\textnormal{TV}}<\infty$. We begin by once again recalling the identity $\E[r(Z)^2] = \lECE(f)$. Equivalently,
\begin{align*}
    \lECE(f) = \frac{1}{n}\sum_{i=1}^{n}\E\left[r(Z_i)^2\right] = \E\left[\frac{1}{n}\sum_{i=1}^{n}r\left(Z_{\pi(i)}\right)^2\right].
\end{align*}
Then the gap  $\Delta_{\textnormal{rank}} = \left|\rankECE(f) - \lECE(f)\right|$ can be decomposed as,
\begin{align*}
    \Delta_{\textnormal{rank}} 
    & = \left|\E\left[\frac{1}{n}\sum_{i=1}^{n-1}r\left(Z_{\pi(i)}\right)\left(r\left(Z_{\pi(i+1)}\right) - r\left(Z_{\pi(i)}\right)\right) - \frac{r\left(Z_{\pi(n)}\right)^2}{n}\right]\right|\\
    & \leq \left|\E\left[\frac{1}{n}\sum_{i=1}^{n-1}r\left(Z_{\pi(i)}\right)\left(r\left(Z_{\pi(i+1)}\right) - r\left(Z_{\pi(i)}\right)\right)\right]\right| + \frac{1}{n}
\end{align*}
where the last inequality follows by recalling that $|r|\leq 1$. Once again using the same bound on $r$ we get,
\begin{align*}
    \Delta_{\textnormal{rank}} \leq \E\left[\frac{\sum_{i=0}^{n}\left|r\left(Z_{\pi(i+1)}\right) - r\left(Z_{\pi(i)}\right)\right|}{n}\right] -\E\left[\frac{\left|r\left(Z_{\pi(1)}\right) - r\left(0\right)\right| + \left|r\left(1\right) - r\left(Z_{\pi(n)}\right)\right|}{n}\right] + \frac{1}{n}.
\end{align*}
where we use the convention $Z_{\pi(0)} = 0$ and $Z_{\pi(n+1)} = 1$. Using the definition of total variation $\|r\|_{\textnormal{TV}}$ we can now immediately show,
\begin{align*}
    \Delta_{\textnormal{rank}}\leq \frac{\|r\|_{\textnormal{TV}}}{n} + \frac{1}{n},
\end{align*}
which concludes the proof.

\subsection{Proof of Theorem \ref{thm:null_calibration}}\label{sec:proofof_thm_null_calibration}
The proof proceeds by first establishing the claim under the assumption that the underlying distribution of $Z$ is $\textnormal{Unif}[0,1]$, and then transferring the result to the general case via the inverse probability transform. To this end, let $F$ denote the CDF of $Z$. By the non-atomicity of $Z$, there exist independent random variables $U_1, \ldots, U_n \sim \textnormal{Unif}[0,1]$ such that
\[
    Z_i = F^{-1}(U_i), \qquad 1 \leq i \leq n.
\]
where
\begin{align*}
    F^{-1}(u) = \inf\{z\in[0,1]: F(z)\geq u\}\text{ for all }u\in [0,1]
\end{align*}
is the generalized inverse of $F$. Recalling the identity in \eqref{eq:rank_r_equal} and using the monotonicity of $F^{-1}$, we obtain
\begin{align}\label{eq:rank_uniform}
    \rankECE(f)
    =
    \E\left[\frac{1}{n}\sum_{i=1}^{n-1}\bar r(U_{(i)})\bar r(U_{(i+1)})\right],
\end{align}
where $\bar r(\cdot) = r(F^{-1}(\cdot))$ with $r$ being the residual function from Proposition \ref{prop:rank_l2_finite_bd}, and
\[
    U_{(1)} \leq U_{(2)} \leq \cdots \leq U_{(n)}
\]
denote the order statistics of $U_1, \ldots, U_n$. Note that to prove Theorem \ref{thm:null_calibration}, it suffices to show that $\seqsplit{\rankECE(f) = 0}$ if and only if $r = 0$ almost surely. By the definition of $\bar r$, we have that $r = 0$ almost surely with respect to $F$ if and only if $\bar r = 0$ almost everywhere with respect to Lebesgue measure on $[0,1]$. 
Therefore, to complete the proof, it suffices to verify that $\rankECE(f) = 0$ if and only if $\bar r = 0$ almost everywhere with respect to Lebesgue measure on $[0,1]$. This is shown in the following lemma:

\begin{lemma}\label{lemma:rank_uniform_0}
    Let $n \geq 4$, and let $U_1, \ldots, U_n$ be independent $\textnormal{Unif}[0,1]$ random variables with order statistics
    \[
        U_{(1)} \leq U_{(2)} \leq \cdots \leq U_{(n)}.
    \]
    Let $g:[0,1]\ra\R$ be a bounded function. Then,
    \begin{align*}
        \E\left[\frac{1}{n}\sum_{i=1}^{n-1}
        g\left(U_{(i)}\right)g\left(U_{(i+1)}\right)\right] = 0
    \end{align*}
    if and only if $g = 0$ almost everywhere with respect to Lebesgue measure on $[0,1]$.
\end{lemma}

\subsubsection{Proof of Lemma \ref{lemma:rank_uniform_0}}\label{sec:proofof_lemma_rank_uniform_0}
To simplify notation define,
\begin{align*}
    M_n := \frac{1}{n}\sum_{i=1}^{n-1}
        g\left(U_{(i)}\right)g\left(U_{(i+1)}\right).
\end{align*}
Note that the \textit{if} direction follows immediately from the definition of $M_n$. In the following we prove the \textit{only if} direction. To that end assume that $\E[M_n] = 0$. For $0\leq x<y\leq 1$ let $\bar g(x,y) = \int_{x}^{y}g(t)\rmd t$ and consider,
    \begin{align}\label{eq:defG}
        G(x,y) = \E\left[g(U)\mid x<U<y\right] = \frac{1}{y-x}\int_{x}^{y}g(t)\rmd t = \frac{\bar{g}(x,y)}{y-x}
    \end{align}
    Next we simplify the expression $\E[M_n]$. By tower property of conditional expectations,
    \begin{align*}
        \E[M_n] = \frac{1}{n}\sum_{i=1}^{n-1}\E\left[\E\left[g\left(U_{(i)}\right)g\left(U_{(i+1)}\right)\mid U_{(i-1)}, U_{(i+2)}\right]\right]
    \end{align*}
    where we use the boundary convention $U_{(0)} = 0$ and $U_{(n+1)} = 1$. Note that conditioned on $U_{(i-1)}=a$ and $U_{(i+2)} = b$ the pair $\left(U_{(i)}, U_{(i+1)}\right)$ has the same distribution as the order statistics $\left(\bar{U}_{(1)},\bar{U}_{(2)}\right)$ of $\bar{U}_1,\bar{U}_2$ generated independently from $\textnormal{Unif}[a,b]$ (see Lemma 1 from \cite{gupta2021distribution}). Then,
    \begin{align*}
        \E\left[g\left(U_{(i)}\right)g\left(U_{(i+1)}\right)\mid U_{(i-1)}, U_{(i+2)}\right] = G\left(U_{(i-1)}, U_{(i+2)}\right)^2.
    \end{align*}
    Consequently,
    \begin{align*}
        \E[M_n] = \frac{1}{n}\sum_{i=1}^{n-1}\E\left[G\left(U_{(i-1)}, U_{(i+2)}\right)^2\right].
    \end{align*}
    Recalling that $\E[M_n] = 0$, it immediately follows that $G\left(U_{(i-1)}, U_{(i+2)}\right) = 0$ almost surely for all $1\leq i\leq n-1$. In particular $G\left(U_{(1)}, U_{(4)}\right) = 0$ almost surely. Moreover, a direct computation shows that the density of $\left(U_{(1)}, U_{(4)}\right)$ is strictly positive on $S:=\left\{(x,y): 0< x<y< 1\right\}$. Therefore $G(x,y) = 0$ for Lebesgue almost every $(x,y)\in S$. Recalling the defintion of $G$ from \eqref{eq:defG}, $\bar g(x,y) = 0$ for Lebesgue almost every $(x,y)\in S$. Next we will show that this conclusion indeed holds for every $(x,y)\in S$. To that end fix $(x,y)\in S$ and consider a sequence $(x_k,y_k)$ such that $x_k\ra x$ and $y_k\ra y$ as $k\ra\infty$. For all $k$ large enough, $x_k<y_k$ and,
    \begin{align}\label{eq:continuity_bar_g}
        \left|\bar g(x,y) - \bar g(x_k, y_k)\right| = \left|\int_{x_k}^{y_k}g(t)\rmd t - \int_{x}^{y}g(t)\rmd t\right|\lesssim \left|x_k - x\right| + \left|y_k - y\right|\ra 0
    \end{align}
    where the last inequality follows by recalling that $g$ is bounded. Now consider $0< x<y< 1$ such that $\bar g(x,y)\neq 0$. Then by the continuity established in \eqref{eq:continuity_bar_g} there exists $\vep>0$ such that $\bar g\neq 0$ in $B_{\vep}((x,y))\subseteq S$, where $B_{\vep}((x,y))$ denotes a ball of radius $\vep$ around $(x,y)$. However, Lebesgue measure of $B_{\vep}((x,y))$ is positive contradicting the fact $\bar g(x,y) = 0$ for Lebesgue almost every $(x,y)\in S$. Hence we conclude $\bar g(x,y) = 0$ for all $(x,y)\in S$. Moreover, by the same continuity argument as above we can conclude that $\bar g(x,y) = 0$ for all $0\leq x<y\leq 1$. Finally note that the collection,
    \begin{align*}
        \cS := \left\{(a,b]: 0\leq a<b\leq 1\right\}\bigcup\{\emptyset\}
    \end{align*}
    is a $\pi$-system that generates the Borel-sigma algebra on $[0,1]$. The proof is now completed by recalling the definition of $\bar g$ and applying Theorem 2.6.2 from \cite{gut2006probability}.

\subsection{Proof of Proposition \ref{prop:finite_test}}\label{sec:proofof_prop_finite_test}

The proof of Proposition \ref{prop:finite_test} proceeds by decomposing $\hatrankECE(f)$ into two terms, each of which can be expressed as a sum of conditionally independent random variables given $\bZ$. Bernstein's inequality is then applied to each term separately to obtain the desired concentration bound. 

We begin by recalling the definition of $\hatrankECE(f)$ from \eqref{eq:rank_ece_sample}. For notational convenience, let $W_i = Y_{\pi(i)} - Z_{\pi(i)}$ for all $1 \leq i \leq n$, and define $m_1 = \left\lceil \frac{n-1}{2}\right\rceil$ and $m_2 = \left\lfloor \frac{n-1}{2}\right\rfloor$. With this notation, we decompose $\hatrankECE(f)$ into two terms as follows:
\begin{align*}
    n\cdot\hatrankECE(f)
    =
    \underbrace{\sum_{i=1}^{m_1} V_{1,i}}_{R_1}
    +
    \underbrace{\sum_{i=1}^{m_2} V_{2,i}}_{R_2},
\end{align*}
where $V_{1,i} = W_{2i-1}W_{2i}$ for all $1 \leq i \leq m_1$, and $V_{2,i} = W_{2i}W_{2i+1}$ for all $1 \leq i \leq m_2$. We now apply Bernstein's inequality to both $R_1$ and $R_2$ conditional on $\bZ$. To this end, we first establish the required conditional independence structure.

Conditional on $\bZ$, the random variables $\{W_i\}_{i=1}^n$ are mutually independent. Since each $V_{1,i}=W_{2i-1}W_{2i}$ depends only on the pair $(W_{2i-1},W_{2i})$, and these pairs are disjoint across $i$, it follows that the collection $\{V_{1,i}\}_{i=1}^{m_1}$ consists of conditionally independent random variables. An identical argument shows that $\{V_{2,i}\}_{i=1}^{m_2}$ is also a collection of conditionally independent random variables. Consequently, both $R_1$ and $R_2$ are sums of conditionally independent random variables.

Under the null hypothesis $\bm H_0$, we have $\E[W_i\mid\bZ]=0$ for all $1\leq i\leq n$. By conditional independence, this implies $\E[V_{j,i}\mid\bZ]=0$, and consequently $\E[R_j\mid\bZ]=0$ for $j=1,2$. Finally, since the summands $\{V_{j,i}\}$ are conditionally independent and centered under $\bm H_0$,
\begin{align*}
    \sigma_j^2(\bZ)
    :=\Var_{\bm H_0}(R_j\mid\bZ)
    =\sum_{i=1}^{m_j}\Var\!\left(V_{j,i}\mid\bZ\right),
    \qquad j=1,2.
\end{align*}
A direct computation then yields
\begin{align*}
    \sigma_1^2(\bZ)
    &=
    \sum_{i=1}^{m_1}
    Z_{\pi(2i-1)}\!\left(1-Z_{\pi(2i-1)}\right)
    Z_{\pi(2i)}\!\left(1-Z_{\pi(2i)}\right),\\
    \sigma_2^2(\bZ)
    &=
    \sum_{i=1}^{m_2}
    Z_{\pi(2i)}\!\left(1-Z_{\pi(2i)}\right)
    Z_{\pi(2i+1)}\!\left(1-Z_{\pi(2i+1)}\right).
\end{align*}
Finally, observe that $|V_{j,i}|\leq 1$ for all $i\in[m_j]$ and $j=1,2$. Therefore, conditional on $\bZ$, Bernstein's inequality implies that, for each $j=1,2$,
\begin{align*}
    \P\left(R_j\leq \sqrt{2\sigma_j^2(\bZ)\log(2/\alpha)} + \frac{1}{3}\log(2/\alpha)\,\middle|\,\bZ\right)\geq 1-\alpha/2.
\end{align*}
Applying the union bound yields
\begin{align*}
    \P\left(\hatrankECE(f)\leq \frac{\sqrt{2\sigma_1^2(\bZ)\log(2/\alpha)} + \sqrt{2\sigma_2^2(\bZ)\log(2/\alpha)} + \frac{2}{3}\log(2/\alpha)}{n}\,\middle|\, \bZ\right)\geq 1-\alpha.
\end{align*}
Next, using the elementary inequality $\sqrt{a}+\sqrt{b}\leq \sqrt{2(a+b)}$, we obtain
\begin{align*}
    \sqrt{2\sigma_1^2(\bZ)\log(2/\alpha)}
    +\sqrt{2\sigma_2^2(\bZ)\log(2/\alpha)}
    &\leq
    2\sqrt{\left(\sigma_1^2(\bZ)+\sigma_2^2(\bZ)\right)\log(2/\alpha)}\\
    &=2\sqrt{n\sigma^2(\bZ)\log(2/\alpha)}.
\end{align*}
Combining the preceding displays, we conclude that
\begin{align*}
    \P\left(\hatrankECE(f)\leq 2\sqrt{\frac{\sigma^2(\bZ)\log(2/\alpha)}{n}}+\frac{2\log(2/\alpha)}{3n}\,\middle|\,\bZ\right)\geq 1-\alpha.
\end{align*}
The proof of validity is completed by tower property of conditional expectations. The proof of consistency follows immediately by the convergence of $\hatrankECE(f)$ to $\lECE(f)$ from Proposition \ref{prop:consistency}, and by noticing that $\sigma^2(\bZ)\leq 1/16$.

\subsection{Proof of Theorem \ref{thm:null_asymptotic}}\label{sec:proofof_thm_null_asymptotic}
The proof proceeds by identifying $\hatrankECE(f)$ as the terminal value of a martingale. In particular, the cross-products of the residuals form a martingale difference sequence, which allows us to apply a martingale central limit theorem. We begin by introducing the necessary notation.  Let $\sZ_n = \{Z_1,\ldots, Z_n\}$ and $\vep_i = Y_{\pi(i)} - Z_{\pi(i)}$ for all $1\leq i\leq n$. To construct the martingale define the filtration $\cF_{n,i} = \sigma\left\{\sZ_n, \vep_1,\ldots, \vep_{i+1}\right\}$ for all $0\leq i\leq n-1$. Moreover let $D_{n,0} = 0$ and $D_{n,i} = \frac{1}{\sqrt{n}}\vep_i\vep_{i+1}$ for all $1\leq i\leq n-1$. With the above notations define,
\begin{align*}
    S_{n,i}(f) = \sum_{j=0}^{i}D_{n,j}, \text{ for all }0\leq i\leq n-1.
\end{align*}
By definition note that we can identify $S_{n,n-1} = \sqrt{n}\ \hatrankECE(f)$. Next we identify the underlying martingale. By definition note that $|D_{n,i}|\leq 4/\sqrt{n}\leq 4$ for all $1\leq i\leq n, n\geq 1$ and moreover $\E\left[D_{n,i}\mid \cF_{i-1}\right] = \frac{1}{\sqrt{n}}\vep_i\E\left[\vep_{i+1}\mid \cF_{i-1}\right] = 0$, where the last equality follows by definition of concomitant pairs. Hence $\{\left(S_{n,i}(f), \cF_{n,i}\right):0\leq i\leq n-1\}$ forms a zero mean, square integrable martingale with martingale difference $\{D_{n,i}:0\leq i\leq n-1\}$. To apply the Martingale CLT (see Corollary 3.1 and the following remark from \cite{hall2014martingale}) we need to verify the conditional Lindeberg condition (see \eqref{eq:cond_lind}) and convergence of the conditional variance (see \eqref{eq:cond_var_convg}). We begin with simplifying the conditional variance. Let,
\begin{align*}
    V_n = \sum_{i=0}^{n-1}\E\left[D_{n,i}^2\mid\cF_{n,i-1}\right] = \frac{1}{n}\sum_{i=1}^{n-1}\vep_{i}^2Z_{\pi(i+1)}(1-Z_{\pi(i+1)})
\end{align*}
be the conditional variance. The final equality follows by recalling that under null $\E[Y|Z] = Z$ almost surely. Conditional on $Z_{\pi(1)},\ldots, Z_{\pi(n)}$ we get,
\begin{align*}
    \E\left[V_n\mid Z_{\pi(1)},\ldots, Z_{\pi(n)}\right] = \frac{1}{n}\sum_{i=1}^{n-1}Z_{\pi(i)}(1-Z_{\pi(i)})Z_{\pi(i+1)}(1-Z_{\pi(i+1)}),
\end{align*}
and,
\begin{align}\label{eq:cond_var_0}
    \Var\left[V_n\mid Z_{\pi(1)},\ldots, Z_{\pi(n)}\right] 
    & = \frac{1}{n^2}\sum_{i=1}^{n-1}\Var\left[\vep_{i}^2\mid Z_{\pi(1)},\ldots, Z_{\pi(n)}\right]\left[Z_{\pi(i+1)}(1-Z_{\pi(i+1)})\right]^2\nonumber\\
    &\lesssim \frac{1}{n}\ra 0
\end{align}
where the first equality follows by noticing that $\vep_1,\ldots, \vep_n$ are independent given $Z_{\pi(1)},\ldots, Z_{\pi(n)}$. To prove convergence of $V_n$ we first notice that $g(z) = z(1-z)$ is Lipschitz on $[0,1]$ with Lipschitz constant $1$. Then,
\begin{align}\label{eq:EVn_decomp}
    \E\left[V_n\mid Z_{\pi(1)},\ldots, Z_{\pi(n)}\right] = \frac{1}{n}\sum_{i=1}^{n-1}g(Z_i)^2 - \frac{g(Z_{\pi(n)})^2}{n} + \frac{1}{n}\sum_{i=1}^{n-1}g(Z_{\pi(i)})\left(g\left(Z_{\pi(i+1)}\right) - g\left(Z_{\pi(i)}\right)\right).
\end{align}
Note that we can apply law of large number on the first term and by definition of $g$, the second term is $O(1/n)$. Hence, to prove convergence of the conditional variance $V_n$ we first control the third term on right hand side of \eqref{eq:EVn_decomp} as follows,
\begin{align*}
    \left|\frac{1}{n}\sum_{i=1}^{n-1}g(Z_{\pi(i)})\left(g\left(Z_{\pi(i+1)}\right) - g\left(Z_{\pi(i)}\right)\right)\right|
    &\leq \frac{1}{4n}\sum_{i=1}^{n-1}\left|g\left(Z_{\pi(i+1)}\right) - g\left(Z_{\pi(i)}\right)\right|\\
    &\leq \frac{1}{4n}\sum_{i=1}^{n-1}\left|Z_{\pi(i+1)} - Z_{\pi(i)}\right|
    \leq \frac{1}{4n}\left(Z_{\pi(n)} - Z_{\pi(1)}\right)\ra 0.
\end{align*}
Hence by law of large numbers from \eqref{eq:EVn_decomp},
\begin{align*}
    \E\left[V_n\mid Z_{(1)},\ldots, Z_{(n)}\right] = \E\left[g^2(Z)\right] + o_p(1).
\end{align*}
Now recalling the convergence from \eqref{eq:cond_var_0} and using Markov's inequality we conclude,
\begin{align}\label{eq:cond_var_convg}
    V_n\pto \E\left[Z^2(1-Z)^2\right]>0,
\end{align}
where the final inequality follows from the non-atomicity of $Z$. Finally to apply Martingale CLT in the following we verify the conditional Lindeberg condition as follows. Now fix $\delta>0$ and recall that $|D_{n,i}|\leq 4/\sqrt{n}$. Then,
\begin{align}\label{eq:cond_lind}
    \sum_{i=0}^{n-1}\E\left[D_{n,i}^2\one\left\{|D_{n,i}|>\delta\right\}\mid \cF_{n,i-1}\right] = 0
\end{align}
for all large enough $n$. Hence we can now apply Martingale CLT \citep{hall2014martingale} to conclude,
\begin{align*}
    \sqrt{n}\ \hatrankECE(f) = S_{n,n-1}\dto\mathrm{N}\left(0, \E\left[Z^2(1-Z)^2\right]\right).
\end{align*}

\subsection{Proof of Corollary \ref{cor:rank_test_calib}}\label{appendix:proofof_cor_rank_test}
By the Law of Large Numbers and Slutsky's lemma, along with the distributional convergence established in Theorem \ref{thm:null_asymptotic}, under $\bm H_0$,
\begin{align*}
\frac{\sqrt{n}\ \hatrankECE(f)}{\sqrt{\frac{1}{n}\sum_{i=1}^{n}Z_i^2(1-Z_i)^2}}\dto \mathrm{N}(0,1).
\end{align*}
Consistency of the test follows by noticing that under the alternative $\bm H_1$, $\hatrankECE(f)\pto \lECE(f)>0$ and by the non-atomicity of $Z$, $\frac{1}{n}\sum_{i=1}^{n}Z_i^2(1-Z_i)^2\pto \E[Z^2(1-Z)^2]\leq 1/4$.

\subsection{Proof of Remark~\ref{rmk:hardness_result}}\label{app:hardness_result}
The claim stated here is different from the result stated in \citet[Theorem 12.5]{angelopoulos2024theoretical}, but follows a similar construction for the proof. Here we give details for completeness. 

Fix any $\epsilon>0$. Fix any distribution $Q_X$ on $X\in\cX$ such that the induced distribution of $f(X)$ is nonatomic and is supported on $[c-\epsilon,c+\epsilon]$, and let $Y\mid X \sim\textnormal{Bernoulli}(f(X))$. Let $Q$ denote the resulting joint distribution. Then $f$ is perfectly calibrated for this distribution, i.e., $\lECE_Q(f)=0$ (where we use the subscript for the remainder of proof to help clarify the underlying distribution). 
Then
\begin{equation}\label{eqn:bound1_proof_rmk_hardness}\left|\E_{Q^n}[\hat{E}_n(f)] \right|  = \left|\E_{Q^n}[\hat{E}_n(f)] - \lECE_Q(f)\right| \leq {\textstyle\sup_P} \,\left| \E[\hat{E}_n(f)] - \lECE(f)\right|.\end{equation}

Next,
for a large integer $M\gg n$, let $(X^{(1)},Y^{(1)}),\dots,(X^{(M)},Y^{(M)})\sim Q$, and let $\widehat{Q}_M$ denote the empirical distribution of these data points. Now let $(X_1,Y_1),\dots,(X_n,Y_n)$ be sampled i.i.d.\ from $\widehat{Q}_M$ (i.e., sampled with replacement from the finite collection of $M$ data points).
Exactly as in the proof of \citet[Theorem 12.5]{angelopoulos2024theoretical}, under the distribution $\widehat{Q}_M$, conditional on $Z=f(X^{(i)})$ we have $Y=Y^{(i)}$ (i.e., the conditional distribution is a point mass), and therefore
\[\lECE_{\widehat{Q}_M}(f) = \frac{1}{M}\sum_{i=1}^M \big(Y^{(i)} - f(X^{(i)})\big)^2,\]
on the (almost sure) event that $f(X^{(1)}),\dots,f(X^{(M)})$ are distinct. Conditional on $\widehat{Q}_M$, therefore,
\begin{align*}
    \left|\E_{\widehat{Q}_M}[\hat{E}_n(f)] - \frac{1}{M}\sum_{i=1}^M \big(Y^{(i)} - f(X^{(i)})\big)^2 \right|  
    & = \left|\E_{\widehat{Q}_M}[\hat{E}_n(f)] - \lECE_{\widehat{Q}_M}(f)\right|\\
    &\leq {\textstyle\sup_P} \,\left| \E[\hat{E}_n(f)] - \lECE(f)\right|.
\end{align*}
Now Let $\tilde{Q}$ denote the distribution of data points $(X_1,Y_1),\dots,(X_n,Y_n)\sim \widehat{Q}_M$ after marginalizing over  $\widehat{Q}_M$. Then, marginalizing over $\widehat{Q}_M$ and applying Jensen's inequality,
\begin{multline*}
\E\left[\left|\E_{\widehat{Q}_M}[\hat{E}_n(f)] - \frac{1}{M}\sum_{i=1}^M \big(Y^{(i)} - f(X^{(i)})\big)^2 \right|\right]
\geq \left|\E_{\tilde{Q}}[\hat{E}_n(f)] - \E\left[\frac{1}{M}\sum_{i=1}^M \big(Y^{(i)} - f(X^{(i)})\big)^2 \right]\right|
\\
= \left|\E_{\tilde{Q}}[\hat{E}_n(f)] - \E_Q[f(X)(1-f(X)]\right|,
\end{multline*}
where the last step holds since, for $(X^{(i)},Y^{(i)})\sim Q$,
\[\E[\big(Y^{(i)} - f(X^{(i)})\big)^2] = \E\left[\E\left[\big(Y^{(i)} - f(X^{(i)})\big)^2\mid X^{(i)}\right]\right] = \E[f(X^{(i)})(1-f(X^{(i)}))] ,\]
since $Y^{(i)}\mid X^{(i)}\sim\textnormal{Bernoulli}(f(X^{(i)}))$. Since $f(X^{(i)})\in(c-\epsilon,c+\epsilon)$, we also have
\[\left|f(X^{(i)})(1-f(X^{(i)}))  - c(1-c)\right|\leq \epsilon,\]
and therefore,
\[\left|\E_{\tilde{Q}}[\hat{E}_n(f)] - c(1-c)\right|  \leq \epsilon + {\textstyle\sup_P} \,\left| \E[\hat{E}_n(f)] - \lECE(f)\right|.\]

Finally, again following the proof of \citet[Theorem 12.5]{angelopoulos2024theoretical}, we have $d_{\mathrm{TV}}(Q^n,\tilde{Q})\leq \frac{n(n-1)}{2M}$, where $d_{\mathrm{TV}}$ denotes total variation distance between distributions.
Then
\[\left|\E_{Q^n}[\hat{E}_n(f)]  - \E_{\tilde{Q}}[\hat{E}_n(f)]\right| \leq d_{\mathrm{TV}}(Q^n,\tilde{Q})\leq \frac{n(n-1)}{2M},\]
since we can assume $\hat{E}_n(f)$ takes values in $[0,1]$ without loss of generality, and so
\begin{equation}\label{eqn:bound2_proof_rmk_hardness}\left|\E_{Q^n}[\hat{E}_n(f)] - c(1-c)\right| \leq \epsilon + \frac{n(n-1)}{2M} + {\textstyle\sup_P} \,\left| \E[\hat{E}_n(f)] - \lECE(f)\right|.\end{equation}

Combining~\eqref{eqn:bound1_proof_rmk_hardness} and~\eqref{eqn:bound2_proof_rmk_hardness}, we have shown that
\begin{align*}
    {\textstyle\sup_P} \,\bigg| \E[\hat{E}_n(f)] 
    & - \lECE(f)\bigg|\\
    & \geq \frac{1}{2}\left|\E_{Q^n}[\hat{E}_n(f)] \right|  + \frac{1}{2}\left(\left|\E_{Q^n}[\hat{E}_n(f)] - c(1-c)\right| - \epsilon - \frac{n(n-1)}{2M}\right)\\
    & \geq \inf_{t\in[0,1]} \left\{\frac{1}{2}\left|t \right|  + \frac{1}{2}\left(\left|t - c(1-c)\right| - \epsilon - \frac{n(n-1)}{2M}\right)\right\} = \frac{c(1-c)}{2} - \frac{\epsilon}{2} - \frac{n(n-1)}{4M}.
\end{align*}

Since $\epsilon>0$ can be taken arbitrarily small, and $M\gg n$ can be taken to be arbitrarily large, this completes the proof.

\section{Proof of results from Section \ref{sec:compare}}

\subsection{Proof of Proposition \ref{prop:conc_binhat}}\label{sec:proofof_conc_binhat}

The proof of Proposition \ref{prop:conc_binhat} proceeds by an application of McDiarmid's Bounded Difference inequality. To that end, for notational convenience let $\bW_i = (Y_i, Z_i)$ for all $1\leq i\leq n$ and define (recall \eqref{eq:binhatECE}),
\begin{align}\label{eq:def_g_binhat}
    g(\bW_1,\ldots, \bW_n) = \binhatECE(f) = \frac{1}{n}\sum_{j=1}^{K} \frac{D_j^2}{n_j}
\end{align}
where $D_j = \sum_{i:Z_i\in \cB_j}\left(Y_i - Z_i\right)$ for all $j\in [K]$. Fix $i\in [n]$. In the following we replace $\bW_i$ by an independent copy $\bW_i^\prime = (Y_i^\prime, Z_i^\prime)$ and bound the difference
\begin{align*}
    \left|g(\bW_1,\ldots, \bW_{i-1}, \bW_i,\bW_{i+1}, \ldots, \bW_n) - g(\bW_1,\ldots, \bW_{i-1}, \bW_i^\prime,\bW_{i+1}, \ldots, \bW_n)\right|.
\end{align*}
Next, suppose that $Z_i\in \cB_\alpha$ and $Z_i^\prime\in \cB_\beta$ for some $\alpha,\beta\in [K]$. We proceed with the proof by considering two cases:

\paragraph{Case 1: $\alpha = \beta$.} Notice that in this case, after replacement the number of samples $n_\alpha$ in bin $\cB_\alpha$ remains unaffected. Since samples in all other bins remain unaffected then,
\begin{align*}
    \big|g(\bW_1,\ldots, \bW_{i-1}, \bW_i,\bW_{i+1}, \ldots, \bW_n) 
    & - g(\bW_1,\ldots, \bW_{i-1}, \bW_i^\prime,\bW_{i+1}, \ldots, \bW_n)\big|\\
    & \leq \frac{1}{nn_\alpha}\left|(T_\alpha + (Y_i-Z_i))^2 - (T_\alpha + (Y_i^\prime-Z_i^\prime))^2\right|
\end{align*}
where $T_\alpha = D_\alpha - (Y_i-Z_i)$. A direct computation shows,
\begin{align*}
    \frac{\left|(T_\alpha + (Y_i-Z_i))^2 - (T_\alpha + (Y_i^\prime-Z_i^\prime))^2\right|}{nn_\alpha} = \frac{|(Y_i-Z_i) - (Y_i^\prime - Z_i^\prime)||2T_\alpha + (Y_i-Z_i) + (Y_i^\prime - Z_i^\prime)|}{nn_\alpha}.
\end{align*}
Since $Y_t-Z_t\in [-1,1]$ for all $t\in [n]$, then we immediately get $|(Y_i-Z_i) - (Y_i^\prime - Z_i^\prime)|\leq 2$ and $|2T_\alpha + (Y_i-Z_i) + (Y_i^\prime - Z_i^\prime)|\leq 2(n_\alpha-1) + 2 = 2n_\alpha$. Then in this case we conclude,
\begin{align}\label{eq:bound_g_diff_1}
    \big|g(\bW_1,\ldots, \bW_{i-1}, \bW_i,\bW_{i+1}, \ldots, \bW_n) 
    & - g(\bW_1,\ldots, \bW_{i-1}, \bW_i^\prime,\bW_{i+1}, \ldots, \bW_n)\big|\leq \frac{4}{n}.
\end{align}

\paragraph{Case 2: $\alpha\neq \beta$.} In this case only samples in bins $\cB_\alpha$ and $\cB_\beta$ are affected. Hence,
\begin{align}\label{eq:diff_g_unequal}
    \big|g(\bW_1,\ldots, \bW_{i-1}, \bW_i,\bW_{i+1}, \ldots, \bW_n) 
    & - g(\bW_1,\ldots, \bW_{i-1}, \bW_i^\prime,\bW_{i+1}, \ldots, \bW_n)\big|\nonumber\\
    & = \underbrace{\left|\left(\frac{D_\alpha^2}{nn_\alpha} - \frac{\bar D_\alpha^2}{n\bar n_\alpha}\right)\right|}_{\Delta_\alpha} + \underbrace{\left|\left(\frac{D_\beta^2}{nn_\beta} - \frac{\bar D_\beta^2}{n\bar n_\beta}\right)\right|}_{\Delta_\beta}
\end{align}
where $\bar D_{\alpha}$ and $\bar D_{\beta}$ are defined analogously to $D_\alpha$ and $D_\beta$, respectively, and $\bar n_\alpha$ and $\bar n_\beta$ are defined analogously to $n_\alpha$ and $n_\beta$ using the samples $\bW_1,\ldots, \bW_{i-1}, \bW_i^\prime,\bW_{i+1}, \ldots, \bW_n$. We first bound the first term in right hand side of \eqref{eq:diff_g_unequal}. Note that if $n_\alpha = 1$, then $\Delta_\alpha = (Y_i - Z_i)^2/n\leq 1/n$. If $n_\alpha\geq 2$, then
\begin{align*}
    \Delta_\alpha = \frac{1}{n}\left|\frac{T_\alpha^2}{n_\alpha-1} - \frac{(T_\alpha + Y_i-Z_i)^2}{n_\alpha}\right| = \left|\frac{T_\alpha^2 - 2(n_\alpha-1)T_\alpha(Y_i-Z_i) - (n_\alpha-1)(Y_i-Z_i)^2}{nn_\alpha(n_\alpha-1)}\right|,
\end{align*}
where the first equality follows by observing that by definition $\bar D_\alpha = T_\alpha$. Using the bounds $T_\alpha^2\leq (n_\alpha-1)^2$, $|2(n_\alpha-1)T_\alpha(Y_i-Z_i)|\leq 2(n_\alpha-1)^2$ and $(n_\alpha-1)(Y_i-Z_i)^2\leq n_\alpha-1$ we conclude,
\begin{align*}
    \Delta_\alpha \leq \frac{3(n_\alpha-1)^2 + (n_\alpha-1)}{nn_\alpha(n_\alpha-1)}\leq \frac{3}{n}
\end{align*}
Similarly one can show that $\Delta_\beta\leq \frac{3}{n}$. Then combining with \eqref{eq:diff_g_unequal} we get,
\begin{align}\label{eq:bound_g_diff_2}
    \big|g(\bW_1,\ldots, \bW_{i-1}, \bW_i,\bW_{i+1}, \ldots, \bW_n) 
    & - g(\bW_1,\ldots, \bW_{i-1}, \bW_i^\prime,\bW_{i+1}, \ldots, \bW_n)\big|\leq \frac{6}{n}.
\end{align}

The proof of Proposition \ref{prop:conc_binhat} is now completed by plugging the bounds from \eqref{eq:bound_g_diff_1} and \eqref{eq:bound_g_diff_2} in McDiarmid's bounded difference inequality. 

\subsection{Proof of Lemma \ref{lemma:stat_bias_bound}}\label{sec:proofof_lemma_stat_bias_bound}
Recalling the expression of $\binhatECE(f)$ from \eqref{eq:binhatECE} note that,
\begin{align}\label{eq:exp_binhat}
    \E\left[\binhatECE(f)\right] = \sum_{j=1}^{K}\E\left[\frac{n_j}{n}\left(\bar Y_j - \bar Z_j\right)^2\one\{n_j\geq 1\}\right]
\end{align}
Fix $j\in [K]$. Then note that conditional on $n_j = m\geq 1$ the samples in $\cB_j$ are drawn independently from the distribution $P_{Y,Z\mid Z\in \cB_j}$. A direct computation then shows,
\begin{align*}
    \E\left[\left(\bar Y_j - \bar Z_j\right)^2\one\{n_j\geq 1\}\mid n_j=m\right] = \left(\mu_j^2 + \frac{\sigma_j^2}{m}\right)\one\{m\geq 1\},
\end{align*}
where $\mu_j = \E[Y-Z\mid Z\in \cB_j]$ and $\sigma_j^2 = \Var\left(Y-Z\mid Z\in \cB_j\right)$. Then using tower property of conditional expectation the identity of \eqref{eq:exp_binhat} becomes,
\begin{align}\label{eq:Ebin_expansion}
    \E\left[\binhatECE(f)\right] 
    & = \frac{1}{n}\sum_{j=1}^{K}\E\left[(n_j\mu_j^2 + \sigma_j^2)\one\{n_j\geq 1\}\right] = \frac{1}{n}\sum_{j=1}^{K}\E\left[n_j\mu_j^2 + \sigma_j^2\one\{n_j\geq 1\}\right]\nonumber\\
    & = \sum_{j=1}^{K}p_j\mu_j^2 + \frac{\sigma_j^2}{n}\left(1-(1-p_j)^n\right)
\end{align}
where $p_j = \P(Z\in \cB_j)$ for all $j\in [K]$. To complete the proof note that $Y-Z\in [-1,1]$ and hence by Theorem 2 from \cite{bhatia2000better} we conclude $\sigma_j^2\leq 1$ for all $j\in [K]$. Hence using the identity from \eqref{eq:Ebin_expansion} we conclude,
\begin{align*}
    \left|\E\left[\binhatECE(f)\right] - \lbinECE(f)\right|\leq \frac{K}{n}.
\end{align*}

\subsection{Proof of Theorem \ref{thm:rank_bin_comparison}}\label{sec:proofof_thm_rank_bin_comparision}

Given samples $Z_1,\ldots, Z_n$ let $n_k$ denote the number of samples in bin $\cB_k$ and let us denote the samples in $\cB_k$ as $Z_{1}^{k}, \ldots, Z_{n_k}^{k}$. Moreover we consider the ordered version of samples in bin $\cB_k$ as $Z_{(1)}^{k}\leq \ldots\leq Z_{(n_k)}^{k}$. Splitting the consecutive pairs $(Z_{\pi(i)}, Z_{\pi(i+1)})$ in the global ordering into intra-bin and boundary pairs, noting that bin $\cB_k$ contributes intra-bin pairs only when $n_k\geq 2$, and using the identity from \eqref{eq:rank_r_equal} we have
\begin{align*}
    \rankECE(f) = \underbrace{\sum_{k=1}^{K}\E\left[\frac{1}{n}\sum_{i=1}^{n_k-1}r\left(Z_{(i)}^{k}\right)r\left(Z_{(i+1)}^k\right)\one\{n_k\geq 2\}\right]}_{T_n} + B_n
\end{align*}
where $B_n$ collects the contributions from consecutive pairs whose two endpoints fall in different bins and $r$ is the residual function from Proposition \ref{prop:rank_l2_finite_bd}. The number of such boundary pairs equals the number of non-empty bins minus one, which is at most $K-1$, and each summand is bounded in absolute value by $1$. Hence
\begin{align}\label{eq:Bn_bound}
    |B_n|\leq \frac{K-1}{n}.
\end{align}
We now lower bound $T_n$. For $k\in [K]$ let $\mu_k = \E\left[r(Z)\mid Z\in \cB_k\right]$ and define
\begin{align*}
    r_k^\circ(\cdot) = r(\cdot) - \mu_k\text{ for all }k\in [K].
\end{align*}
Expanding the product,
\begin{align}\label{eq:r_prod_expansion}
    r\left(Z_{(i)}^k\right)r\left(Z_{(i+1)}^k\right) = r_k^\circ\left(Z_{(i)}^k\right)r_k^\circ\left(Z_{(i+1)}^k\right) + \mu_k\left(r\left(Z_{(i)}^k\right) + r\left(Z_{(i+1)}^k\right)\right) - \mu_k^2.
\end{align}
Now note that conditional on $n_k = m$ with $m\geq 2$, $Z_1^{k}, \ldots, Z_{m}^{k}$ has the same joint distribution as $L_1,\ldots, L_{m}$ generated independently from $P_{Z\mid Z\in B_k}$. Then for $1\leq i\leq m-1$,
\begin{align*}
    \E\left[r_k^\circ\left(Z_{(i)}^k\right)r_k^\circ\left(Z_{(i+1)}^k\right)\mid n_k = m\right] = \E\left[\E\left[r_k^\circ\left(L_{(i)}\right)r_k^\circ\left(L_{(i+1)}\right)\mid L_{(i-1)}, L_{(i+2)}\right]\right]
\end{align*}
where $L_{(1)}\leq L_{(2)}\leq \cdots\leq L_{(m)}$ are the order statistics of $L_1,\ldots, L_m$ and we take the boundary cases $L_{(0)} = a_{k-1}$ and $L_{(m+1)} = a_{k}$. One again applying Lemma 1 from \cite{gupta2021distribution}, conditional on $L_{(i-1)}, L_{(i+2)}$, the pair $(L_{(i)}, L_{(i+1)})$ is distributed as the order statistics of two independent samples $M_1, M_2$ generated from $P_{L\mid L\in (L_{(i-1)}, L_{(i+2)})}$. Let $M_{(1)}\leq M_{(2)}$ denote the order statistics of $M_1, M_2$. Since the product is symmetric in its two arguments, $r_k^\circ(M_{(1)})r_k^\circ(M_{(2)}) = r_k^\circ(M_1)r_k^\circ(M_2)$ and so by independence
\begin{align*}
    \E\left[r_k^\circ\left(M_{(1)}\right)r_k^\circ\left(M_{(2)}\right)\mid L_{(i-1)}, L_{(i+2)}\right] = \left(\E\left[r_k^\circ(M_1)\mid L_{(i-1)}, L_{(i+2)}\right]\right)^2 \geq 0.
\end{align*}
Hence the cross term in \eqref{eq:r_prod_expansion} is non-negative in expectation, and so
\begin{align}\label{eq:lb_tn_1}
    T_n \geq \sum_{k=1}^{K}\E\left[\frac{1}{n}\sum_{i=1}^{n_k-1}\left(\mu_k\left(r(Z_{(i)}^k) + r(Z_{(i+1)}^k)\right) - \mu_k^2\right)\one\{n_k\geq 2\}\right].
\end{align}
On $\{n_k\geq 2\}$ the sum telescopes:
\begin{align*}
    \sum_{i=1}^{n_k-1}\left(r(Z_{(i)}^k) + r(Z_{(i+1)}^k)\right) = 2\sum_{i=1}^{n_k}r(Z_i^k) - r(Z_{(1)}^k) - r(Z_{(n_k)}^k),
\end{align*}
and conditional on $n_k$ we have $\E[\sum_{i=1}^{n_k}r(Z_i^k)\mid n_k] = n_k\mu_k$. Substituting in \eqref{eq:lb_tn_1},
\begin{align}\label{eq:lb_tn_2}
    T_n \geq \frac{1}{n}\sum_{k=1}^{K}\E\left[\left((n_k+1)\mu_k^2 - \mu_k\left(r(Z_{(1)}^k) + r(Z_{(n_k)}^k)\right)\right)\one\{n_k\geq 2\}\right].
\end{align}
We lower bound the two pieces on the right hand side separately. For the first piece, write $p_k = \P(Z\in \cB_k)$ and note that
\begin{align*}
    \E\left[(n_k+1)\one\{n_k\geq 2\}\right]
    & = \E[n_k+1] - \E\left[(n_k+1)\one\{n_k\leq 1\}\right]\\
    & = (np_k + 1) - \P(n_k = 0) - 2\P(n_k = 1)\geq np_k - 1
\end{align*}
where the last inequality uses $\P(n_k = 0) + 2\P(n_k = 1)\leq 2$. Summing over $k$ and using $\mu_k^2\leq 1$,
\begin{align}\label{eq:first_piece}
    \frac{1}{n}\sum_{k=1}^{K}\mu_k^2\E\left[(n_k+1)\one\{n_k\geq 2\}\right]\geq \lbinECE(f) - \frac{K}{n}.
\end{align}
For the second piece on the right hand side in \eqref{eq:lb_tn_2}, since $|r|\leq 1$ and $|\mu_k|\leq 1$,
\begin{align}\label{eq:second_piece}
    \left|\frac{1}{n}\sum_{k=1}^{K}\E\left[\mu_k\left(r(Z_{(1)}^k) + r(Z_{(n_k)}^k)\right)\one\{n_k\geq 2\}\right]\right|\leq \frac{2K}{n}.
\end{align}
Substituting the bounds from \eqref{eq:first_piece} and \eqref{eq:second_piece} in \eqref{eq:lb_tn_2},
\begin{align*}
    T_n\geq \lbinECE(f) - \frac{3K}{n}.
\end{align*}
Finally, putting this together with \eqref{eq:Bn_bound} we complete the proof,
\begin{align*}
    \rankECE(f) = T_n + B_n \geq \lbinECE(f) - \frac{3K}{n} - \frac{K-1}{n}\geq \lbinECE(f) - \frac{4K}{n}.
\end{align*}

\subsection{Justification of Remark \ref{rmk:stat_bias_bound}}\label{sec:e_binhat_O1}
From the proof of Lemma \ref{lemma:stat_bias_bound}, and in particular from \eqref{eq:Ebin_expansion}, recall that
\begin{align}\label{eq:Ebin_expand_perfect}
\E\left[\binhatECE(f)\right] = \sum_{j=1}^{K}p_j\mu_j^2 + \frac{\sigma_j^2}{n}\left(1-(1-p_j)^n\right)
\end{align}
where $p_j = \P(Z\in \cB_j)$, $\mu_j = \E[Y-Z\mid Z\in \cB_j]$, and $\sigma_j^2 = \Var(Y-Z\mid Z\in \cB_j)$, with $\cB_1,\ldots,\cB_K$ denoting the bins. Under the perfect calibration assumption, namely $\E[Y|Z] = Z$, it follows immediately that $\mu_j = 0$ for all $j\in [K]$. Furthermore, since $Z\sim \textnormal{Unif}[0,1]$, we have $Z\mid Z\in \cB_j\sim \textnormal{Unif}[(j-1)/K, j/K]$ for every $j\in [K]$. Consequently, $p_j = \frac{1}{K}$ for all $j\in[K]$. Substituting these identities into \eqref{eq:Ebin_expand_perfect} yields
\begin{align*}
\E\left[\binhatECE(f)\right]
& = \frac{1}{n}\left(1-(1-1/K)^n\right)\sum_{j=1}^{K}\sigma_j^2\\
& = \frac{1}{n}\left(1-\left(1-\frac{1}{K}\right)^n\right)\sum_{j=1}^{K}\E\left[Z(1-Z)\mid Z\in \cB_j\right]\\
& \geq \frac{K}{6n}\left(1-\left(1-\frac{1}{K}\right)^n\right)\geq \frac{K}{6n}\left(1-\frac{1}{e}\right)
\end{align*}

\subsection{Justification of Remark \ref{remark:rate_of_convg_compare}}\label{sec:remark_lip_rate}
Similar to Remark \ref{rmk:stat_bias_bound} in this section we consider $[0,1] = \cB_1\cup\cdots\cup\cB_K$ to be a partition of $[0,1]$ into $K$ equal-width bins.
Recall the residual function $r(z) = \E[Y-Z|Z=z]$ from Proposition \ref{prop:rank_l2_finite_bd}. By assumption $r$ is Lipschitz continuous with Lipschitz constant $L$. Recalliing definitions from \eqref{eq:L2_binned_ECE} and \eqref{eq:L2_ECE} note that,
\begin{align*}
    \left|\lbinECE(f) - \lECE(f)\right| 
    & = \left|\sum_{j=1}^{K}p_j\left(\E[r(Z)^2\mid Z\in \cB_j] - \E[r(Z)\mid Z\in \cB_j]^2\right)\right|\\
    & = \sum_{j=1}^{K}p_j\Var(r(Z)\mid Z\in \cB_j)
\end{align*}
where $p_j = \P(Z\in \cB_j)$. To bound the right hand side in the above identity we will once again invoke Theorem 2 from \cite{bhatia2000better} in each bin $\cB_j$. To that end note that $\inf_{z\in \cB_j}r(z)\leq r(z)\leq \sup_{z\in \cB_j}r(z)$ for all $z\in \cB_j$ and,
\begin{align*}
    \sup_{z\in \cB_j}r(z) - \inf_{z\in \cB_j}r(z)\leq \sup_{z_1, z_2\in \cB_j}\left|r(z_1) - r(z_2)\right|\leq \frac{L}{K}
\end{align*}
where the final inequality follows by using the Lipschitz continuity of $r$ and bin width $1/K$. Now we can apply the bound on $\Var(r(Z)\mid Z\in \cB_j)$ from Theorem 2 in \cite{bhatia2000better} to show that,
\begin{align*}
    \left|\lbinECE(f) - \lECE(f)\right|\leq\sum_{j=1}^{K}p_j\frac{L^2}{4K^2} \leq \frac{L^2}{4K^2}.
\end{align*}

\section{On non-atomicity of $f(X)$}\label{appendix:atomic}
In Section \ref{sec:def_rank_ece}, recall that, for a prediction function $f:\cX\ra[0,1]$ and covariate $X$, we defined the measure $\rankECE(f)$ under the convenient assumption that $Z=f(X)$ is nonatomic, i.e., $\P(Z=t)=0$ for all $t\in[0,1]$. This assumption was made to avoid ties among the predictions $Z_1,\ldots,Z_n$. In this section, we show that the nonatomicity assumption is merely a technical convenience that simplifies the exposition and avoids notational complexities. We therefore generalize the definition of $\rankECE$ to an assumption-free setting. 

Following the order of definitions in Section \ref{sec:def_rank_ece}, we first introduce $\hatrankECE$ in an assumption-free setting. Consider observations $(Y_1,Z_1),\ldots,(Y_n,Z_n)$, where $Z_i=f(X_i)$ for each $i\in[n]$. Independently for each $i\in[n]$, draw $U_i\sim\textnormal{Unif}[0,1]$ and define $\bW_i=(Z_i,U_i)$. Let $\pi:[n]\ra[n]$ be the permutation obtained by sorting $\bW_1,\ldots,\bW_n$ in lexicographic order; equivalently,
\begin{align}\label{eq:def_perm_sigma}
    \bW_{\pi(i)}\prec \bW_{\pi(j)}
    \iff
    \begin{cases}
        Z_{\pi(i)} < Z_{\pi(j)}, & \text{if } Z_{\pi(i)}\neq Z_{\pi(j)},\\
        U_{\pi(i)} < U_{\pi(j)}, & \text{if } Z_{\pi(i)}=Z_{\pi(j)}.
    \end{cases}
\end{align}
By construction, $\pi$ is unique. We can now define the empirical $\hatrankECE$ as
\begin{align}\label{eq:hat_rank_sigma}
    \hatrankECE(f)
    =
    \frac{1}{n}\sum_{i=1}^{n-1}
    \left(Y_{\pi(i)}-Z_{\pi(i)}\right)
    \left(Y_{\pi(i+1)}-Z_{\pi(i+1)}\right).
\end{align}
Taking expectations on both sides of \eqref{eq:hat_rank_sigma} we can now define $\rankECE(f)$ as,
\begin{align}\label{eq:rank_ece_sigma}
    \rankECE(f)
    =
    \E\left[
    \frac{1}{n}\sum_{i=1}^{n-1}
    \left(Y_{\pi(i)}-Z_{\pi(i)}\right)
    \left(Y_{\pi(i+1)}-Z_{\pi(i+1)}\right)
    \right].
\end{align}

Up to minor modifications, all the results established above for $\rankECE$ (in the nonatomic setting) hold in this more general setting as well. To formally verify some of the most important properties, the next result (with proof presented in Appendix \ref{appendix:proofof_rank_proper_sigma}) shows that Proposition~\ref{prop:lower_bound}, Theorem~\ref{thm:null_calibration}, and Theorem~\ref{thm:rank_bin_comparison} all hold in this setting.

\begin{theorem}\label{thm:rank_sigma_proper}
    Let $f:\cX\ra[0,1]$ be any predictor and for $K\geq 1$ and $0=a_0<a_1<\cdots<a_K=1$, let $\cB_1=[a_0,a_1), \cB_2=[a_1,a_2),\ldots, \cB_K=(a_{K-1},a_K]$ be a partition of $[0,1]$ into $K$ bins. Take $\rankECE(f)$ defined in \eqref{eq:rank_ece_sigma}, $\lECE(f)$ defined in \eqref{eq:L2_ECE} and $\lbinECE(f)$ defined in \eqref{eq:L2_binned_ECE} with the bins $\cB_1,\ldots, \cB_K$. Then the following holds:
    \begin{itemize}
        \item [(a)] For $n\geq 2$, $0\leq \rankECE(f)\leq \lECE(f)$.
        \item [(b)] For $n\geq 4$, $\rankECE(f) = 0$ if and only if $f$ satisfies \eqref{eq:def_calibration}. 
        \item [(c)] For $n\geq 2$, $\rankECE(f)\geq \lbinECE(f) - \frac{4K}{n}$.
    \end{itemize}
\end{theorem}

\begin{remark}\label{remark:rankece_derandom}
    The construction of $\hatrankECE(f)$ in \eqref{eq:hat_rank_sigma} introduces additional randomness through the random variables $U_1,\ldots,U_n$, which are used to define the permutation $\pi$ in \eqref{eq:def_perm_sigma}. Consequently, a practitioner may obtain different values of $\hatrankECE(f)$ for the same set of observations. To bypass this additional randomness we make the observation,
    \begin{align*}
        \rankECE(f) = \E\left[\E\left[
    \frac{1}{n}\sum_{i=1}^{n-1}
    \left(Y_{\pi(i)}-Z_{\pi(i)}\right)
    \left(Y_{\pi(i+1)}-Z_{\pi(i+1)}\right)
    \mid \cD_n\right]\right]
    \end{align*}
    where $\cD_n = \{(Y_i, Z_i): 1\leq i\leq n\}$. With this observation we can define a derandomised version of $\hatrankECE(f)$ by averaging the expression from \eqref{eq:hat_rank_sigma} over the randomness of $U_1,\ldots, U_n$. In particular this is equivalent to definining $\hatrankECE(f)$ as follows:
    \begin{align}\label{eq:rankhat_average}
        \hatrankECE(f) 
        & = \E\left[
    \frac{1}{n}\sum_{i=1}^{n-1}
    \left(Y_{\pi(i)}-Z_{\pi(i)}\right)
    \left(Y_{\pi(i+1)}-Z_{\pi(i+1)}\right)
    \mid \cD_n\right]\nonumber\\
        & = \frac{1}{\left|\cP_n\right|}\sum_{\sigma\in \cP_n}\frac{1}{n}\sum_{i=1}^{n-1}
    \left(Y_{\sigma(i)}-Z_{\sigma(i)}\right)
    \left(Y_{\sigma(i+1)}-Z_{\sigma(i+1)}\right)
    \end{align}
    where $\cP_n = \{\sigma: Z_{\sigma(1)}\leq Z_{\sigma(2)}\leq\cdots\leq Z_{\sigma(n)}\}$ is the collection of all permutations which order $Z_1,\ldots, Z_n$ in a non-decreasing manner. From a computational perspective, the expression in \eqref{eq:rankhat_average} may appear computationally intensive, as it involves an average over $\cP_n$. However, it can be simplified considerably. Let $m_i = \sum_{\ell=1}^{n}\one\{Z_\ell = Z_i\}$ and $\cS_n = \{(i,j): Z_i<Z_j\text{ and}\not\exists \ell\in [n]\text{ such that }Z_i<Z_\ell<Z_j\}$. Then,
    \begin{align}\label{eq:rankhat_simple}
        \hatrankECE(f) = \frac{1}{n}\left[\sum_{i=1}^{n}\sum_{\substack{j\neq i\\Z_i = Z_j}}\frac{(Y_i - Z_i)(Y_j-Z_j)}{m_i} + \sum_{(i,j)\in \cS_n}\frac{(Y_i-Z_i)(Y_j-Z_j)}{m_im_j} \right]
    \end{align}

\end{remark}

\subsection{Proof of Theorem \ref{thm:rank_sigma_proper}}\label{appendix:proofof_rank_proper_sigma}

In this section, we present the proof of Theorem~\ref{thm:rank_sigma_proper}. The main idea is to use the inverse probability transform to construct uniform random variables in such a way that the permutation $\pi$ can be viewed as the permutation induced by ranking these uniform random variables. With this representation in hand, the remainder of the argument closely follows the proof of the corresponding result under the non-atomicity assumption.

\subsubsection{Proof of Part (a) and Part (b).}

Recall the residual function $r(Z) = \E[Y-Z\mid Z]$. Now, conditioning on $\bW_1,\ldots\bW_n$ we get,
\begin{align}\label{eq:rankece_sigma_Z}
    \rankECE(f) = \E\left[\frac{1}{n}\sum_{i=1}^{n-1}r\left(Z_{\pi(i)}\right)r\left(Z_{\pi(i+1)}\right)\right]
\end{align}
where $\pi$ is the permutation from \eqref{eq:def_perm_sigma}. Following the arguments similar to the upper bound from \eqref{eq:rank_r_equal} we can establish that $\rankECE(f)\leq \lECE(f)$. To prove the lower bound and establish $\rankECE$ as a proper measure of calibration we proceed with the following construction. Consider $q(z,\lambda) = \P(Z<z) + \lambda\P(Z = z)$ for $z,\lambda\in [0,1]$. Then, define $V_i = q(Z_i,U_i)$ for all $1\leq i\leq n$. By Proposition 2.1 of \cite{ruschendorf2009distributional}, the random variables \(V_1,\ldots,V_n\) are independent and each is distributed as \(\operatorname{Unif}[0,1]\).
 The following lemma then shows that $\pi$ is the permutation induced by ordering $V_1,\ldots, V_n$.

\begin{lemma}\label{lemma:V_sigma_rank}
    Consider $\bW_1,\ldots, \bW_n$ from \eqref{eq:def_perm_sigma} and let $V_1,\ldots, V_n$ be as defined above. Then,
    \begin{align*}
        V_{\pi(1)}<V_{\pi(2)}<\cdots<V_{\pi(n)} \text{ a.s. }
    \end{align*}
\end{lemma}
Next, let $F$ be the distribution function of $Z$. Then consider $F^{-1}$ to be the generalised inverse defined as,
\begin{align*}
    F^{-1}(v) = \inf\{z\in [0,1]: F(z)\geq v\} \text{ for all }v\in [0,1]. 
\end{align*}
Then once again applying Proposition 2.1 from \cite{ruschendorf2009distributional} we know $Z_i = F^{-1}(V_i)$ for all $1\leq i\leq n$. Using this identity on \eqref{eq:rankece_sigma_Z} we get,
\begin{align}
    \rankECE(f) 
    & = \E\left[\frac{1}{n}\sum_{i=1}^{n-1}r\left(F^{-1}(V_{\pi(i)})\right)r\left(F^{-1}(V_{\pi(i+1)})\right)\right]\nonumber\\
    & = \E\left[\frac{1}{n}\sum_{i=1}^{n-1}r\left(F^{-1}(V_{(i)})\right)r\left(F^{-1}(V_{(i+1)})\right)\right]\label{eq:rank_V}
\end{align}
where $V_{(1)}\leq V_{(2)}\leq\cdots\leq V_{(n)}$ are the order statistics of $V_1,\ldots,V_n$ and the last equality follows by applying Lemma \ref{lemma:V_sigma_rank}. The proof of lower bound is now completed by following the arguments for lower bound from Appendix \ref{sec:proofof_prop_lower_bound}. The proof of $\rankECE(f) = 0$ if and only if $f$ satisfies \eqref{eq:def_calibration} now follows by an application of Lemma \ref{lemma:rank_uniform_0}.

\subsubsection{Proof of Part (c).}
Given samples $Z_1,\ldots, Z_n$ let $n_k$ be the number of samples in bin $\cB_k$, denote $Z_{1,k},\ldots,Z_{n_k,k}$ to be the samples in bin $\cB_k$ and let $\bW_{i,k} = (Z_{i,k}, U_{i,k})$ be the pair associated with $Z_{i,k}$ for all $1\leq i\leq k$. Moreover let $\pi_k:[n_k]\ra[n_k]$ denote the permutation induced by ranking the pairs $\bW_{1,k},\ldots,\bW_{n_k,k}$ lexicographically (see \eqref{eq:def_perm_sigma}). Then by the equivalence from \eqref{eq:rankece_sigma_Z},
\begin{align*}
    \rankECE(f) 
    & = \frac{1}{n}\sum_{k=1}^{K}\E\left[\sum_{i=1}^{n_k-1}r\left(Z_{\pi_k(i),k}\right)r\left(Z_{\pi_{k}(i+1),k}\right)\one\{n_k\geq 2\}\right]\\
    & + \frac{1}{n}\sum_{k=1}^{K-1}\E\left[r\left(Z_{\pi_k(n_k),k}\right)r\left(Z_{\pi_{j_k}(1), j_k}\right)\one\{n_k\geq 1\}\one\{\max_{j>k}n_j \geq 1\}\right]
\end{align*}
where $j_k = \min\{\arg\min\{j: k<j\leq K, n_j\geq 1\}, K\}$ (and the indicator $\one\{\max_{j>k}n_j \geq 1\}$ is included to ensure that $j_k$ exists---that is, the last sum only includes bins $k$ whose $n_k$th data point $Z_{\pi_k(n_k)}$ is not the last among all $n$ data points). Since $|r|\leq 1$, then
\begin{align}\label{eq:lb_rank_gen}
    \rankECE(f)\geq \frac{1}{n}\sum_{k=1}^{K}\E\left[\sum_{i=1}^{n_k-1}r\left(Z_{\pi_k(i),k}\right)r\left(Z_{\pi_{k}(i+1),k}\right)\one\{n_k\geq 2\}\right] - \frac{K-1}{n}.
\end{align}
Similar to \eqref{eq:r_prod_expansion}, for $1\leq k\leq K$,
\begin{align}\label{eq:r_prod_expansion_gen}
    r\left(Z_{\pi_k(i),k}\right)r\left(Z_{\pi_{k}(i+1),k}\right) = r_k^\circ\left(Z_{\pi_k(i),k}\right)r_k^\circ\left(Z_{\pi_{k}(i+1),k}\right) + \mu_k\left(r_k\left(Z_{\pi_k(i),k}\right)+r_k\left(Z_{\pi_{k}(i+1),k}\right)\right) - \mu_k^2
\end{align}
where $r_k^\circ(\cdot) = r(\cdot) - \mu_k$ with $\mu_k = \E[r(Z)|Z\in \cB_k]$. To further lower bound $\rankECE(f)$ note that conditional on $n_k = m$, $\bW_{1,k},\ldots, \bW_{1,m}$ have the same joint distribution as $\bar{\bW}_{1,k},\ldots,\bar{\bW}_{1,m}$ where $\bar{\bW}_{i,k} = (\bar Z_{i,k}, \bar U_{i,k})$ with $\bar Z_{1,k},\ldots, \bar Z_{m,k}$ generated independently from $P_{Z|Z\in \cB_k}$ and the collection $\bar U_{1,k},\ldots, \bar U_{m,k}$ is generated independently from $\textnormal{Unif}[0,1]$. Let $\bar \pi_k$ be the permutation induced by lexicographic ordering of $\bar{\bW}_{1,k},\ldots, \bar\bW_{m,k}$. Then,
\begin{align*}
    \E\left[r_k^\circ\left(Z_{\pi_k(i),k}\right)r_k^\circ\left(Z_{\pi_{k}(i+1),k}\right)\mid n_k = m\right] = \E\left[r_k^\circ\left(\bar Z_{\bar \pi_k(i),k}\right)r_k^\circ\left(\bar Z_{\bar \pi_{k}(i+1),k}\right)\right].
\end{align*}
Now repeating the argument for proof of equivalence in \eqref{eq:rank_V} with the distribution $P_{Z\mid Z\in \cB_k}$ we can show that there exists some function $s$ such that,
\begin{align*}
    \E\left[r_k^\circ\left(\bar Z_{\bar \pi_k(i),k}\right)r_k^\circ\left(\bar Z_{\bar \pi_{k}(i+1),k}\right)\right] = \E\left[s\left(\bar V_{(i),k}\right)s\left(\bar V_{(i+1),k}\right)\right]
\end{align*}
where $\bar V_{(1),k}\leq \cdots\leq \bar V_{(m),k}$ are order statistics of $\bar V_{1,k},\ldots, \bar V_{m,k}\sim\textnormal{Unif}[0,1]$. With the boundary convention $\bar V_{(0), k} = 0$ and $\bar V_{(m+1), k} = 1$,
\begin{align*}
    \E\left[s\left(\bar V_{(i),k}\right)s\left(\bar V_{(i+1),k}\right)\right] = \E\left[\E\left[s\left(\bar V_{(i),k}\right)s\left(\bar V_{(i+1),k}\right)\mid \bar V_{(i-1),k}, \bar V_{(i+2), k}\right]\right].
\end{align*}
Applying Lemma 1 from \cite{gupta2021distribution},
\begin{align}\label{eq:EsV_nng}
    \E\left[s\left(\bar V_{(i),k}\right)s\left(\bar V_{(i+1),k}\right)\mid \bar V_{(i-1),k}, \bar V_{(i+2), k}\right] = \E\left[s\left(V_{(1),k}\right)s\left(V_{(2),k}\right)\right] = \E\left[\left(V_{1,k}\right)\right]^2\geq 0.
\end{align}
where $V_{(1),k}\leq V_{(2),k}$ are the order statistics of $V_{1,k}, V_{2,k}\sim\textnormal{Unif}[\bar V_{(i-1),k}, \bar V_{(i+2), k}]$. Using the non-negativity from \eqref{eq:EsV_nng}, the expansion from \eqref{eq:r_prod_expansion_gen} and the lower bound from \eqref{eq:lb_rank_gen} we conclude,
\begin{align*}
    \rankECE(f)\geq \frac{1}{n}\sum_{k=1}^{K}\E\left[\sum_{i=1}^{n_k-1}\left(\mu_k\left(r_k\left(Z_{\pi_k(i),k}\right)+r_k\left(Z_{\pi_{k}(i+1),k}\right)\right) - \mu_k^2\right)\one\{n_k\geq 2\}\right]-\frac{K-1}{n}.
\end{align*}
The proof is now completed by arguments analogous to the proof of Theorem \ref{thm:rank_bin_comparison} in Appendix \ref{sec:proofof_thm_rank_bin_comparision}.

\subsubsection{Proof of Lemma \ref{lemma:V_sigma_rank}}
To prove Lemma \ref{lemma:V_sigma_rank} it is enough to show that for $1\leq i\leq n-1$, $V_{\pi(i)}<V_{\pi(i+1)}$ almost surely. To that end, fix $i\in [n-1]$. Then we have the following two cases:\\
\textbf{Case I:} $Z_{\pi(i)} = Z_{\pi(i+1)}$. In this case by definition from \eqref{eq:def_perm_sigma} we know that $U_{\pi(i)}<U_{\pi(i+1)}$. Hence recalling definition of $V$ it immediately follows that $V_{\pi(i)}<V_{\pi(i+1)}$.\\
\textbf{Case II:} $Z_{\pi(i)}<Z_{\pi(i+1)}$. Let $F(z-) = \P(Z<z)$. Then by definition $V_{\pi(i)}\leq F(Z_{\pi(i)})\leq F(Z_{\pi(i+1)}-)\leq V_{\pi(i+1)}$. To complete the proof, recall that $V_1,\ldots, V_n$ are independent sample from $\textnormal{Unif}[0,1]$ and hence $V_{\pi(i)}<V_{\pi(i+1)}$.

\subsection{Justification of Remark \ref{remark:rankece_derandom}}

In this section we justify the equivalence from \eqref{eq:rankhat_average} and the subsequent simplification from \eqref{eq:rankhat_simple}.\\

To justify the equivalence from \eqref{eq:rankhat_average} it suffices to show that conditional on $\cD_n$ the permutation $\pi$ from \eqref{eq:def_perm_sigma} is distributed uniformly on $\cP_n$. To that end consider $\sigma\in \cP_n$. Define,
\begin{align*}
    \mathcal J = \{j\in [n-1]:Z_{\sigma(j)} = Z_{\sigma(j+1)}\}.
\end{align*}
Note that by definition of $\cP_n$, the set $\mathcal J$ is independent of the choice of $\sigma$. By definition of $\pi$ from \eqref{eq:def_perm_sigma} and the set $\cP_n$,
\begin{align*}
    \{\pi = \sigma\} = \{U_{\sigma(j)}<U_{\sigma(j+1)}\text{ for all }j\in \mathcal J\}.
\end{align*}
Then,
\begin{align}\label{eq:pi_equal_sigma}
    \P\left(\pi = \sigma\mid\cD_n\right) 
    & = \P\left(U_{\sigma(j)}<U_{\sigma(j+1)}\text{ for all }j\in \mathcal J\mid\cD_n\right)\nonumber\\
    & = \P\left(U_j<U_{j+1}\text{ for all }j\in J\mid\cD_n\right)
\end{align}
where the last equality follows by observing that conditional on $\cD_n$ the permutation $\sigma$ and the set $J$ is fixed and $(U_{\sigma(1)},\ldots, U_{\sigma(n)})\mid \cD_n \overset{d}{=} (U_1,\ldots, U_n)$. The proof is now completed since the right hand side of \eqref{eq:pi_equal_sigma} is independent of the choice of $\sigma$.\\

To justify the simplification in \eqref{eq:rankhat_simple}, let $R_i = Y_i-Z_i$ for all $i\in [n]$ and note that
\begin{align}\label{eq:hatrank_simple_2}
    \hatrankECE(f) 
    & = \frac{1}{n}\sum_{i=1}^{n}\frac{1}{|\cP_n|}\sum_{\sigma\in \cP_n}R_{\sigma(i)}R_{\sigma(i+1)}\nonumber\\
    & = \frac{1}{n}\sum_{i=1}^{n}\sum_{j\neq i}\frac{R_iR_j}{|\cP_n|}\sum_{\sigma\in \cP_n}\sum_{k=1}^{n}\one\{\sigma(k) = i, \sigma(k+1) = j\}.
\end{align}
To further simplify this expression, fix $i\neq j\in [n]$. The corresponding term on the right-hand side can be non-zero only if there exist $\sigma\in \cP_n$ and $k\in [n]$ such that $\sigma(k) = i$ and $\sigma(k+1) = j$. There are two cases to consider. First, if $Z_i = Z_j$, then, by counting the number of permutations in which $Z_i$ is ranked immediately before $Z_j$, we obtain
\begin{align}\label{eq:i_equal_j}
    \frac{1}{|\cP_n|}\sum_{\sigma\in \cP_n}\sum_{k=1}^{n}\one\{\sigma(k) = i, \sigma(k+1) = j\} = \frac{1}{m_i}.
\end{align}
On the other hand, if $(i,j)\in \cS_n$, then, by counting the number of permutations in which $Z_i$ is ranked immediately before $Z_j$, we obtain
\begin{align}\label{eq:i_j_cS}
    \frac{1}{|\cP_n|}\sum_{\sigma\in \cP_n}\sum_{k=1}^{n}\one\{\sigma(k) = i, \sigma(k+1) = j\} = \frac{1}{m_im_j}.
\end{align}
The proof is completed by substituting the values from \eqref{eq:i_equal_j} and \eqref{eq:i_j_cS} into \eqref{eq:hatrank_simple_2}.

\section{Additional Details on Sentiment Models}\label{sec:sentiment_details}

In Section \ref{sec:sentiment_main} and Appendix \ref{appendix:sentiment_experiments} we evaluate calibration of four publicly available pre-trained sentiment classifiers \texttt{DistilBERT\_SST2}, \texttt{BERT\_SST2}, \texttt{RoBERTa\_Twitter\_Sentiment}, and \texttt{\seqsplit{BERT\_Multilingual\_Stars}} available at \texttt{Hugging Face}.\footnote{\url{https://huggingface.co/models}} In this section we provide more details about the chosen models.

All four models are variants of the language representation model \texttt{BERT} introduced in \cite{devlin2019bert}. The \texttt{BERT} model is pre-trained on BookCorpus \citep{zhu2015aligning} and English Wikipedia. This model is a general purpose task-agnostic encoder which can be fine tuned on a small labeled datasets for downstream classification tasks. Each of the four models we use applies this recipe. In all cases, the fine-tuned model maps an input text to
a vector of class logits, converted to class probabilities via the softmax
transformation; these probabilities are the predicted scores $Z$ whose
calibration we assess. In the following we provide additional details about the individual models.

\paragraph{\texttt{DistilBERT\_SST2} :} This model is a distilled version (trained through knowledge distillation) of \texttt{BERT} \citep{sanh2019distilbert} fine-tuned on the Stanford Sentiment Treebank binary sentiment task (SST-2) \citep{socher2013recursive}, a corpus of movie-review sentences labeled positive or negative.\footnote{\url{https://huggingface.co/distilbert/distilbert-base-uncased-finetuned-sst-2-english}} The model outputs probabilities for the positive and negative classes, and we take $Z$ to be the softmax probability of the positive class.

\paragraph{\texttt{BERT\_SST2} :} This model is the original \texttt{BERT} model fine tuned on on the same SST-2 data
\citep{socher2013recursive} as part of the TextAttack framework \citep{morris2020textattack}.\footnote{\url{https://huggingface.co/textattack/bert-base-uncased-SST-2}} As before, the model outputs probabilities for the positive and negative classes, and we take $Z$ to be the softmax probability of the positive class.

\paragraph{\texttt{RoBERTa\_Twitter\_Sentiment} :} This model is based on the \texttt{RobERTa} \citep{liu2019roberta} model, which preserves the \texttt{BERT} architecture, but pretrains longer with more data. This model is pretrained on $\sim124$M tweets from January 2018 to December 2021 \citep{loureiro2022timelms}, and finetuned for three-class
(negative, neutral, positive) sentiment analysis with the TweetEval benchmark.\citep{barbieri2020tweeteval} \footnote{\url{https://huggingface.co/cardiffnlp/twitter-roberta-base-sentiment-latest}} We take $Z$ to be the softmax probability of the positive class.

\paragraph{\texttt{{BERT\_Multilingual\_Stars}} :} This is based on multilingual \texttt{BERT}, which is pretrained on the Wikipedia text from over a hundred languages with a shared vocabulary \citep{devlin2019bert, pires2019multilingual}, and fine-tuned to predict product-review star ratings (one to five) across six languages (English, Dutch, German, French, Spanish, Italian).\footnote{\url{https://huggingface.co/nlptown/bert-base-multilingual-uncased-sentiment}} Unlike the other three models, the fine-tuning task is ordinal rather than binary; we take $Z$ to be the sum of the softmax probabilities of the upper-half (four- and five-star) classes.

\end{document}